\documentclass[12pt]{article}

\usepackage{amssymb,amsmath,amsthm}
\usepackage[colorlinks,citecolor=blue,urlcolor=blue]{hyperref}

\usepackage{cases}
\usepackage{graphicx}
\usepackage{algorithm}
\usepackage{algpseudocode}
\usepackage[round]{natbib}
\usepackage{color}

\usepackage{subfigure}
\usepackage[toc,page,title,titletoc,header]{appendix}
\usepackage{multirow}
\usepackage{booktabs}
\usepackage{setspace}
\usepackage{url}

\usepackage[overload]{empheq}
\usepackage{cleveref}

\newtheorem{lemma}{Lemma}

\newtheorem{defn}{Definition}

\newtheorem{thm}{Theorem}
\newtheorem{remark}{Remark}[section]
\newtheorem{manualtheoreminner}{Theorem}
\newenvironment{manualtheorem}[1]{%
  \renewcommand\themanualtheoreminner{#1}%
  \manualtheoreminner
}{\endmanualtheoreminner}

\newcommand{\bm}{\boldsymbol}
\newcommand{\bsb}{\boldsymbol}
\newcommand{\rd}{\,\mathrm{d}}

\newcommand{\bsbX}{{\boldsymbol{X}}}
\newcommand{\bsbx}{{\boldsymbol{x}}}
\newcommand{\bsby}{{\boldsymbol{y}}}

\newcommand{\bsbb}{{\boldsymbol{\beta}}}
\newcommand{\bsbg}{{\boldsymbol{\gamma}}}
\newcommand{\bsbGamma}{{\boldsymbol{\Gamma}}}

\newcommand{\bsbH}{{\boldsymbol{H}}}

\newcommand{\bsbI}{{\boldsymbol{I}}}

\newcommand{\bsbxi}{{\boldsymbol{\xi}}}
\newcommand{\bsbD}{{\boldsymbol{D}}}

\newcommand{\bsbU}{{\boldsymbol{U}}}
\newcommand{\bsbV}{{\boldsymbol{V}}}

\newcommand{\bsba}{{\boldsymbol{\alpha}}}
\newcommand{\bsbA}{{\boldsymbol{A}}}
\newcommand{\bsbe}{{\boldsymbol{\eta}}}

\newcommand{\bsbeps}{{\boldsymbol{\epsilon}}}

\newcommand{\bsbr}{{\boldsymbol{r}}}

\newcommand{\bsbdelta}{{\boldsymbol{\delta}}}
\newcommand{\bsbzeta}{{\boldsymbol{\zeta}}}
\newcommand{\bsbeta}{{\boldsymbol{\eta}}}
\newcommand{\bsbh}{{\boldsymbol{h}}}

\newcommand{\bsbDelta}{{\boldsymbol{\Delta}}}

\newcommand{\bsbd}{\boldsymbol{d}}

\newcommand{\Proj}{{{\mathcal P}}}
\newcommand{\EP}{\,\mathbb{P}}
\newcommand{\EE}{\,\mathbb{E}}

\newcommand{\lmtt}{\fontfamily{lmtt}\selectfont}
\newcommand{\iid}{\stackrel{\small iid}{\sim}}

\newcommand{\breg}{{\mathbf{\Delta}}}
\newcommand{\Breg}{{\mathbf{D}}}

\usepackage[normalem]{ulem}

\usepackage[top=2.7cm, bottom=3cm, left=2.7cm, right=2.7cm]{geometry}
\usepackage{setspace}

\begin{document}

\title{  Gaining Outlier Resistance with Progressive Quantiles: Fast Algorithms and Theoretical Studies  }
\author{Yiyuan She, Zhifeng Wang, Jiahui Shen \\ Department of Statistics, Florida State University}
\date{}

\maketitle

\begin{abstract}
  Outliers widely occur in big-data   applications and  may severely affect statistical estimation and inference. In this paper,   a framework of outlier-resistant estimation is introduced to robustify an arbitrarily  given   loss function. It has a close connection to the method of  trimming  and   includes explicit outlyingness parameters for all samples, which in turn   facilitates computation, theory, and parameter tuning.  To tackle   the issues of  nonconvexity and nonsmoothness, we  develop scalable   algorithms with implementation ease and guaranteed fast convergence. In particular,   a new technique is proposed  to alleviate the requirement on the starting point such that on regular  datasets, the number of  data resamplings can be substantially reduced. Based on  combined statistical and computational treatments, we are able to perform   nonasymptotic   analysis beyond  M-estimation.  The obtained resistant estimators, though not necessarily globally or even locally optimal,  enjoy minimax  rate optimality in both low dimensions and high dimensions. Experiments in regression, classification, and neural networks show excellent performance of the proposed methodology at the occurrence of gross outliers.
\end{abstract}

\section{Introduction} \label{sec:intro}

Outliers are bound to occur in real-world big data   and   may severely affect statistical analysis. Many commonly used methods, including the  lasso \citep{tibshirani1996regression},    break down   in the presence of  a single bad observation.   We study the problem in a general setup.  Given   an arbitrary   loss function   $l $,  a response vector $\bsby \in \mathbb R^n$ and a design matrix $\bsbX \in \mathbb R^{n \times p}$, the coefficient vector $\bsbb \in \mathbb R^p$ can be estimated by solving the following optimization problem
\begin{equation} \label{eq:old}
\min_{\bsbb \in \mathbb R^p}  l(\bsbe; \bsby) + \sum_{j=1}^p P(\beta_j) \mbox{ s.t. } \bsbe = \bsbX \bsbb,
\end{equation}
where   $l$ measures the discrepancy between the observed  response and the prediction, and   $P$ represents a regularizer that can also be  a  constraint.  Neither $l$ nor $P$ is necessarily convex. The loss can be, say, a deviance function  corresponding to a   generalized linear model  (GLM), a robust loss like Huber's loss,
 or the well-known hinge loss  or   exponential loss for the purpose of classification. The loss function  need    not     be associated with  any probability distribution.

The goal of this work is to robustify \eqref{eq:old}  for any given $l$ at the occurrence of a  small number of \textit{gross  outliers}---we stress that the desired  \textbf{outlier resistance} should not only accommodate    samples with mild deviations from the model assumption, but also   handle extreme anomalies with severe  outlyingness and leverage.     Some   frequently used robust losses---Huber's loss and the   $\ell_1$-norm loss, in particular---
are not   resistant to gross outliers, and hence one may   want    a  reinforced   robustification in the setup of  \eqref{eq:old}. The nature of the task eventually leads to   an     $\ell_0$-type regularization proposal.

The research of robust statistics undergoes a long history and leads  to fruitful outcomes. A standard means to gain robustness is to modify the loss, or equivalently, reweight the samples. We illustrate the idea  with   robust regression. Instead of minimizing the quadratic loss $\|\bsby - \bsbX \bsbb\|_2^2$,  M-estimation   \citep{huber2009robust,hampel2011robust,maronna2006robust} solves $\min_\bsbb\sum_i \rho(y_i - \bsbx_i^T\bsbb)$ or   $\bsbX^T\psi(\bsbX\bsbb - \bsby) = \bsb 0$, where $\rho$ is a robust loss and $\psi = \rho'$.  Let $w(t) = \psi(t) / t$ and $\bsbr = \bsbX\bsbb - \bsby$. The score equation can   be  written as $ \bsbX^T \{w(\bsbr) \circ (\bsbX\bsbb - \bsby)\} = \bsb0 $, where    the multiplicative weights  $w(r_i)$ are applied componentwise.  The design of a robust loss thus  amounts to that of a weight function.
See \cite{loh2017statistical} and  \cite{avella2017influence}, for example, for some   theoretical properties. We call the above scheme the  \textit{multiplicative}  robustification, to  contrast with the \textit{additive}   fashion to be introduced   in Section \ref{sec:regularized}.   An iterative algorithm called  iterative reweighted least squares (IRLS) is popularly used in computation. The most critical and expensive step however lies in the   {preliminary} resistant fit with  high-breakdown \citep{rousseeuw1984robust,rousseeuw1985multivariate}, also the focus of our paper. Some challenges of resistant estimation in methodology, theory, and computation are summarized as follows.

First, when $ l(\bsbe; \bsby) $ has a   more complicated expression than that    in terms of    $r_i = y_i - \bsb{x}_i^T \bsb{\beta}$ only, how to reshape it  to guard against gross outliers          does not seem straightforward.
According to  \cite{bianco1996robust}, simply applying    $\rho$ on the deviance residuals of a GLM \citep{pregibon1982resistant} does {not} ensure consistency   and  an extra  bias correction term must be added into the criterion.   \cite{croux2003implementing} limited $\rho$ to a certain class to guarantee the existence of a finite solution, and included   an additional weighting step to           get a bounded   influence function.
We refer to   \cite{AMR18} 
    for a robust quasi-likelihood form by use of  both a $\psi$ function and a weight function. To us, another important line of work  on the trimmed likelihood estimator ({TLE}) \citep{hadi1997maximum,  vandev1998regression, Kurnaz2018}, as an extension of {least trimmed squares (LTS) \citep{rousseeuw1985multivariate}}, is most motivating. LTS and TLE  give rise to   our new   regularized   form of resistant estimation that is  more amenable to both optimization  and analysis.

The theoretical challenges faced by today's robust statisticians   are perhaps even  more pressing. For example,  even under $n>p$ (let alone in   high-dimensional  settings), (a) how trimming inflates  the risk,   (b) how small the universal   minimax lower bound could be, and (c) what makes a theoretically sound     model selection   criterion,  are all surprisingly   unknown in the literature.  We emphasize the \textit{finite-sample} nature of desired statistical theory (e.g., Theorems \ref{thm:fixed-point}--\ref{thm:minimax}), because in real-life data applications, it is often hard to know how large the sample size should be relative to the number of unknowns to apply asymptotics. Breakdown   point provides a useful  index of the robustness of  a given method, but as pointed out by \cite{huber2009robust}, such     worst-case studies are not probabilistic and   may be   too conservative on a particular dataset.
In addition to robust analysis,  parameter tuning based on asymptotics or breakdown point     suffers the same theoretical difficulties.

 The last  big challenge lies in computation. Resistant fits are much more  costly than   IRLS or $\Theta$-IPOD \citep{she2011outlier}  which takes them as   initial points. Due to the inherent {nonconvexity},  subset sampling is heavily used   in LTS-like algorithms to generate multiple starting values    \citep{rousseeuw1999computing}. Unfortunately,   the number of required samplings grows exponentially in $p$  
 and so   these procedures already   become inaccurate and/or prohibitive    for  moderate $p$. Developing more efficient outlier-resistant algorithms and     mitigating the burden of subset sampling are at the core of  modern-day robust statistics. Moreover,  the fact that the computationally obtained solutions are   not necessarily globally optimal   poses   nontrivial  challenges in theoretical analysis.\\

This paper    attempts to   address  some aforementioned
obstacles to gain outlier resistance,  with the hope  to advance the practice of   robust statistics to more   sophisticated learning tasks.  Our main contributions are threefold.
a) A  general resistant learning framework  is introduced to robustify an arbitrarily  given   loss, with an ``$\ell_0+\ell_2$'' form of regularization to deal with gross outliers in possibly high dimensions. It has a close connection to the method of  trimming  but  includes explicit outlyingness parameters, which in turn   facilitates computation, theory, and parameter tuning.  b)   Although the associated problem is highly nonconvex and nonsmooth, extremely scalable   algorithms with implementation ease can be developed with  guaranteed   convergence. In particular, the statistical error of the sequence of iterates can   converge  \textit{geometrically} fast. A progressive optimization technique is proposed  to relax the regularity condition and alleviate the requirement on the starting point so that on regular  datasets the number of  data resamplings can be substantially reduced.
c) Combined statistical and computational treatments make it possible  to perform   nonasymptotic  analysis beyond the M-estimation  setting, and the obtained ``A-estimators'', though not necessarily globally or even locally optimal,  enjoy minimax  rate optimality in both low dimensions and high dimensions.

  The rest of the paper is organized as follows.  Section \ref{sec:regularized}  introduces our  regularized resistant learning  framework and compares it to trimming.   Section \ref{sec:algorithm} deals with the computational issues by     progressive iterative quantile thresholding. Section \ref{sec:theory}  provides nonasymptotic theoretical support for the obtained estimators under proper   regularity conditions.      Section  \ref{sec:tuning} reveals some universal minimax lower bounds and investigates the issue of parameter tuning.   Section \ref{sec:exp} performs extensive simulation studies and real world data analysis. The proof details and more simulation results  are given   in the appendix.

\paragraph{Notation.}
We use   $\bsbX = [\bsbx_1, \bsbx_2, \cdots, \bsbx_n]^T$
 to denote the design matrix and $\bsbx_i$  the $i$th sample.
Throughout the paper, we use $ l (\bsbe ; \bsby )$ to denote a differentiable loss defined on the   systematic component  $\bsbe$, and assume that $l$ is   bounded from below and   $\inf_{\bsbe} l(\bsbe; \bsby) \ge 0   $ without   loss of generality. Sometimes, such as when     building  the connection to trimming,    we assume that the overall  loss  takes a sample-additive form,  $l(\bsbe; \bsby)=\sum_{i=1}^n l_0(\eta_i; y_i)$.  $l(\bsbe; \bsby)$ is said to  be $L$-strongly smooth or  have    an $L$-Lipschitz continuous gradient if \begin{equation} \label{eq:derivative_lipschitz_condition}
\| \nabla l (\bsbe_1; \bsby) - \nabla l (\bsbe_2; y)\|_2 \leq  L\|\bsbe_1 - \bsbe_2\|_2, \forall \bsbe_1, \bsbe_2
\end{equation}
for some   $L > 0$.
For ease of presentation, we  frequently use the concatenated notation  $$\bar{\bm X} = [\bm X,\bm I] , \quad \bar{\bm\beta}=[\bm\beta^T,\bm\gamma^T]^T,$$   and so $\bm X\bm\beta+\bm\gamma = \bar{\bm X}\bar{\bm\beta}$. For any matrix $\bsbA$, $\| \bsbA\|_2$ is its spectral norm. Given any $\bsbb\in \mathbb R^p$, $\| \bsbb\|_0 = \sum_{j=1}^p 1_{\beta_j \ne 0}$. The hat matrix   $\bm H$ of $\bsbX$ is $\bsbX (\bsbX^T \bsbX)^+ \bsbX^T$ with $^+$ denoting the Moore-Penrose inverse.
Given any two vectors $\bsba, \bsbb$ of the same size, the inner product $\langle \bsba, \bsbb\rangle$ is given by  $\sum \alpha_j \beta_j$.

Given $\bsbX\in \mathbb R^{n\times p}$ and $s\le p$,   we introduce  $0\le m_{\bsbX} (s),M_{\bsbX} (s)$ associated with the \textit{restricted isometry property} \citep{Candes2005} by
$
m_{\bsbX}(s) \|\bsbb\|_2^2 \le \|\bsbX \bsbb\|_2^2 \le  M_{\bsbX}(s) \|\bsbb\|_2^2$, $ \forall   \bsbb: \| \bsbb\|_0 \le s
$. In particular, $M_{\bar \bsbX}(s, o)$ ($s\le p, o\le n$) obeying $$
\|\bar\bsbX \bar\bsbb \|_2^2 \le  M_{\bar \bsbX}(s, o) \|\bar\bsbb\|_2^2,  \  \forall \bar \bsbb: \| \bsbb\|_0 \le s, \| \bsbg \|_0 \le o
$$ will play an active role in our algorithm design and analysis. Obviously, $M_{  \bsbX}(s)\le \| \bsbX\|_2 ^2$ and  $M_{\bar \bsbX}(s, o)\le \| \bar \bsbX\|_2 ^2$. But   since we are primarily  interested  in  $s, o\ll n$,   $\bsbX \bsbb$ or $\bar \bsbX \bar \bsbb$ for an $s$-restricted $\bsbb$ or an $(s,o)$-restricted $\bar \bsbb$  involves a much thinner design, and so   $M_{  \bsbX}(s ),  M_{\bar \bsbX}(s, o)$ can be way smaller.

 We use $C$, $c$ and $L$ to denote universal constants which  are not necessarily the same at each occurrence and    $\lesssim$  denotes  an inequality that holds up to a multiplicative numerical constant. Finally,   $0\log 0 = 0$,  $a\vee b = \max\{a, b\}$ and $a\wedge b = \min\{a, b\}$ are adopted.

\section{$L_0$-regularization vs. Trimming} \label{sec:regularized}
We begin with the   robust estimation problem \eqref{eq:old} without   regularization. The customizable   loss determines the   performance metric, so  we prefer \textit{not} to alter its form  in the process of robustification. This is  possible by use of an additive regularized estimation.

Concretely,  re-define $\bsbe$ as $ \bsbX\bsbb + \bsbg$,  where $\gamma_i$ characterizes the outlyingness of the $i$th sample and we allow it to be    large.
Taking advantage of the sparsity in $\bsbg$,  as outliers are never the norm, we propose an  $\ell_0$-constrained,  $\ell_2$-penalized    outlier-resistant estimation
\begin{equation} \label{eq:criterion_l2}
\min_{\bsbb , \bsbg   } l(\bsbe;\bsby) + \frac{\nu}{2}\|\bsbg\|_2^2 \ \,  \mbox{ s.t. } \bsbe = \bsbX \bsbb + \bsbg, \mbox{ } \|\bsbg\|_0 \leq q,
\end{equation}
where $q$ satisfies $q\le n/2$  and  is assumed to be an integer throughout the paper. Here, the number of unknowns,  $n+p$, is larger than the number of observations, and the ``$\ell_0+\ell_2$'' regularization plays an effective  role in dealing with the high-dimensional challenge.

Unlike   the popular $\ell_1$-norm penalty, the $\ell_0$ `norm'  used in the constraint   is indifferent to the magnitude of $\gamma_i$   and does not result in any undesired  bias.     An alternative idea is to  utilize a   non-convex sparse penalty---in the regression setting, this       amounts to  M-estimation,  the diverse  choices of the penalty lead to different robust losses   \citep{she2011outlier}, and the conclusion extends to regularized $\bsbb$-estimation in possibly high dimensions \cite[Lemma 2]{She2017RRRR}. But   \eqref{eq:criterion_l2} has some distinct advantages:    the constraint  directly controls  the maximum number of outliers and     $\| \cdot \|_0 $   is   arguably the ideal function to enforce sparsity regardless of the severity of   aberrant samples. In practice,   specifying the value of       $q$, rather than a penalty parameter $\lambda$,   is   easier and more convenient.

 In addition to the sparsity-promoting regularization, we include the   $\ell_2$-shrinkage     to enhance  numerical stability and compensate  for collinearity   in the presence of clustered outliers.   We will see  in Section \ref{subsec:motivating_LTS} that  with $\nu$ introduced, the samples with minor deviations   can still contribute to the  model fitting. Empirically, $\nu$ is not a sensitive parameter, and  we often set  $\nu$  to a mild value (say 1e-4).   As long as $\nu>0$, the   $\bsbg$-estimate is finite.
But if we force  $\nu = 0$, \eqref{eq:criterion_l2}  has an interesting and important connection   to the method of trimming \citep{rousseeuw1985multivariate, hadi1997maximum,  vandev1998regression}.
\begin{thm} \label{th:criterion_equivalence}
Assume    $l(\bsbe; \bsby)$ takes a sample-additive form $\sum_{i=1}^n l_0(\eta_i; y_i)$ and
$\inf_{\eta}l_0(\eta;y) = 0 \mbox{ for all } y \in \mathcal Y \subset\mathbb R$.

i) Given any minimizer    $(\hat \bsbb, \hat \bsbg)$    of the $\ell_0$-\emph{constrained} \emph{joint} problem  $\min_{\bsbb , \bsbg   } l( \bsbX \bsbb + \bsbg;\bsby)$  s.t. $\|\bsbg\|_0\le q$,  $\hat \bsbb$ is also a globally optimal solution to the \emph{trimmed} problem on $\bsbb$:
$
\min_{\bsbb  } \sum_{i=1}^{n-q} {l}_0^{(i)}(\bsbb)
$ 
with ${l}_0^{(1)} (\bsbb)\leq  \cdots \leq {l}_0^{(n)}(\bsbb)$  the order statistics of $l_0(\bsbx_i^T\bsbb; y_i)$;
conversely, given any $\hat \bsbb$ that minimizes $\sum_{i=1}^{n-q} {l}_0^{(i)}(\bsbb)$,   there always exists a  $\hat \bsbg \in \bar{\mathbb R}^n$ with  $\bar {\mathbb R} =[-\infty, \infty]$ so that $(\hat \bsbb, \hat \bsbg)$ is a solution to the $\ell_0$-constrained joint problem.

ii)
The same connection holds between the $\ell_0$-\emph{penalized} \emph{joint} problem $\min_{\bsbb , \bsbg   } l( \bsbX \bsbb + \bsbg;\bsby) + \tau\|\bsbg\|_0$ and  the \emph{winsorized} problem $\min_{\bsbb}\sum  \tau \wedge l_0(\bsbx_i^T\bsbb; y_i)$ for any cutoff $\tau\ge 0$.

iii)
Finally, any minimizer  $\hat \bsbb$ of the winsorized optimization problem is  a solution to  the  trimmed problem with $q =| \{i:  l_0(\bsbx_i^T\hat \bsbb; y_i)>\tau\}|$, and  any       $ (\hat \bsbb, \hat \bsbg) $ that minimizes the $\ell_0$-penalized joint criterion is also an optimal solution to the $\ell_0$-constrained joint problem with  $q = \| \hat \bsbg\|_0$.
\end{thm}

So in a sense,   imposing the $\ell_0$ regularization  on $\bsbg$ amounts to using a trimmed version or setting a cutoff of  the original loss on $\bsbb$, both of which  can effectively bound   the influence of outliers and necessarily result in nonconvexity. The  joint $(\bsbb, \bsbg)$ formulation under $\ell_0$-penalization   corresponds to winsorizing the  loss (which unsurprisingly gives a  redescending $\psi$  in the setup of robust regression), and  the $\ell_0$-constrained form has a  similar power but utilizes a data-adaptive cutoff. It is noteworthy that our additive regularized framework makes it possible to   utilize any loss (which   can be convex) for resistant estimation.

\begin{remark}
Because of the connections, some  breakdown-point and influence-function results of the estimators defined by  \eqref{eq:criterion_l2} can be directly obtained from say  \cite{alfons2013sparse}  and \cite{ollerer2015influence}. For example, Proposition 1 of \cite{alfons2013sparse} can be slightly modified  to argue that the breakdown point is greater than $q/n$ when $\eta>0$ and $l$\ is convex under some regularity conditions. Moreover, interested readers may refer to Appendix \ref{subsec:extra} for   a risk-based   breakdown point defined via 
 the Orlicz norm of the effective noise, which takes the randomness  into account, cf. \eqref{eq:bpyvaryingspace}, \eqref{sbp-1}. One can then combine the empirical process theory and generalized Bregman calculus to perform breakdown point studies for an  {extended} real-valued loss $l(\bsbe;\bsby )$ that may not be a function  of $\bsby - \bsbe$ as in M-estimation. See Remark \ref{rmk:genbp} and
Theorem \ref{theorem:bregBP} for details, where the benefit of using a strongly convex $l$ on the systematic component is shown.

On the other hand, the equivalence between   the      sets of globally optimal solutions  does not preclude   practically obtained estimates    from being   different.  For the analysis of   locally optimal or alternatively optimal estimators to be presented in later sections, our   additive     regularized form appears  to be advantageous  over the trimming form. 
\end{remark}

Finally, an extension  of \eqref{eq:criterion_l2} to  high dimensions gives rise to      resistant variable selection or  resistant shrinkage estimation \begin{equation}\label{eq:criterion_hd}
\min_{\bsbb,   \bsbg  } l(\bsbX \bsbb + \bsbg;\bsby) + P(\bsbb; \lambda) + \frac{\nu}{2}\|\bsbg\|_2^2 \ \mbox{ s.t. }     \|\bsbg\|_0 \leq q.
\end{equation}
     $P(\bsbb; \lambda)$ can be, say,   a sparsity-inducing penalty/constraint, or an $\ell_2$-type penalty. \eqref{eq:criterion_hd}   allows for  $p>n$. Similar to Theorem \ref{th:criterion_equivalence}, in the special case of $l (\bsbe; \bsby) =\| \bsby - \bsbe\|_2^2/2  $, $P(\bsbb; \lambda) = \lambda \|\bsbb\|_1$ and $\nu = 0$,  \eqref{eq:criterion_hd} can be shown to be equivalent to the sparse LTS    \citep{alfons2013sparse}.
In comparison, our explicit introduction of the outlyingness vector $\bsbg$ provides  new        insights. First, outliers  can be directly  revealed    from the $\bsbg$-estimate.    Second, the formulation eases the design of optimization-based algorithms   (Section \ref{sec:algorithm}). Third, \eqref{eq:criterion_l2} and \eqref{eq:criterion_hd} enable us to borrow high dimensional statistics tools to investigate the non-asymptotic behavior of resistant estimators   (Section \ref{sec:theory}) and develop a new information criterion for  regularization parameter tuning (Section \ref{sec:tuning}). 


\section{Progressive Iterative Quantile-thresholding} \label{sec:algorithm}

This section studies how to solve the computational problems defined in Section \ref{sec:regularized},   where the main obstacle lies in   nonconvexity and nonsmoothness. Thanks to the \emph{sparsity} of the problem, obtaining a globally optimal solution  may not be necessary in many real world applications. This section develops two classes of   iterative   algorithms, the BCD type and the MM type, both of which can accommodate a differentiable loss function that is not necessarily convex. More importantly, we propose a   progressive  optimization technique to    relax the condition required to enjoy the best statistical accuracy.  

\subsection{Outlier-resistant regression} \label{subsec:motivating_LTS}

To  illustrate the main ideas and techniques, this part studies   outlier-resistant regression    with   $n > p$:
\begin{equation} \label{reg:criterion}
\min_{\bm\beta\in \mathbb R^p, \bm\gamma\in \mathbb R^n}  \frac{1}{2}\|\bm y - \bm X\bm\beta - \bm\gamma\|_2^2 + \frac{\nu}{2}\|\bm\gamma\|_2^2:= f(\bm\beta,\bm\gamma)   \text{ s.t. } \|\bm\gamma\|_0\le q.
\end{equation}
Compared with LTS, the inclusion of $\bm\gamma$ in \eqref{reg:criterion}   renders the quadratic loss intact, and so    block-wise coordinate descent (\textbf{BCD}) can be used for computation. The sub-problem  of   $\bsbb$ is trivial, and   given $\bm\beta$, the best  $\bsbg$   can be obtained with a quantile thresholding operator $\Theta^{\#}$ \citep{She2013Super}.
Given any $\bm s \in \mathbb R^n$,  define  $\Theta^{\#}(\bsb s; q, \nu)$  to be a vector $\bsb t \in \mathbb R^n$ satisfying
$t_{(j)} = s_{(j)} / (1 + \nu) \mbox{ if } 1 \leq j \leq q \mbox{, and } 0 \mbox{ otherwise,}$
where $s_{(1)}, \ldots, s_{(n)}$ are the order statistics of $s_1, \ldots, s_n$ satisfying $|s_{(1)}| \geq \cdots \geq |s_{(n)}|$, and $t_{(1)}, \ldots, t_{(n)}$ are defined similarly.
To ensure $\Theta^\#$ is a function and avoid ambiguity, we assume   throughout the paper that in performing $\Theta^{\#}(\bm s; q,\nu)$ , either $|s_{(q)}| > |s_{(q+1)}|$ or $|s_{(q)}| = |s_{(q+1)}|=0$ occurs, referred to as the $\Theta^\#$-uniqueness assumption.
The resultant BCD algorithm proceeds as follows
\begin{equation}\label{eq:update_rule_BCD}
\begin{cases}
\bsbg^{(t+1)} = \Theta^{\#}(\bsby - \bsbX \bsbb^{(t)}; q, \nu),\\
\bsbb^{(t+1)} = (\bsbX^T\bsbX)^{+}\bsbX^T(\bsby - \bsbg^{(t+1)}),
\end{cases}
\end{equation}
or
simply\begin{equation} \label{eq:update_BCD_line}
\bsbg^{(t+1)} = \Theta^{\#}(\bsbH \bsbg^{(t)} + (\bsbI - \bsbH) \bsby; q, \nu),
\end{equation}
starting with   $\bm\beta^{(0)}\in\mathbb R^p$ or   $\bsbg^{(1)}\in\mathbb R^n$.
 \eqref{eq:update_rule_BCD}  and
\eqref{eq:update_BCD_line}    are referred to as iterative quantile-thresholding (IQ) algorithms in the paper. Intuitively, 
\eqref{eq:update_BCD_line} repeatedly thresholds a weighted average of  $\bsby$ and the current $\bsbg $-estimate. But the quantile thresholding is no ordinary hard-thresholding due to its \textit{iteration-varying} threshold and   concurrent $\ell_2$ shrinkage.   \eqref{eq:update_BCD_line} shares some similarity with the   celebrated   ``{C-step}''  designed on the basis of trimming  \citep{rousseeuw1999computing}.  On  the other hand,    at each iteration     C-step only  uses  a subset of observations to get the coefficient vector,   while  IQ uses all the samples with     adjustment. In a sense, C-step's subset selection amounts to checking   the zero-nonzero pattern of $\gamma_i^{(t+1)}$ instead of its exact value.  


The   algorithm also suggests some virtues of the complementary  $\ell_2$-shrinkage.  When $\nu = 0$, given the optimal $\hat{\bm\beta}$, $\hat{\gamma_i}=0$ or $  y_i -  \bsbx_i^T\hat\bsbb$. Letting $d_i = 1_{\hat \gamma_i \ne 0}$ and  $\bsbD = \mbox{diag}\{d_i\}$,     from \eqref{eq:update_rule_BCD}, $\hat \bsbb$  satisfies
$$
\bsbX^T (\bsbI - \bsbD) \bsbX \hat \bsbb = \bsbX^T (\bsbI - \bsbD) \bsby.
$$
Because of the binary nature of $d_i$,   the $i$th observation is either kept  unaltered  or  removed completely, which   might hurt the estimation accuracy in the presence of  a   number of   mildly outlying observations.  When   $\nu>0$, the above equation still holds for $ d_i =  1_{\hat \gamma_i \ne 0}/(1+\nu)$,  and the fact that  $\bsbI - \bsbD$ is nonsingular    even when $q$ takes a conservatively large  value     secures numerical stability and  variance reduction. More insights can be gained by transforming  $\bm y = \bm X\bm\beta^*+\bm\gamma^*+\bm\epsilon$ to $\bsby' = \bsbV ^T\bsbg^* + \bsbeps'$ where $\bsbV$ is the orthogonal complement to the range of $\bsbX$,  $\bsby' = \bsbV^T \bsby$ and $\bsbeps' = \bsbV^T \bsbeps$. Then, during the occurrence of clustered outliers with a large  overall leverage, the $\ell_2$ regularization  has a beneficial effect to cope with the coherent  Gram matrix $\bsbV^T \bsbV = \bsbI - \bsbH$.

 Establishing IQ's rate of convergence is nontrivial since it is    a nonconvex optimization algorithm in high dimensions (rigorously speaking, $n$ observations and $p+n$ unknowns). We give two results regarding the convergence of its computational error \textit{and} statistical error.

\begin{thm} \label{thm:BCD_linear_rate}
Consider the algorithm defined by \eqref{eq:update_rule_BCD}  or \eqref{eq:update_BCD_line}.
\begin{description}
  \item a) For any $T \ge 1$,
  \begin{equation} \begin{split}
 \mathop{\min}_{1\le t\le T}  \| \bsbX &(\bsbb^{(t)}  - \bsbb^{(t+1)})\|_2^2 =   \mathop{\min}_{1\le t\le T} \|\bm H(\bm\gamma^{(t)}-\bm\gamma^{(t+1)})\|_2^2 \nonumber\\ \le \ & \frac{1}{ T }(\|(\bm I-\bm H)(\bm\gamma^{(1)}-\bm y)\|_2^2+\nu\|\bm\gamma^{(1)}\|_2^2).
\label{BCD_sublinear}
\end{split}\end{equation}


  \item b)  Let $\bm y = \bm X\bm\beta^*+\bm\gamma^*+\bm\epsilon$, where $\|\bm\gamma^*\|_0 = o^*$, $\bm\epsilon$ is a sub-Gaussian random vector with mean zero and scale bounded by $\sigma$ (cf. Definition \ref{def:subgauss} in Appendix \ref{proof:resistregalg}).   Let $q = \vartheta o^*$ with $\vartheta>1$. Suppose the design matrix  satisfies a restricted isometry property (\emph{RIP}) for some $\varepsilon$:  $0<\varepsilon<1$
      \begin{equation} \label{BCD_linear_cond}
      \| \bm H \bm\Delta\|_2^2 \le (1-\varepsilon)\|\bm\Delta\|_2^2, \  \forall
\bm\Delta:\|\bm\Delta\|_0\le(1+\vartheta)o^*.
      \end{equation} Then with probability at least $1-C(n-p)^{-c}$,   for any $t\ge 1$
      \begin{align} \label{BCD_linear_rate}
       \begin{split}
 \|\bm H(\bm\gamma^{(t)}-\bm\gamma^*)\|_2^2  \lesssim  &  \Big(\frac{1}{1+\kappa}\Big)^{t-1}\|\bm H(\bm\gamma^{(1)}-\bm\gamma^*)\|_2^2+    \frac{\nu}{\kappa} \|\bm\gamma^*\|_2^2
          +    \frac{ \sigma^2}{\kappa} \vartheta  o^*\log\frac{en}{o^* },
       \end{split}
      \end{align}
       where
\begin{align}
       \kappa =\{  \varepsilon-1/\sqrt{\vartheta}+ (1 -1/\sqrt{\vartheta})\nu\}/(1-\epsilon), \label{kappadef}\end{align}       
    and  $C,c>0$ are constants.
Furthermore, under a slightly stronger noise assumption that  $\epsilon_i$ are independent  sub-Gaussian(0,  $\sigma^2$), the following convergence result holds for all $t\ge 1$ with  probability at least $1-C(n-p)^{-c} - C \exp(-c p)$
\begin{align} \label{BCD_linear_rate-beta}
       \begin{split}
          \| \bsbX (\bsbb^{(t)} - \bsbb^*)\|_2^2 \lesssim  &  \Big(\frac{1}{1+\kappa}\Big)^{t-1}\|  \bsbX(\bsbb^{(1)}-\bsbb^*)\|_2^2+    \frac{\nu}{\kappa} \|\bm\gamma^*\|_2^2
       + p\sigma^2   +    \frac{ \sigma^2}{\kappa} \vartheta  o^*\log\frac{en}{o^* }.
       \end{split}
      \end{align}

\end{description}
\end{thm}

The first result in (a) states   a \textit{universal}   convergence rate  of the order       $\mathcal O(1/t)$, and the conclusion is {free} of any    conditions. 
On the other hand, this rate   
   does not take the model {parsimony}   into full account.

The statistical error convergence shown  in (b) is    surprising and delightful.
Although $\Theta^\#$ is  not    a nonexpansive mapping,  let alone  a contraction,  IQ   approaches
the statistical target
\emph{geometrically} fast  with high probability, as long as   $\theta$ or  $\nu$ is properly large:
\begin{align}
\varepsilon + (1 - 1/\sqrt \vartheta) \nu>1/\sqrt \vartheta. \label{epsbound}
\end{align}
Such a linear convergence result  in the discrete nonconvex setup   is novel to the best of our knowledge.
A  positive $\nu$ can improve the  linear  rate parameter $\kappa$. The error bound remains the same order  as long as           $  {\nu}  \|\bm\gamma^*\|_2^2$  is dominated by  a constant multiple of   $  \sigma^2 \vartheta  o^*\log\frac{en}{o^* }$. (The best value of $\nu$ should be  data-adaptive, but fixing it  at a mild value already gives  satisfactory performance.)
The statistical error does not vanish  as $t\rightarrow \infty$ due to the existence of noise, but   Section \ref{sec:tuning}  will show that $p\sigma^2 + o^*\sigma^2 \log(en/o^*)$ is   the desired    error rate in robust regression. Hence over-optimization seems to be unnecessary.

Of course, to enjoy the fast  convergence, a  regularity condition on $\bsbH$ must be assumed.   Restricted isometry like      \eqref{BCD_linear_cond}    is widely used and known to hold in compressed sensing and many high-dimensional sparse problems  \citep{Candes2005,hastie2015statistical}.
 Intuitively,
$
\bsbDelta^T (\bsbI - \bm H) \bsbDelta \ge  \varepsilon \|\bm\Delta\|_2^2
$,   $\forall  \bm\Delta:\|\bm\Delta\|_0\le(1+\vartheta)o^*$   means all principal submatrices  of size $(1+\vartheta)o^* $ of   $\bsbI - \bsbH$ should have   eigenvalues greater than or equal to a positive $\epsilon$. The cone to define the restricted eigenvalue   shrinks when $o^*$ is small. It is well known that RIP is much less demanding than the mutual coherence restrictions on $h_{i,j} (i\ne j)$  \citep{van2009conditions,hastie2015statistical} which are realized by   low-leveraged data since   $|h_{i,j}|\le (h_{ii} h_{jj})^{1/2}$.

  We will show that for strongly convex   losses, pursuing \textit{globally} optimal estimators in our additive regularized framework   can ultimately remove all regularity conditions, namely,     outlier resistance   regardless of    leverage and outlyingness. But on  regular datasets this may  be    unnecessary. The global minimization  comes at a heavy computational cost:  multiple random starts have to be used due to  the nonconvexity, and experience shows that even when  $p=50$, the    subset sampling popularly      used in robust statistics already becomes ineffective and   expensive  \citep{rousseeuw1997recent}.

So a legitimate question on large data  is how to \textit{relax} the condition required for   statistical accuracy while  maintaining   computational efficiency. Theorem \ref{thm:BCD_linear_rate} sheds   new light on the  issue. For example, with $\nu=0$,   $\varepsilon > 1/\sqrt \vartheta$     ensures the geometric statistical convergence, and to make it hold more easily, an   obvious  means   is to  increase $\vartheta$.   
The larger the value of  $\vartheta$  takes,  the bigger $\varepsilon$ is, and so the smaller the influence of the initial point according to \eqref{BCD_linear_rate} or \eqref{BCD_linear_rate-beta}. The adverse effect is the larger  error term    $\vartheta  o^*\log\frac{en}{  o^* }$.
But   we can    {\textit{gradually}} tighten the cardinality constraint with a  sequence  $Q(t)$ which deceases from $n$ to $q$ as $t$ increases. The resultant algorithm   is termed as  \textbf{P}rogressive \textbf{I}terative \textbf{Q}uantile-thresholding (\textbf{PIQ}).
  PIQ can substantially improve the mediocre performance of IQ. It is easy to extend Theorem \ref{thm:BCD_linear_rate} from fixed quantiles to   varying quantiles,   but   deriving a theoretically  optimal cooling scheme is challenging. Fortunately,   this is not a big issue in practice; various schemes  seem to do a decent job, such as       $Q(t) = n-at^2$ (quadratic), $Q(t)=2n/(1+\exp(at))$ (sigmoidal),   $Q(t)=n-a\log t$ (logarithmic) and so on.

Our progressive quantile proposal is innovative in that it is not just another intricate sampling scheme and does not    increase the number of   initial points. Compared with  state-of-the-art algorithms, the tables and figures  in   Section \ref{subsec:simu} clearly show that it can
significantly reduce  the need of data resampling (thereby the overall computational time) without sacrificing the statistical accuracy.  In addition, we  must  point out that PIQ  works in the \textit{opposite} direction to  the class of forward pathwise algorithms and boosting methods  \citep{buhlmann2003boosting,needell2009cosamp,zhang2013multi} which all grow a model from the null. Both of our analysis and computer experiments support   the backward  manner in resistant learning.   
\\

An alternative (and perhaps simpler) algorithm can be developed following the principle of majorization-minimization (\textbf{MM}) \citep{Hunter2004tutorial}. The technique is general and will play a vital role in extending   BCD-type PIQ to conquer a general loss.
Recall the concatenated notation $\bar \bsbX = [\bsbX \  \bsbI]$, $\bar \bsbb = [\bsbb^T \  \bsbg^T]^T$.
Instead of directly minimizing the original objective function $f(\bar \bsbb)= \|\bm y-\bar{\bm X}\bar{\bm\beta}\|_2^2/2 + \nu \|\bsbg\|_2^2/2$, we construct a surrogate function $g (\bar \bsbb; \bar \bsbb^-) $ by (joint) linearization
\begin{equation} \label{eq:surrogate}
g_\rho (\bar \bsbb; \bar \bsbb^-) := \|\bm y-\bar{\bm X}\bar{\bm\beta}^-\|_2^2/2 + \langle \bar \bsbX^T (\bar \bsbX \bar \bsbb^- -\bsby), \bar \bsbb - \bar \bsbb^- \rangle + \frac{\rho}{2} \|\bar \bsbb - \bar \bsbb^-\|_2^2 + \frac{\nu}{2}\|\bsbg\|_2^2,
\end{equation}
where $\rho$  is the inverse stepsize parameter to be chosen later.
Then   minimizing    $g $  with respect to $\bar{\bm\beta}$     yields a sequence of iterates
$
\bar{\bm\beta}^{(t+1)}=\arg\min_{\bar{\bm\beta}}g_\rho(\bar{\bm\beta};\bar{\bm\beta}^{(t)}) \text{ s.t. }\|\bm\gamma\|_0\le q.
$ 
After some algebraic manipulations, $g$ can be rephrased  as
\begin{equation}
\begin{split}
g_\rho(\bar \bsbb; \bar \bsbb^-) =\,&\frac{\rho}{2}\|\bar \bsbb - \bar \bsbb^- +  \bar \bsbX^T (\bar \bsbX \bar \bsbb^- - \bsby)/\rho\|_2^2 + \frac{1}{2}\|\bm y-\bar{\bm X}\bar{\bm\beta}^-\|_2^2\\
 &- \frac{1}{2\rho}\|\bar{\bm X}^T(\bar{\bm X}\bar{\bm\beta}^--\bm y)\|_2^2+ \frac{1}{2}\nu\|\bsbg\|_2^2\\
=\,& \frac{\rho}{2}\|\bar \bsbb - \bar \bsbb^- +  \bar \bsbX^T (\bar \bsbX \bar \bsbb^- - \bsby)/\rho\|_2^2 + \frac{\nu}{2}\|\bsbg\|_2^2 +c(\bar \bsbb^-),
\end{split}
\end{equation}
where $c(\bar \bsbb^-)$ does not depend on $\bm\beta$.
Noticing the separability  in $\bsbb$ and $\bsbg$,  the resultant update is given by
\begin{equation} \label{eq:joint_upadte_LTS}
\begin{cases}
\bsbb^{(t+1)} = \bsbb^{(t)} - \bsbX^T (\bsbX \bsbb^{(t)} + \bsbg^{(t)}-\bm y) / \rho,\\
\bsbg^{(t+1)} = \Theta^{\#}(\bsbg^{(t)} - (\bsbX \bsbb^{(t)} + \bsbg^{(t)}-\bm y) / \rho;q, {\nu}/{\rho}).
\end{cases}
\end{equation}
As long as the stepsize is properly small, e.g.,   $\rho >  \|\bar{\bm X}\|_2^2$, \eqref{eq:joint_upadte_LTS} is convergent    (cf.\,Theorem \ref{theorem:IQ_convergence}).
A similar conclusion  to Theorem \ref{thm:BCD_linear_rate}  can be proved, which supports   a progressive   scheme  with $Q(t)$ decreasing to $q$     to relax the requirement of the initial point.

Different from  \eqref{eq:update_rule_BCD}, \eqref{eq:joint_upadte_LTS} updates $(\bm\beta ,  \bm\gamma)$ simultaneously  
and does not involve any matrix inversion or expensive operations. A downside of MM however lies in its smaller step-size, which  not only slows down  convergence but  may (surprisingly) affect    statistical accuracy in resistant learning; see   Section \ref{sec:theory}.

\subsection{Harnessing a general loss} \label{subsec:comp_general}
This part uses the optimization techniques developed for  resistant regression to deal with   a general objective with   $p$  possibly larger than $n$:
\begin{equation}\label{eq:criterion_general_hd}
\min_{(\bsbb , \bsbg )} l(\bsbX \bsbb + \bsbg;\bsby) + P(\varrho \bsbb; \lambda) + \frac{\nu}{2}\|\bsbg\|_2^2 \ \mbox{  s.t. } \mbox{ } \|\bsbg\|_0 \leq q,
\end{equation}
where   $P$ is often   a sparsity-promoting penalty, but can also take the form of a ridge penalty or an $\ell_0$-constraint. Here,  we add  the algorithmic parameter $\varrho>0$   to match the scale of the design and aid   stepsize selection.
With linearization and MM,  we can extend   \eqref{eq:joint_upadte_LTS} to \eqref{eq:joint_upadte_general-b},  \eqref{eq:joint_upadte_general-g} below, applicable  to a broad family of $l$ and $P$.


\begin{thm}\label{theorem:IQ_convergence}
Assume that $\nabla l(\eta; y)$ is $L$-Lipschitz continuous. Given   an arbitrary thresholding rule  $\Theta$ with $\lambda$ as the threshold parameter (cf. Appendix \ref{app:proofiq}),   consider the following algorithm
\begin{numcases} {}
\bsbb^{(t+1)}  = \Theta(\varrho \bsbb^{(t)} - \bsbX^T  \nabla l(\bsbX \bsbb^{(t)} + \bsbg^{(t)}; \bm y)/\varrho; \lambda )/\varrho , \label{eq:joint_upadte_general-b} \\
\bsbg^{(t+1)}  = \Theta^{\#}(\bsbg^{(t)} -\nabla l(\bsbX \bsbb^{(t)} + \bsbg^{(t)}; \bm y) / \varrho^{2};q, {\nu}/{\varrho^2}), \label{eq:joint_upadte_general-g}
\end{numcases}
where $t\ge 0$, $\varrho>0$.
Let   $f(\bar{\bm\beta}) = l(\bar{\bm X}\bar{\bm\beta};\bm y) + P(\varrho \bsbb; \lambda) + \nu\|\bsbg\|_2^2/2$, where   $P$ is defined based on  $\Theta$:
$
P(\theta; \lambda) - P(0; \lambda) = \int_0^{|\theta|} (\sup \{s:\Theta(s;\lambda) \leq u\} - u) \rd u + q(\theta; \lambda)
$ 
for some  $q(\cdot, \lambda)\ge 0$  satisfying  $q(\Theta(t;\lambda);\lambda) = 0, \forall t \in \mathbb R$.
Then,  given any starting point $\bar{\bm\beta}^{(0)}$ and any $\varrho>0$, $
f(\bar \bsbb^{(t+1)})\le  f(\bar \bsbb^{(t)})-(\bar{\bm\beta}^{(t+1)}-\bar{\bm\beta}^{(t)})^T(\varrho^2 \bm I- L \bar{\bm  X}^T\bar{\bm X})(\bar{\bm\beta}^{(t+1)}-\bar{\bm\beta}^{(t)})/2
$ 
and
 the following convergence rate holds
\begin{equation}
\mathop{\min}_{t\le T} \  (\bar{\bm\beta}^{(t+1)}-\bar{\bm\beta}^{(t)})^T(\varrho^2 \bm I- L \bar{\bm  X}^T\bar{\bm X})(\bar{\bm\beta}^{(t+1)}-\bar{\bm\beta}^{(t)})
\le \frac{2f(\bar{\bm\beta}^{(0)})}{T+1}, \forall T\ge 0.
\end{equation}
  In particular, if  $\varrho^2  >  {L}M_{  \bar \bsbX}(p, 2q ) $ (cf. Section \ref{sec:intro}) and $\nu>0$,  any limit point $(\hat \bsbb, \hat\bsbg)$ of $(\bsbb^{(t)}, \bsbg^{(t)})$  satisfies
$ 
\hat\bsbb = \Theta(\varrho\hat \bsbb - \bsbX^T \nabla l(\bsbX \hat\bsbb + \hat \bsbg; \bm y)/\varrho; \lambda)/\varrho$ and $ 
\hat \bsbg = \Theta^{\#}(\hat \bsbg -  \nabla l(\bsbX \hat \bsbb + \hat \bsbg;\bm y) / \varrho^2;q, {\nu}/{\varrho}^2)
$ 
provided $\Theta$ is continuous at $\varrho\hat \bsbb - \bsbX^T \nabla l(\bsbX \hat \bsbb +\hat  \bsbg; \bm y)/\varrho$.
\end{thm}

 In robust regression, \eqref{eq:joint_upadte_general-b}   degenerates to the  linear $\bsbb$-update in \eqref{eq:joint_upadte_LTS} with $\rho = \varrho^2$, but notice the distinctive    scaled form of   \eqref{eq:joint_upadte_general-b}   associated with a nonlinear thresholding operator $\Theta$.  The Lipschitz condition   is satisfied by the squared error loss, logistic deviance, and Huber's loss, among many others, but is only used to give a universal stepsize. In fact, using a backtracking line search  \citep{Boyd2004},      $L$  need not be known in implementation; see Remark \ref{appremlinesearch}.  The equations satisfied by $(\hat \bsbb, \hat \bsbg)$ actually define  a  broader class of estimators that will  be investigated  in Section \ref{sec:theory}. 
\\

In the remaining, we use BCD  to tackle \eqref{eq:criterion_general_hd}. The   alternating update involves solving two sub-problems:
\begin{align}
\bm\gamma^{(t+1)} & = \mathop{\arg\min}_{\bm\gamma} l( \bm\gamma+\bm X\bm\beta^{(t)}) + \nu\|\bm\gamma\|_2^2/2 \ \, \text{ s.t. }\|\bm\gamma\|_0\le q, \label{eq:general_BCD1} \\
\bm\beta^{(t+1)} & =  \mathop{\arg\min}_{\bm\beta} l(\bm X\bm\beta + \bm\gamma^{(t+1)}) + P(\varrho\bm\beta;\lambda). \label{eq:general_BCD2}
\end{align}
The first problem can be addressed using a similar MM trick as in Theorem \ref{theorem:IQ_convergence}. But there is a more efficient way to   directly locate the optimal support of  $\bsbg $      when $l(\bsbe; \bsby)= \sum_{i=1}^n l_0(\eta_i; y_i)$.
Specifically, with $\tilde l(t;a,b) = l_0(a+t,b) + \nu t^2/2$, the criterion becomes $\sum_{i:\gamma_i = 0}l_0(\bm x_i^T\bm\beta^{(t)}, y_i) + \sum_{i:\gamma_i\ne 0}\tilde l(\gamma_i;\bm x_i^T\bm\beta^{(t)},y_i)$.
Let $l_{0,i} = l_0(\bm x_i^T\bm\beta^{(t)}; y_i)$ and $\tilde l_{\min,i}= \min_t \tilde l(t;\bm x_i^T\bm\beta^{(t)}, y_i)$.  Then the index set $I$ associated with the $q$ largest  differences $\delta_i: = l_{0,i} - \tilde l_{\min,i}$ ($1\le i \le n$)   gives the support.   When $\nu = 0$ and $\tilde l_{\min,i} = 0$,  simply sorting the losses $l_{0,i}$   achieves the goal.
Next, one merely needs to solve   $q$ univariate problems $\min_s\tilde l(s;\bm x_i^T\bm\beta^{(t)}, y_i)$ to get     $\bsbg^{(t+1)}  $. 

For \eqref{eq:general_BCD2},   an iterative MM algorithm similar to \eqref{eq:joint_upadte_general-b} can be developed, and we will show that  it suffices to  performing the  update  just once,
$
\bm\beta^{(t+1)}=\Theta(\varrho\bm\beta^{(t)}-\bm X^T \nabla l(\bm X\bm\beta^{(t)}+\bm\gamma^{(t+1)}; \bm y)/\varrho; \lambda)/\varrho
$. 
   Again, we  advocate   pursuing $\bsbb$-sparsity via the hybrid ``$\ell_0$ + $\ell_2$'' regularization   in place of \eqref{eq:general_BCD2},
$$
\bm\beta^{(t+1)} =  \mathop{\arg\min}_{ } l(\bm X\bm\beta + \bm\gamma^{(t+1)}) +  \nu_\beta\|\bsbb\|_2^2/2  \mbox{ s.t. }\|\bsbb\|_0\le q_{\beta}.
$$
By use of linearization, the resultant   progressive quantile-thresholding   is
given by$$
\bm\beta^{(t+1)} = \Theta^{\#}(\bsbb^{(t)} - \bsbX^T \nabla l(\bsbX \bsbb^{(t)} +\bsbg^{(t+1)};\bm y)/\rho ; q_\beta^{(t)}, \nu_\beta/\rho)
, $$ 
where  $\rho \ge   L M_{  \bsbX}(2q_\beta ) $ (recall      the definition of $M_{  \bsbX}$ in Section \ref{sec:intro}) and the sequence   $q_\beta^{(t)}$   drops to $q_\beta$ eventually.   (One might   want to  merge the two $\ell_0$-constraints into $\|\bar \bsbb\|_0\le q_\beta + q_\gamma$ to simplify the computation, but this would result in a larger statistical error due to the enlarged search space; see the analysis in Section \ref{sec:theory}.)

At the end of the section, we make a comparison between the  MM-type and the BCD-type  PIQ algorithms. The first kind  has a single-loop structure and maintains low per-iteration complexity.     We do observe that MM is   faster than BCD in  some   large-$p$   classification problems (say when $p>500$), but  in   the conventional   $n>p$     setup and   many other scenarios, BCD is more accurate and efficient.   We will mainly use the BCD-type PIQ in applications and analysis unless otherwise specified.

\section{Statistical Accuracy of  A-estimators}
\label{sec:theory}


A major theoretical challenge brought by the class of  estimators  from our computational algorithms     is their lack of global optimality and stationarity,
 which significantly   differentiates them from  conventional \textbf{M}-estimators and \textbf{Z}-estimators  \citep{vdV98}. We call such estimators \textbf{A}-estimators due to  their \textit{alternative}   optimality nature.
In fact,  even the alternative optimality
may  be merely in a local sense. As far as we know, there is a lack of statistical accuracy studies   for  such algorithm-driven estimators. In this section, we aim for  developing new   techniques to reveal tight error rates of A-estimators  as well as the    comprehensive roles of algorithmic parameters in resistant learning, in both low and high dimensions.

Introducing the notion of noise in  the non-likelihood setting is  an essential component.
 We define the \textit{effective noise} associated with a statistical truth $ \bar \bsbb^* = [\bm\beta^{*T}, \bm\gamma^{*T}]^T$   by
\begin{equation} \label{noise-def}
\bm\epsilon = -\nabla l(\bar \bsbX \bar \bsbb^*  ),
\end{equation}
which is not affected by the regularizer.
In a GLM with cumulant function $b$ and canonical link function $g=(b')^{-1}$, the deviance based loss is given by $ l(\bm\eta) = -\langle\bm y, \bm\eta\rangle + \langle\bm 1, b(\bm\eta)\rangle$, and so $$\bm\epsilon = \bm y - g^{-1}(\bm X\bm\beta^*+\bm\gamma^*) = \bm y-\mathbb E(\bm y).$$ 
In general, however, the effective noise jointly determined by the loss and $\bsby$ may not be the raw noise on $\bsby$. For example, a robust loss associated with a bounded $\psi$ function always gives rise to a bounded  $\bsbeps$,  thereby sub-Gaussian regardless of the distribution of $\bsby$, making  our analysis nonparametric in nature; the same is true for  many classification losses.

Unless otherwise specified, we assume that $\bm\epsilon$ is a \textit{sub-Gaussian} random vector with mean zero and scale bounded by $\sigma$ (cf. Definition \ref{def:subgauss} in Appendix \ref{proof:resistregalg}).  This gives  a broad family including Gaussian  and bounded random variables; in particular, any Lipschitz $l$ results in a sub-Gaussian effective noise. However, our analysis  is  not restricted by sub-Gaussianity;   Theorem \ref{theorem:generalnoise} in Appendix \ref{subsec:extra} gives  a universal risk bound when  $\bm\epsilon$  has a general Orlicz  norm.  Also, we allow        $\gamma_i^*$ to take    arbitrarily large values to capture extreme anomalies; another possible way is to     assume $\gamma_i$ are i.i.d. following a certain   distribution, but treating $\bsbg^*$ as an $n$-dimensional unknown parameter vector (with sparsity) is     convenient in detecting outliers and in handling  models with  non-additive noise.


Our main tool to tackle the diverse   form of the loss        is the \textit{generalized Bregman  function} (GBF) \citep{SheBregman}: given any     differentiable $l$\begin{equation}
\bm\Delta_l(\bm\alpha,\bm\beta) := l(\bm\alpha) - l(\bm\beta) - \langle\nabla l(\bm\beta), \bm\alpha-\bm\beta\rangle.
\end{equation}
If, further, $l$ is strictly convex, $\bm\Delta_l(\bm\alpha,\bm\beta)$ becomes the standard Bregman divergence $\mathbf D_\psi(\bm\alpha,\bm\beta)$  \citep{Bregman1967}.   When $l(\bm\beta) = \|\bm\beta\|_2^2/2$, $\bm\Delta_l(\bm\alpha,\bm\beta) = \|\bm\alpha-\bm\beta\|_2^2/2$,   abbreviated as $\mathbf D_2(\bm\alpha,\bm\beta)$. In general, $\Delta_l(\bm\alpha,\bm\beta)$ may not be symmetric, and we define   its symmetrization   by   $\bar{\bm\Delta}_l(\bm\alpha,\bm\beta):=(\bm\Delta_l(\bm\alpha,\bm\beta)+\bm\Delta_l(\bm\beta,\bm\alpha))/2$.

We work in a general setting with $p$ possibly larger than $n$ unless otherwise mentioned.
   For simplicity,     all  $\ell_2$-shrinkage parameters are set  to    zero, and   we assume   $\nabla l$ is      $L$-Lipschitz continuous, and    take $L=1 $ without loss of generality.

Let's first consider       the
advocated doubly constrained form introduced  at the end of Section \ref{sec:algorithm} to  exemplify  the error bounds. As aforementioned, to make the statistical error study  more realistic,  rather than   restricting to  the set of globally optimal solutions, one   needs to pay particular attention to the A-estimators defined by
\begin{subequations}\label{eq:globalopt_BCD}
\begin{align}[left ={ \empheqlbrace}]
      \hat \bsbb \in \arg \min_{ \bm\beta:\|\bm\beta\|_0 \le q_\beta }\allowbreak l(\bm X\bm\beta + \hat \bsbg; \bm y) \label{eq:globalopt_BCDa} \\  \hat \bsbg \in \arg \min_{ \bm\gamma:  \|\bm\gamma\|_0 \le q_\gamma
} l(\bm X\hat \bsbb
+   \bsbg; \bm y). \label{eq:globalopt_BCDb}
\end{align}
\end{subequations}
 But the $\hat \bsbb$ delivered  from the fast   PIQ algorithm  may not possess the   global  alternative optimality in \eqref{eq:globalopt_BCDa}.  The good news is that  according to     Theorem \ref{th:Theta_equation}  in Appendix \ref{proof:Theta_equation},  all these global or local A-estimators satisfy
\begin{equation}\label{eq:fix_point_BCD}
\begin{cases}
\hat \bsbb = \Theta^{\#}(\hat \bsbb - \bsbX^T \nabla l(\bsbX \hat \bsbb+ \hat \bsbg;\bm y)/\rho ; q_\beta) \\
\hat \bsbg = \Theta^{\#}(\hat \bsbg-  \nabla l(\bsbX \hat \bsbb+ \hat \bsbg;\bm y) ;q_\gamma),
\end{cases}
\end{equation}  for any $\rho> M_{\bsbX}(2q_\beta)$. Our statistical accuracy analysis is   applicable to the whole   class (and MM-type non-global estimators as well,  with  slight modification). For clarity, define $o^*=\|\bm\gamma^*\|_0$, $s^*=\|\bm\beta^*\|_0$.


\begin{thm} \label{thm:fixed-point}
Let  $(\hat{\bm\beta}, \hat{\bm\gamma})$ be an A-estimator     that satisfies \eqref{eq:fix_point_BCD} for some $\rho>0$, $  q_\gamma= \vartheta o^*,   q_\beta = \vartheta s^* $ ($\vartheta\ge 1$) and $\|\hat{\bm\gamma}\|_0 = q_\gamma$, $\|\hat{\bm\beta}\|_0 = q_\beta$.
 Let $\bar \bsbX_\rho =[  {\bsbX}/{\sqrt \rho} \ \ \bsbI]$. Assume that there exists some $\delta >0$ such that
\begin{equation} \label{fixed-point-condition}
\begin{split}
  (2\bar{\bm\Delta}_l -  \delta\Breg_2 )(\bar \bsbX_{\rho} \bar \bsbb , \bar \bsbX_{\rho} \bar \bsbb ')
 \ge \frac{1}{\sqrt{\vartheta}} \Breg_2(  \bar \bsbb,  \bar \bsbb ')   \end{split}
\end{equation}
for sparse $\bar \bsbb, \bar \bsbb'$ satisfying  $\|\bm\beta\|_0\le \vartheta s^*, \|\bm\beta'\|_0 \le  s^*$, $\|\bm\gamma\|_0 \le \vartheta o^*, \| \bm\gamma'\|_0 \le  o^*$. Then the following   error bound
holds\begin{equation} \label{fixed-point-rate}
\EE [\Breg_2 (\bar \bsbX \hat {\bar \bsbb}, \bar \bsbX \bar \bsbb^*  ) ]
\lesssim \frac{1}{\delta^2}\Big \{\vartheta o^*\sigma^2\log\frac{en}{ \vartheta o^*} + \vartheta s^* \sigma^2\log\frac{ep}{ \vartheta s^*}+{\sigma^2 }\Big\}
, \end{equation}
 where $C,c$ are positive constants.

 More generally, if   the starting point $(\bsbb^{(0)},\bsbg^{(0)})$: $ \| \bsbb^{(0)}\|_0\le q_\beta, \|\bsbg^{(0)}\|_0\le q_\gamma$  already satisfies $\EE [\Breg_2 (\bar \bsbX  {\bar \bsbb}^{(0)}, \bar \bsbX \bar \bsbb^*  ) ]= O(R) \allowbreak\{\vartheta o^*\sigma^2 \allowbreak \log\frac{en}{ \vartheta o^*} + \vartheta s^* \sigma^2\log\frac{ep}{ \vartheta s^*} +\sigma^2 \}$ with $  +\infty\ge R \ge 1$,  to obtain the same conclusion \eqref{fixed-point-rate}, \eqref{fixed-point-condition}   can be relaxed to
\begin{equation}\big\{2(1-\frac{1}{R})\bar{\bm\Delta}_l +\frac{C}{  R(R\delta \vee 1)} \breg_l -  \delta\Breg_2 \big\}(\bar \bsbX_{\rho} \bar \bsbb , \bar \bsbX_{\rho} \bar \bsbb ')
 \ge \frac{(1-1/R)}{\sqrt{\vartheta}} \Breg_2(  \bar \bsbb,  \bar \bsbb ')\label{eq:genregcond}   \end{equation}
for some   constant $C>0$.
\end{thm}

The regularity condition is an extension of the classical $\ell_2$-form restricted isometry imposed   on the augmented design.  Indeed,    for $l(\bsbe; \bsby) = \| \bsbe  - \bsby \|_2^2/2$,  the   condition  can written as  $
\|\bm X\bm\beta+\bm\gamma \|_2^2 \ge  [{1}/\{{(2 -\delta)}\sqrt{\vartheta}\}] \big(\rho \|\bsbb\|_2^2 + \|\bm\gamma\|_2^2\big) , \ \forall (\bsbb, \bsbg): \|\bm\beta\|_0\le (1+\vartheta) s^* ,  \|\bm\gamma\|_0 \le (1+\vartheta) o^*
$ for some     $\delta>0$ and large enough $\vartheta$.
According to the proof, under $q_\beta = p$,  the term   $\rho \|\bsbb\|_2^2 $ on its right-hand side  can be dropped  and so the condition degenerates to   the ordinary RIP used  in Theorem \ref{thm:BCD_linear_rate}, due to the orthogonal decomposition $\| \bsbX \bsbb + \bsbg\|_2^2 = \|\bsbX \{\bsbb + (\bsbX^T \bsbX)^+\bsbX^T \bsbg\}\|_2^2 + \| (\bsbI - \bsbH)\bsbg\|_2^2$. We have argued in Section \ref{sec:algorithm} that in low dimensions,   low leverage implies   low mutual coherence which in turn implies RIP. But the plain  definition of leverage    cannot   incorporate    outlier scarcity,  and is noninformative as   $p>n$,   since   $\bsbH = \bsbI$ when $\bsbX$ has full row rank. The Bregman-based restricted isometry  on the augmented design gives an extension of   the    notion of leverage to   high dimensions and non-quadratic losses.

Theorem \ref{thm:fixed-point} gives  useful   implications   regarding the algorithmic parameters as well. For example, from     \eqref{fixed-point-condition}, to enjoy the    statistical guarantee, the inverse stepsize $\rho$   must be     properly small.  The finding is insightful and novel. Throughout   the  machine learning literature, slow   learning rate is  recommended   when training  a nonconvex   learner.   But our statistical error analysis   strongly cautions  against using an  extremely small step size when nonconvexity occurs. The BCD-type algorithm design appears advantageous in this aspect.   Moreover,  a high-quality  starting point, usually expensive in computation, certainly brings some statistical benefit  seen from the generalized regularity condition \eqref{eq:genregcond} (consider an extreme case  $R=1$). But simply enlarging  $\vartheta$ gives another  promising way  to relax \eqref{fixed-point-condition}, and     supports  the progressive tightening proposal  in Section \ref{sec:algorithm}.

When $\vartheta, \delta$ are treated as constants and $s^*\vee o^* >0$,   \eqref{fixed-point-rate} shows that A-estimators can achieve an error bound of the order $\mathcal \sigma^2s^*\log(ep/s^*)+\sigma^2o^*\log(en/o^*)$, which reduces to $p\sigma^2 +\sigma^2o^*\log(en/o^*)$ when there is no $\bsbb$-sparsity.   This error rate    involves  the number of outliers apart from
the number of relevant predictors. But the influence of outliers remains {controlled}  in the sense that the bound does not grow with the magnitude of $\bsbg^*$. In the  situation of  relatively few    outliers: $o^*\log (en/o^*)\ll s^* \log (e p /s^*)$, the proposed method can  attain
the  celebrated   rate   $\sigma^2 s^* \log (ep/s^*)$  for (clean)  large-$p$ variable selection.  On the other hand, one should be aware that on `big dirty' data with large $n$ and $ o^*$, $\sigma^2 o^*\log (en/o^*)$ could be the dominant  error term.

\begin{remark} 
 One may wonder whether the overall error rate in Theorem \ref{thm:fixed-point} can be further improved by pursuing a global minimizer. Theorem \ref{thm:global_rate} shows that  this is not the case,      but the regularity condition gets
relaxed. Concretely,   the symmetrized GBF   $\bar{\bm \Delta}_l$   in
\eqref{fixed-point-condition}  will be replaced by  ${\bm \Delta}_l$,  and    the right-hand side  of
\eqref{fixed-point-condition}  will  be  just zero. Hence  if the loss $l$ defined on the systematic component is $\mu$-strongly convex (e.g., regression), the error bound holds   with $\delta = \mu$, \emph{free} of any RIP or leverage restriction. This shows the inherent soundness of the proposed resistant learning framework.
\end{remark}

\begin{remark}
 The $\ell_2$-recovery  result proved in Theorem \ref{thm:fixed-point} is fundamental, from which one can also obtain   estimation error bounds.  For example,
Theorem \ref{th:esterr} shows that under  a slightly stronger regularity condition,  $$ \|  \hat { \bsbb} -  \bsbb^*  \|_2^2 \lesssim  \sigma^2 \{ s^*\log(ep/s^*) + o^* \log (en/o^*)\}/n
$$  holds with high probability. In the classical    $n>p$  setup with no penalty imposed on $\bm\beta$,   it becomes  $ \|  \hat { \bsbb} -  \bsbb^*
\|_2^2 \lesssim   \{ p \sigma^2 +  o^* \sigma^2 \log (en/o^*)\}/n
$. See    Appendix \ref{proof:fixed-point}  for details. 
\end{remark}

Our theoretical techniques  can cope with the penalized form of \eqref{eq:criterion_general_hd} as well.   For instance, Theorem \ref{thm:fixed-point-lasso}    studies  the   outlier-resistant lasso (cf. Section \ref{sec:regularized}), where    $$l(\bm X\bm\beta+\bm\gamma; \bm y) = \|\bm X\bm\beta + \bm\gamma - \bm y\|_2^2/2, \ P(\bm\beta;\lambda) = \lambda \varrho  \|\bm\beta\|_1,$$ with  $\varrho> 0 $. (The theorem proved
in Appendix \ref{proof:fixed-point-l1} is actually regarding   a general loss and its proof applies to all subadditive penalties.)

\begin{thm} \label{thm:fixed-point-lasso}
Let $\lambda = A\sigma\sqrt{\log(ep)}$  with $A$ a sufficiently large constant. Then if      $(\hat{\bm\beta},\hat{\bm\gamma})$ is an A-estimator of resistant lasso for some $\varrho$ and $\vartheta$ (cf. Theorem \ref{th:Theta_equation}),  the   risk  bound   
\begin{equation}
\EE [\Breg_2 (\bar \bsbX \bar \bsbb^*, \bar \bsbX \hat {\bar \bsbb}) ]\lesssim \sigma^2 \vartheta o^*\log\frac{en}{o^*} + K^2 \sigma^2 s^*\log(ep) +\sigma^2 
\label{l1errorbound} \end{equation}
holds under the assumption that  there exists a large enough $K$ such that
$ K\sqrt J\|\bm X\bm\beta +\bm\gamma \|_2
+\varrho \| \bm\beta_{\mathcal J^c}\|_1
\ge \|\bm\gamma\|_2^2/(2\lambda\sqrt{\vartheta})  + (1+\varepsilon)\varrho \| \bm\beta_{\mathcal J}\|_1$   for any $\bm\beta,\bm\gamma$ with $ \|\bm\gamma\|_0\le(1+\vartheta)o^*$,  where   $\varepsilon$ is a positive constant, $\mathcal J = \{j:\beta_j^*\ne 0\}$ and $J = |\mathcal J|$. \end{thm}

When $K$ and $\vartheta$ are treated as   constants, the error rate of the resistant lasso is slightly worse  (by a
logarithmic term) than      \eqref{fixed-point-rate}  for the doubly constrained form.


\section{Minimax Lower Bounds and  Model Comparison  } \label{sec:tuning}

 We first  derive  some   universal lower bounds that are satisfied by   all  estimators. These nonasymptotic results are stated for  GLMs. Let $\bsbe^* = \bar \bsbX \bar \bsbb^*$ and   $y_i | \eta_i^*$  be  independent and follow a distribution in the regular exponential   family with dispersion $\sigma^2$ and  $l_0$ be the negative log-likelihood function (cf. Appendix \ref{subsec:proofofminimax} for  details), and consider a general   signal class with $p$ possibly larger than $n$
\begin{equation}
\mathcal B(s^*,o^*) = \{(\bm\beta^*,\bm\gamma^*): \|\bm\beta^*\|_0\le s^*, \|\bm\gamma^*\|_0 \le o^*\}.
\end{equation}
\begin{thm} \label{thm:minimax}

In the  regular exponential dispersion family with $n\ge 2, p\ge 2, 1\le o^* \le n/2, 1\le s^*\le p/2,$ define   \begin{align}
P (s^*,o^*) = s^*\log(ep/s^*) + o^*\log(en/o^*).
\end{align}
Let $I(\cdot)$ be any nondecreasing   function with $I(0)=0, I\not\equiv 0$.
 (i) Suppose for some $\kappa > 0$    $  \breg_l (\bsb0, \bar \bsbX \bar \bsbb       )\sigma^2\allowbreak\le    \kappa\Breg_2(  \bsb0,  \bar \bsbb    )$ $\forall  \bar{\bm\beta}  :   \|\bm\beta \|_0 \le s^*$,   $\|\bm\gamma \|_0\le o^*$.  Then there exist positive constants $\tilde c, c$, depending on $I(\cdot)$ only, such that
\begin{equation}
\inf_{(\hat{\bm\beta},\hat{\bm\gamma})}\,\sup_{(\bm\beta^*,\bm\gamma^*)\in \mathcal B(s^*,o^*)} \mathbb E[I(\Breg_2  (   \bar \bsbb^*,    \hat{ \bar \bsbb} )/\{\tilde c \sigma^2  P (s^*,o^*)/\kappa\})] \ge c >0,
\label{minimaxrobglm-est}
\end{equation}
where $(\hat{\bm\beta},\hat{\bm\gamma})$ denotes any estimator of $(\bm\beta^*,\bm\gamma^*)$. (ii) Suppose $  \breg_l (\bsb{0},\bar \bsbX \bar \bsbb_1  )\sigma^2\allowbreak\le   \overline \kappa\Breg_2(   \bsb{0}, \bar \bsbb_1      )$ and  $\underline \kappa\Breg_2(  \bsb0,  \bar \bsbb_2     )\le  \Breg_2( \bsb0,\bar \bsbX    \bar \bsbb_2   ) $,    $\forall  \bar{\bm\beta}_i  :    \| \bsbb_i\|_0 \le i s^*,    \|\bsbg_i\|_0\le i o^*, i =1,2$,  where $\underline \kappa/ \overline \kappa$ is a    positive   constant. Then  there exist positive constants $\tilde c, c$ such that
\begin{equation}
\inf_{(\hat{\bm\beta},\hat{\bm\gamma})}\,\sup_{(\bm\beta^*,\bm\gamma^*)\in \mathcal B(s^*,o^*)} \mathbb E[I(\Breg_2  (\bar \bsbX   \bar \bsbb^*,   \bar \bsbX \hat{ \bar \bsbb} )/\{\tilde c \sigma^2  P (s^*,o^*)\})] \ge c >0.
\label{minimaxrobglm-pred}
\end{equation}
\end{thm}

Finer analysis and lower bounds are presented in Appendix  \ref{subsec:proofofminimax}.
We  illustrate the conclusion for  classification. Because  the logistic deviance has a $1/4$-Lipschitz continuous gradient, the regularity condition in (i) is satisfied for any  $\kappa \ge M_{\bar \bsbX}(s^*, o^*) /4$ (cf. Section \ref{sec:intro} and Appendix \ref{app:proofiq}). So when $\kappa \le c   n$,   $\{s^*\log(ep/s^*) + o^*\log(en/o^*)\}/n$,   or $(p + o^*\log(en/o^*))/n$ when $s^* = p$, is the  desired minimax lower error rate for  all classifiers,    both with constant positive probability and in expectation (corresponding to $I(t) = 1_{t\ge c}$ and $I(t) =t$). A similar conclusion for prediction  can be drawn from (ii).  To the best of our knowledge, such a kind of results with the usage of GBFs     for resistant GLMs
is novel  in robust statistics.

 Together with the upper error bounds (e.g., Theorem  \ref{thm:fixed-point} and Theorem  \ref{th:esterr}), PIQ   enjoys minimax
rate optimality provided that $q_\gamma$  is not set too large relative to $o^*$. It is common  in practice to   directly specify  the cardinality bound based on domain knowledge, say,   $\| \bsbg\|_0\le    0.25 n$, for the sake of $\bsb{\beta}$-estimation.
 On the other hand, if the goal is    to identify all outliers  to have the best predictive model, one ought to choose the  regularization parameters   in a more data-adaptive manner. This is regarded as  a vital and challenging task in robust statistics \citep{ronchetti1997robust,muller2005outlier,salibian2008robust}.   Indeed, outliers may render  resampling based cross-validation   inappropriate, and although  there is an abundance of  information criteria, which one is sound  in theory and  reality   lacks a clear conclusion.  The joint  variable selection in high dimensions and the flexible   choice of the loss     exacerbate the issue. (The $\ell_2$-shrinkage parameter is however much less sensitive and can be fixed at a mild value.)

In the general setup   where both       $\bm\beta$   and $\bsbg$ are   possibly sparse,  the upper  and lower error bounds  suggest  a universal  function to penalize model  complexity
\begin{equation}\label{generalpenfortuning}
P (\bm\beta,\bm\gamma) =  \| \bsbg\|_0+ \| \bsbg\|_0\log(en/\| \bsbg\|_0)+ \| \bsbb\|_0+ \| \bsbb\|_0\log(ep/\| \bsbb\|_0).
\end{equation}
We refer to the   associated criterion as the \textit{predictive information criterion} (\textbf{PIC}).
In the case that no  variable selection is needed ($\|  \bsbb\|_0=p$), one can rewrite  \eqref{generalpenfortuning}  as
\begin{equation}\label{generalpenfortuning-gammaonly}
P ( \bm\gamma) =  \| \bsbg\|_0+ \| \bsbg\|_0\log(en/\| \bsbg\|_0)
\end{equation}
for pure outlier identification.

    \eqref{generalpenfortuning} has some  involved logarithmic terms and is not a simple  $\ell_0$-type penalty,  but it leads to an ideal  model comparison criterion with a solid finite-sample theoretical support.   
\begin{thm} \label{thm:tuning}
Assume    $\bm\epsilon$ defined in \eqref{noise-def} is a sub-Gaussian random vector with mean zero and   scale bounded by a constant and   $\bar\bsbb^*\in \mathcal M$ and $\bar\bsbb^*\ne \bsb0 $. Define     $P(\bar \bsbb) = P (\bm\beta,\bm\gamma) =  J(\bm\beta)\log(ep/J(\bm\beta)) + J(\bm\gamma)\log(en/J(\bm\gamma))$. Assume that there exist constants $\delta>0$ and $A_0\ge 0$ such that $(\bm\Delta_l- \delta\mathbf D_2)(\bar \bsbX \bar\bsbb , \bar \bsbX \bar\bsbb ') +A_{0} (P(\bar \bsbb ) +  P(\bar \bsbb' ))\ge0$,  $\forall \bar\bsbb, \bar\bsbb'\in \mathcal M$.    Then for a sufficiently large constant $A$, any
      $(\hat{\bm\beta},\hat{\bm\gamma})$  that minimizes
\begin{equation} \label{PIC-general}
l(\bar \bsbX \bar\bsbb ; \bm y) + AP (\bar \bsbb)
\end{equation}
subject to $\bar \bsbb\in \mathcal M$  must satisfy
\begin{equation} \label{tuning_general-1}
\mathbb E\big[ \Breg_2 (\bar \bsbX \hat {\bar \bsbb}, \bar \bsbX \bar \bsbb^*  )\big] \lesssim  s^*\log(ep/s^*) + o^*\log(en/o^*).
\end{equation}
\end{thm}


Theorem \ref{thm:tuning} achieves the optimal error rate as   Theorem \ref{thm:fixed-point} and Theorem \ref{thm:fixed-point-lasso}, but a major difference here is that  no  parameters (like $q ,  \lambda$) are involved,
but just some absolute constants. Moreover, Theorem \ref{thm:tuning} does not require any  RIP or large signal-to-noise ratio conditions.

   When  the noise distribution has a  dispersion parameter $\sigma^2$, Theorem \ref{thm:tuning} still applies, but the penalty in \eqref{PIC-general} becomes $A \sigma^2 P (\bsbb, \bsbg)$ with an unknown factor.  A preliminary robust scale estimate can be possibly used. But an appealing result     for robust regression is that  the estimation of $\sigma$ can be totally bypassed.
For clarity, Theorem \ref{thm:tuning_log_form}   assumes    $n>p$ and
   no sparsity in $\bsbb$, in which situation    a scale-free form of PIC, motivated by \cite{SCV}, is given by
\begin{equation} \label{PIC-log}
(n-p)\log \big(\|\bm X\bm\beta +\bm\gamma - \bm y\|_2^2\big)+ \alpha_1  \|\bsbg\|_0 + \alpha_2  \|\bsbg\|_0\log(en/ \|\bsbg\|_0),
\end{equation}
where   $\alpha_1, \alpha_2$ are absolute constants.
 Again, the soundness of \eqref{PIC-log} when the outliers are not awfully many places  no limit on their leverage or outlyingness.  (More generally, using \eqref{generalpenfortuning}  in place of \eqref{generalpenfortuning-gammaonly} gives a more general form of scale-free PIC for   joint variable selection and outlier identification;    Theorem  \ref{thm:tuning_log_form}'
proved in
Appendix \ref{subsec:sfpic-high}   implies Theorem \ref{thm:tuning_log_form},  and   applies to $p>n$.)  

\begin{thm} \label{thm:tuning_log_form}
Let $\bm y = \bm X\bm\beta^* + \bm\gamma^*+\bm\epsilon$, where $\bm X$ has full column rank, $\epsilon_i$ are independent sub-Gaussian$( 0,\sigma^2)$ and $\EE \epsilon_i^2 =  \sigma_i^2 = c_i \sigma^2$ with $c_i$  some positive constants and $\sigma^2$ unknown.  Define   $l(\bm X\bm\beta +\bm\gamma ; \bm y) = \|\bm X\bm\beta +\bm\gamma - \bm y\|_2^2$ and $P (\bm\gamma) =  \|\bsbg\|_0\log(en/ \|\bsbg\|_0)$. Assume  $P(\bm\gamma^*)\le (n-p)/A_0$ for some constant $A_0 > 0$. Let $\delta(\bm\gamma) = AP (\bm\gamma)/(n-p)$, where $A$ is a positive constant satisfying $A<A_0$, and so $\delta(\bm\gamma^*)<1$. Then, for sufficiently large values of $A_0$ and $A$, any $(\hat \bsbb, \hat{\bm\gamma})$ that minimizes  $\log l(\bm X\bm\beta +\bm\gamma ; \bm y) + \delta(\bm\gamma)$ 
subject to $\delta(\bm\gamma)<1$ must satisfy
$ 
\Breg_2 (\bar \bsbX \hat {\bar \bsbb}, \bar \bsbX \bar \bsbb^*  )   \lesssim \sigma^2o^*\log(en/o^*) + p\sigma^2
$ 
with probability at least $1-c  (n-p) ^{-c' }$ for some constants $c ,c '>0$.
\end{thm}



 Combining Theorem \ref{thm:minimax} and Theorem  \ref{thm:tuning_log_form}  supports     
  the {modified BIC} proposed  by \cite{she2011outlier} based on empirical studies.
 However,
\eqref{PIC-log}   has some subtle but important differences from    BIC, and the resemblance is just a coincidence. First, unlike most information criteria, its derivation and justification  do not need an    infinite sample size. Second, the factor  $\log (en/  \|\bsbg\|_0 )$---not  exactly $\log n$  as in BIC that assumes clean data---arises  due to the seek for outliers. Third,  the general PIC complexity term  \eqref{generalpenfortuning} (cf. Theorem \ref{thm:tuning} or  Theorem \ref{thm:tuning_log_form}')  has the overall inflation effect \textit{added} to  the degrees of freedom, as opposed to the standard multiplicative relation in    AIC,    BIC, and many others. 

 In common with most nonasymptotic studies,  our theorems do not show  tight constants. However, these  numerical constants   do not depend on a specific problem and in some easier cases can be determined by Monte Carlo experiments---for example, when $n>p$, we recommend $5.5   \|\bsbg\|_0 +   \|\bsbg\|_0\log(en/\|\bsbg\|_0)$  in \eqref{PIC-log} for robust regression. 
Theoretically deriving the optimal   constants  requires much finer
analysis and we will not pursue     in the current paper.

\section{Experiments} \label{sec:exp}

\subsection{Simulations} \label{subsec:simu}
In this part, we perform simulation studies to compare PIQ with some popular robust and/or sparse estimation methods in $p<n$ and $p>n$ setups. Unless otherwise mentioned, the raw predictor matrix $\bm X = [\bm x_1, \bm x_2, \ldots, \bm x_n]^T\in \mathbb R^{n\times p}$ is generated  in the default way: $\bm x_i\iid \mathcal N(\bm 0,\bm\Sigma)$, where $\bm \Sigma$ is a Toeplitz design covariance matrix with $\Sigma_{ij} = \rho^{|i-j|}$. Here, we consider regression and classification models contaminated with highly-leveraged outliers in the following four setups. More experiment results in other settings are in reported in Appendix \ref{sec:add_exp} due to limited space. Recall $o^* = \| \bsbg^*\|_0, s^* = \| \bsbb^*\|_0$.

\vspace{2ex}
\noindent\textbf{Example 1 ($n>p$, regression)}: $n=1000, p=10, \rho=0.5$, $\bm\beta^*= [1,1,0.5,0.5,-1.5,\allowbreak-1.5,-1,\allowbreak -1,1,1]^T$. We set $\bm\gamma^* = [5,\ldots,5,0,\ldots,0]^T$ with the first $o^*$   components   nonzero, and modify the first $o^*$ rows of $\bm X$ to $[3,\ldots, 3]$. The response vector is    generated according to $\bm y = \bm X\bm\beta^* + \bm\gamma^* + \bm\epsilon$ with $\epsilon_i\iid \mathcal N(0,1)$.\\

\noindent\textbf{Example 2 ($n>p$, classification)}: $n=1000, p=10, \rho=0.5$, $\bm\beta^* = [3,3,1.5,\allowbreak1.5,3,3,-3,-3,3,3]^T$. To introduce  outliers, we change the first $o^*$ rows of $\bm X$ to $[3,\ldots, 3]$   so that the first $o^*$ elements in $\bm X\bm\beta^*$ are   45 and  set $\gamma_i^* = -90$, $1\le i\le o^*$ and 0 otherwise. The response vector is generated according to the Bernoulli distribution with mean $1/(1+\exp(-\bm x_i^T\bm\beta^* - \gamma_i^* ))$ for the $i$th observation.\\

\noindent\textbf{Example 3 ($p>n$, regression)}: $n=200, p=1000$, $\bm\beta^*= [1,0.5,0,0,\allowbreak-0.5,-1, 0,  \ldots, 0]^T$ so that $s^*=4$, $\bm\gamma^* = [5,\ldots,5,0,\ldots,0]^T$. High-leveraged outliers are introduced in the same way as in Example 1.
\\

\noindent\textbf{Example 4 ($p>n$, classification)}: $n=200, p=1000$, $\bm\beta^* =  [3,1.5,3,\allowbreak0,\ldots,0]^T$. Similar to Example 2,      outliers are added,  with  $\gamma^*_i = -45$, $1\le i\le o^*$ and 0 otherwise. 
\\

The following 11 methods are included for comparison:     S-estimator  \citep{rousseeuw1984robust}, LTS \citep{rousseeuw1985multivariate},  Bianco and Yohai's robust logistic regression (B-Y)  \citep{bianco1996robust}, TLE \citep{hadi1997maximum},     robust quasi-likelihood estimator (QLE) \citep{robustGLM2001},  quantile lasso (QL)  \citep{Belloni2011}, robust LARS (RLARS) \citep{khan2007robust}, sparse LTS (S-LTS) \citep{alfons2013sparse}, sparse maximum tangent-likelihood estimator  (S-MTE) \citep{Qin2017},     elastic-net LTS for classification  (enetLTS)  \citep{Kurnaz2018} and PENSE \citep{freue2019robust}.
To reduce the interference of different regularization parameter tuning schemes, and to make a fair comparison, sparsity parameters and cut-off values for the sake of outlier identification or variable selection are all chosen to yield $1.5o^*$ or $1.5s^*$ nonzeros using the true model.
The other parameters are set to their default values. In calling PIQ, we use the BCD-type, with $\nu$  fixed at $10^{-4}$.
A quadratic cooling schedule $Q(t) = -at^2 + U$ with $a = (U-L)/T^2$ ($0\le t\le T$) is applied in regression and a logarithmic cooling $Q(t) = - a\log t + U$  ($0< t\le T$) with $a=(U-L)/\log T$ is used in classification, where $L$ is the desired cardinality,    $U=n$ or $p$ for outlier identification or variable selection,  and $T=200$.
Owing to the progressive backward optimization, a simple single starting point $\bsbb^{(0)} = \bsb0$ already seems satisfactory in PIQ.
In each setup, we repeat the experiment 50 times and evaluate the performance of an algorithm using  the following metrics.
For outlier identification, we report the masking or missing (\textbf{M}) rate, as well as the rate of successful joint detection (\textbf{JD}); see \cite{she2011outlier}. Concretely, the masking probability is the fraction of undetected true outliers (misses), and the JD is the fraction of simulations with zero miss.
For variable selection, we report the false alarm (\textbf{FA}) rate (the fraction of spuriously identified variables) in addition to M and JD. In regression experiments, the mean square error on $\bm\beta$, denoted by \textbf{Err}, is   used to evaluate estimation accuracy, while in classification, \textbf{Err} refers to the misclassification error on a separate clean testset containing  10,000 observations.
The computational time in seconds, denoted by \textbf{T}, is also included in the tables or figures to describe the complexity of each algorithm.

\begin{table}[!h]
  \setlength{\tabcolsep}{3pt}\centering\footnotesize
{\footnotesize 
  \caption{\small Performance comparison in $n>p$ settings (Example 1, Example 2 with $n=1000,p=10$). The outlier percentage varies from 1\% to 20\% for regression, and   3\% to 15\% for classification. \label{table:ex1-2}
}
\vspace{.1in}
  \begin{tabular}{lccccccccccccccccccc}
  \hline
  & \multicolumn{19}{c}{\textbf{Regression}}\\

   & \multicolumn{3}{c}{$o^*=10$} && \multicolumn{3}{c}{$o^*=50$} && \multicolumn{3}{c}{$o^*=100$} && \multicolumn{3}{c}{$o^*=150$} && \multicolumn{3}{c}{$o^*=200$}\\ \cmidrule(lr){2-4} \cmidrule(lr){6-8} \cmidrule(lr){10-12} \cmidrule(lr){14-16} \cmidrule(lr){18-20}
              & Err & M & JD && Err & M &JD && Err & M & JD && Err & M & JD && Err & M & JD\\
  \hline
  S           & 0.06 & 2.0 & 82 &  & 0.05 & 0.6  & 78 &  & 0.05 &  2.2 &  76 &   & 0.27 & 60  & 26 &  & 0.40 & 77 & 0 \\
  LTS         & 0.02 & 2.0 & 82 &  & 0.02 & 0.8  & 70 &  & 0.03 &  4.3 &  64 &   & 0.13 & 38  & 40 &  & 0.29 & 74 & 2 \\
 RLARS & 0.02 & 2.4 & 78 &  & 0.11 & 48   & 8  &  & 0.25 &  78  &  0  &   & 0.26 & 78  & 0  &  & 0.28 & 74 & 0 \\
  PENSE       & 0.02 & 1.6 & 86 &  & 0.02 & 0.6  & 78 &  & 0.23 &  79  &  0  &   & 0.27 & 80  & 0  &  & 0.29 & 75 & 0 \\
  \textbf{PIQ} & 0.02 & 1.6 & 86 &  & 0.02 & 0.4  & 84 &  & 0.03 &  0.2 & 88  &   & 0.03 & 0.1 & 90  &  & 0.05 & 0.1  & 90 \\


  \hline\hline

  & \multicolumn{19}{c}{\textbf{Classification}}\\

     & \multicolumn{3}{c}{$o^*=30$} && \multicolumn{3}{c}{$o^*=60$} && \multicolumn{3}{c}{$o^*=90$} && \multicolumn{3}{c}{$o^*=120$} && \multicolumn{3}{c}{$o^*=150$}\\ \cmidrule(lr){2-4} \cmidrule(lr){6-8} \cmidrule(lr){10-12} \cmidrule(lr){14-16} \cmidrule(lr){18-20}
              & Err & M & JD && Err & M & JD  && Err & M & JD  && Err & M & JD && Err & M & JD\\
  \hline
  B-Y     & 0.06 & 0.1 &  98     &  & 0.14 &  42 & 56  &   & 0.28 &  100 & 0     &  & 0.29 & 100 & 0   &  & 0.30 & 100& 0\\
  QLE     & 0.06 &  0  &  100    &  & 0.23 &  88 & 12  &   & 0.28 &  100 & 0     &  & 0.29 & 100 & 0   &  & 0.30 & 100& 0\\
  TLE     & 0.09 &  0  &  100    &  & 0.09 &  0  & 100 &   & 0.08 &   0  & 100   &  & 0.11 &  16 & 84  &  & 0.35 & 100& 0\\
  \textbf{PIQ}   & 0.06 &0&100    &  &  0.06  & 0& 100 &  & 0.07 &   0 & 100 &  & 0.07 &  0 & 100  & & 0.08 & 0  & 100 \\
  \hline
  \end{tabular}
}
\end{table}

\begin{figure}[!h]
  \centering
  \subfigure[regression]{\includegraphics[width=0.48\textwidth]{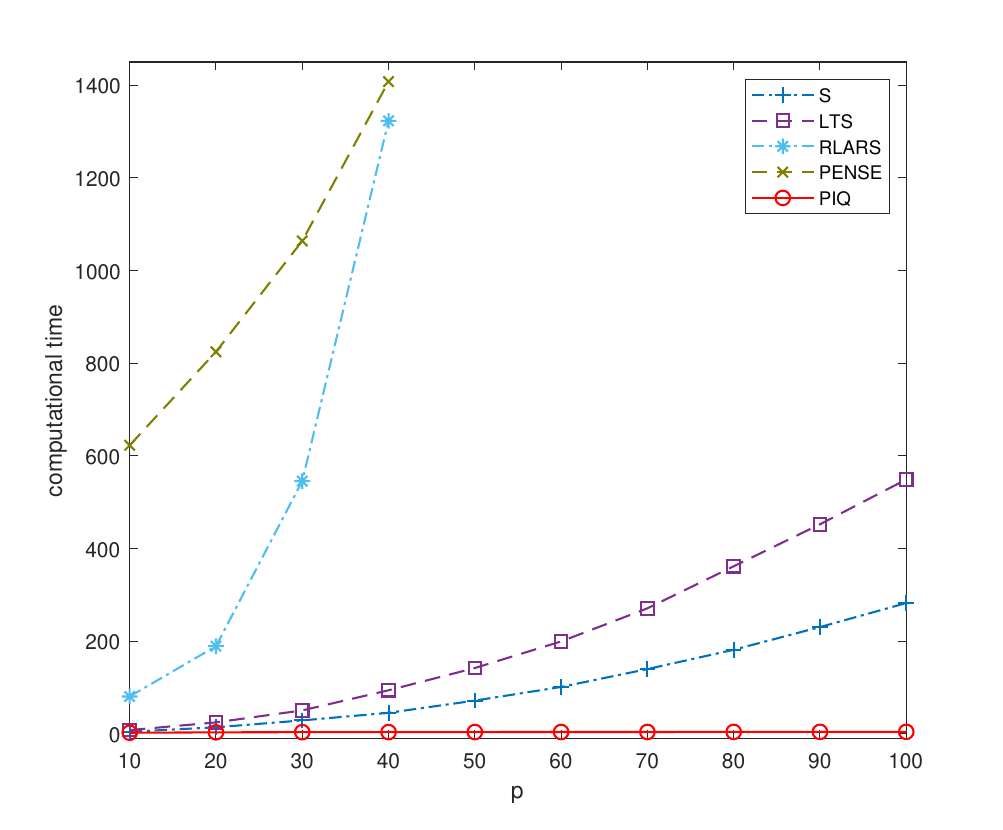}}
  \subfigure[classification]{\includegraphics[width=0.48\textwidth]{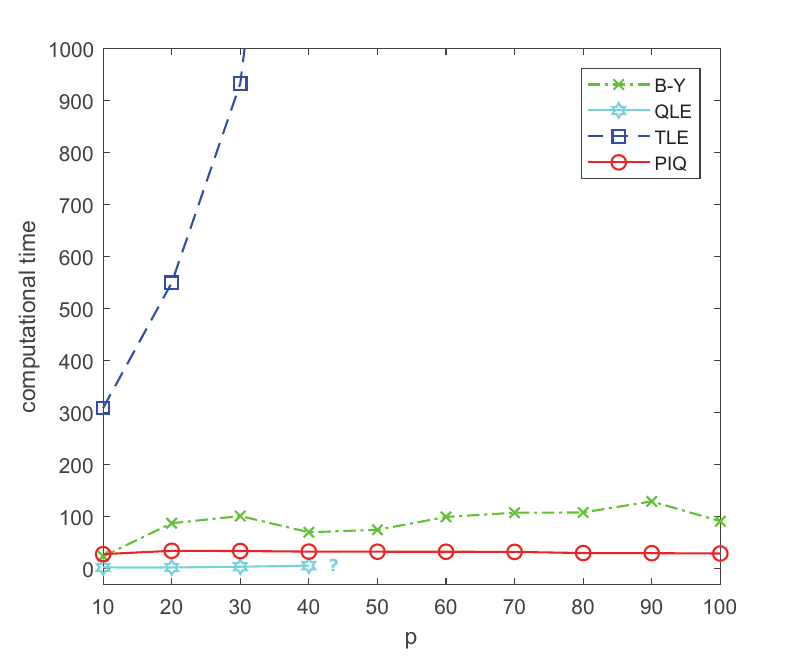}}
  \caption{\small Computational time comparison between different robust methods as $p$ increases. Left: regression  with $o^*=0.1n$ and $n=1000$ (RLARS and PENSE are too costly compared with the other methods and only parts of their cost curves are shown). Right: classification with $o^*=0.05n$ and $n=1000$ (when $p>40$, QLE delivered an error message  and could not continue).  } \label{figure:ex1-2}
\end{figure}

Table \ref{table:ex1-2} shows a performance comparison for robust regression and robust classification in Example 1 and Example 2, respectively.
To investigate algorithm scalability, we also give  a computational time comparison  in Figure \ref{figure:ex1-2} with respect to the dimension $p$. According to the table, with less than $5\%$ outliers,     all  methods  show  satisfactory results for   regression or classification. But as the percentage of outliers goes up to     $10\%$ (say) or more, the performance of many conventional   methods degrades significantly. Also,  some algorithms have poor numerical stability as $p$ is just over 40. In comparison,    PIQ is much more accurate and robust, and its scalability is evident.

\begin{table}[!h]
  \setlength{\tabcolsep}{3pt}\centering \footnotesize

{\footnotesize 
\caption{\small Performance comparison in $p>n$ settings (Example 3, Example 4 with $n=200, p =1000$). The outlier percentage varies from 5\% to 20\%.} \label{table:ex3-4}
\vspace{.1in}
  \begin{tabular}{lccc ccc ccc ccc ccc}
  \hline
  & \multicolumn{15}{c}{\textbf{Regression}}\\
  & \multicolumn{7}{c}{$o^*=10$}  & \multicolumn{7}{c}{$o^*=20$}   \\ \cmidrule(lr){2-8} \cmidrule(lr){10-16}

   & Err  & M$^{\gamma}$ & JD$^{\gamma}$ & M$^{\beta}$ & FA$^{\beta}$ & JD$^{\beta}$ & \textbf{T} & & Err  & M$^{\gamma}$ & JD$^{\gamma}$ & M$^{\beta}$ & FA$^{\beta}$ & JD$^{\beta}$ & \textbf{T}\\
  \hline
  QL      &  0.37    & 36  & 16 & 19 & 0.1  & 40 &  43  && 0.71 &  81 & 0 & 25 & 0.2 & 24 & 42  \\
  S-MTE   &  0.44    & 73  & 0  & 14 & 0.2  & 54 &  35  && 0.64 &  82 & 0 & 18 & 0.2 & 48 & 39 \\
  RLARS   &  0.53    & 82 & 4  & 25 & 0.2  & 40 &  42  && 1.06 &  89& 0 & 42 & 0.3 & 0  & 42  \\ 
  S-LTS   &  0.28    & 4  & 80  & 1  & 0.1  & 98 &  134  && 0.27 &  5 & 74& 1  & 0.2 & 96 & 134  \\
  \textbf{PIQ}   & 0.11 & 2 & 86 & 3 & 0.2 & 90 & 4  &&  0.23 &   5 & 80 & 10 & 0.2 &64 & 4 \\
  \midrule[1pt]
  & \multicolumn{7}{c}{$o^*=30$}  & \multicolumn{7}{c}{$o^*=40$}   \\ \cmidrule(lr){2-8} \cmidrule(lr){10-16}
   & Err  & M$^{\gamma}$ & JD$^{\gamma}$ & M$^{\beta}$ & FA$^{\beta}$ & JD$^{\beta}$ & \textbf{T} & & Err  & M$^{\gamma}$ & JD$^{\gamma}$ & M$^{\beta}$ & FA$^{\beta}$ & JD$^{\beta}$ & \textbf{T}\\
  \hline
  QL     & 0.89 &  81 & 0  &  33 & 0.2  &  8  & 45    && 1.10 &  76  & 0 &  41 & 0.2  & 4  & 50 \\
  S-MTE  & 0.82 &  79 & 0  &  28 & 0.2  &  26 & 42    && 0.97 &  75  & 0 &  34 & 0.2  & 14 & 54 \\
RLARS  & 1.52 &  91 & 0  &  58 & 0.3  &  0  & 41    && 1.55 &  85  & 0 &  57 & 0.3  & 0  & 40 \\
  S-LTS  & 1.24 &  79 & 10 &  29 & 1.4  &  16 &  132  && 1.52 &  82  & 0 &  46 & 1.5  & 0  & 136 \\
  \textbf{PIQ}& 0.34 & 3 &  78 & 16 & 0.3 &  46 & 4   && 0.62 &  5 &  76 & 27 & 0.3 & 22&4\\

  \hline\hline

  & \multicolumn{15}{c}{\textbf{Classification}}\\

  & \multicolumn{7}{c}{$o^*=10$}  & \multicolumn{7}{c}{$o^*=20$}   \\ \cmidrule(l){2-8} \cmidrule(l){10-16}
    & Err  & M$^{\gamma}$ & JD$^{\gamma}$ & M$^{\beta}$ & FA$^{\beta}$ & JD$^{\beta}$ & \textbf{T} & & Err  & M$^{\gamma}$ & JD$^{\gamma}$ & M$^{\beta}$ & FA$^{\beta}$ & JD$^{\beta}$ & \textbf{T}\\
  \hline
  enetLTS  & 12.4 & 4 & 96 & 1.3 & 0.1 &96 & 62  & &16.0 & 22 & 78 & 5.3 & 0.24 & 86 & 67  \\
  \textbf{PIQ}& 11.0 & 0 &100 & 0 & 0.1 & 100 & 16 & &11.6 &  0 & 100 & 0 &0.1 &100& 16  \\
  \midrule[1pt]
  & \multicolumn{7}{c}{$o^*=30$}  & \multicolumn{7}{c}{$o^*=40$}   \\ \cmidrule(l){2-8} \cmidrule(l){10-16}
  & Err  & M$^{\gamma}$ & JD$^{\gamma}$ & M$^{\beta}$ & FA$^{\beta}$ & JD$^{\beta}$ & \textbf{T} & & Err  & M$^{\gamma}$ & JD$^{\gamma}$ & M$^{\beta}$ & FA$^{\beta}$ & JD$^{\beta}$ & \textbf{T}\\
  \hline
  enetLTS & 36.2  &  94  &  6 & 43 & 1.0 & 16 &86 &    &  19.8 &   100 &0  & 66 & 0.9 &0 &96  \\
  \textbf{PIQ} & 12.6 &0 &  100 & 0.7 & 0.1 & 98 & 14  & & 17.0 &  8 & 84 & 9 &0.1 & 78 &14\\
  \hline
  \end{tabular}
}
\end{table}

The advantages of PIQ are even more impressive in Table \ref{table:ex3-4} which summarizes the high-dimensional results. In particular, its boost in accuracy and outlier-resistance is not accompanied by a sacrifice in computational time relative to  other robust methods.  PIQ offered substantial time savings: its computational cost was at most about  one fourth of that of the other algorithms.

We also conducted experiments with different correlation strengths, covariance structures and sparsity levels, and similar conclusions can be drawn from the performance comparison; see Appendix \ref{sec:add_exp} for details.

\subsection{Robust classification of spam email}

The email spam dataset (available at  \url{ftp.ics.uci.edu}) contains      57 predictors  including   relative frequencies of 54 most commonly occurring words and characters and 3 statistics of   capital letters.
 The total number of emails is 4,601,   of which   1,813 are spam ($y=1$) and 2,788 are non-spam ($y=0$).
We  would like to know   whether accommodating   \textit{mislabeled} emails (if any) could lead to improved   spam/non-spam classification.

We applied PIQ with  the  logistic deviance   and the Huberized hinge loss \citep{Chapelle2007},  denoted by PIQ-L and PIQ-H, respectively,   with  PIC for tuning.
 Table \ref{spam_test} compares PIQ   with    logistic regression, B-Y, TLE, and QLE   (cf.  Section \ref{subsec:simu}) in terms of misclassification error and   F1-score \citep{FScore}, both averaged over 100 data splits ($70\%$ for training and $30\%$ for   test).
      From the table, TLE and QLE  failed  due to the      large   sample size.     PIQ outperforms the other methods and the choice of the loss does not affect  its performance on this dataset.

\begin{table}[h!]
        \centering
        \small
        \caption{\small Email spam data:  misclassification error rates  and F1-scores.}
        \vspace{2ex}
        \begin{tabular}{lcccc}
                \hline
                & Miscls (\%)   & F1 (\%)    \\
                \hline
                Logistic   & 7.5 & 90.3  \\ 
                B-Y        & 7.0   &  89.4  \\
                TLE  & -- &  -- \\
                QLE  & -- &  -- \\
                PIQ-L      & 6.8  &  91.3   \\
                PIQ-H       & 6.8  &  91.3  \\
                \hline
        \end{tabular}
        \label{spam_test}
\end{table}

We then studied  the subset of outliers identified by PIQ. Take PIQ-L for instance: there are  73 nonzeros $\hat\gamma_i$'s,   the smallest magnitude being   10.97. To verify if there is indeed a distinction between  the  outliers ($o=1$) and clean    observations ($o=0$), we examined 6 indicative features,   {\lmtt free}, {\lmtt !},  {\lmtt \$}, {\lmtt hp},  {\lmtt 650}  and  {\lmtt cap\_avg}. The first three   words and signs are often used in spam to attract  people's attention; {\lmtt hp} and   {\lmtt 650} are related to the data creator's working company and area code; {\lmtt cap\_avg} measures the average length of uninterrupted sequences of capital letters.
Given each feature, we made   a $2\times 2$ table  with its rows corresponding to    spam and non-spam ($y=1$, $0$) and columns  corresponding to outlying and  clean ($o =1$, $0$).  The cross-product ratio or    odds ratio    visualizes the strong association between $y$ and $o$, as shown in   Figure \ref{fig:spam_odds_ratio}. Indeed,     even the odds ratios   closer to 1   are either as small as     $\exp(-2.8)=  0.06$ or as large as $\exp(4.89)=  133$, demonstrating  a strong discrepancy between the outlier subset and the clean subset.  We also built  an ANOVA model on each feature to test the significance of $o$;  all the obtained $p$-values are below 1e-10, a clear evidence of  the outlyingness.

The data analysis suggests that  even on some conventional datasets, mislabeling frequently occurs and   guarding against   outliers is beneficial. It is fair to say that modern statistical analysis involves more and more large-scale datasets, but with bigger data also come more junk and errors.  This could be a severe issue  for binary classification, since any inaccuracy in  labeling simply means  the response goes to the other extreme. Applying a classifier designed under no consideration of gross outliers could be risky  in reality, and our  resistant learning framework   offers an effective  fix in this regard.

\begin{figure}[h!]
        \centering
        \includegraphics[width=0.55\textwidth]{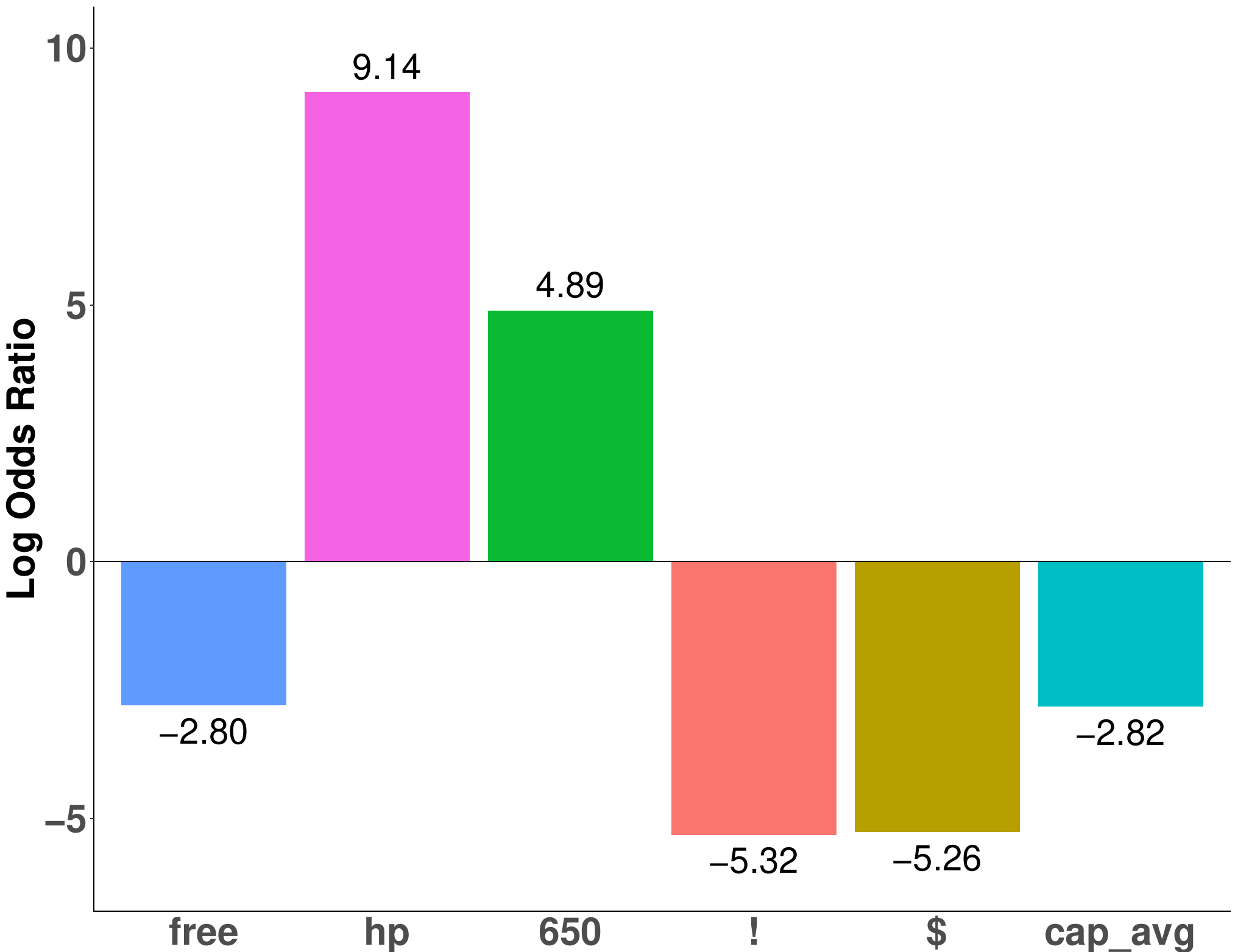}
        \caption{\small Log odds ratios of six $2\times 2 $ tables with row  variable $y $ and column  variable $o$.} \label{fig:spam_odds_ratio}
\end{figure}


\subsection{Robust neural network for Parkinson's disease detection}
Neural networks have attracted a great deal of attention in machine learning. In this experiment, we adapt the PIQ technique to a neural net model and test its performance for Parkinson's disease (PD) detection. There has been a surge of interest in speech analysis of PD recently, as voice recordings from these patients tend to show typical patterns like aperiodic vibrations.
We use the dataset from \cite{little2007exploiting} which collects a range of biomedical voice measurements from normal people and patients with PD, where the 22 predictors for 195 recordings include various measures of vocal fundamental frequency and variation in amplitude, as well as some nonlinear metrics constructed from raw voice recordings.
We randomly split the dataset into a training subset (60\%) and a test subset (40\%).

Previous studies suggest the presence of nonlinear effects and interactions of the features in this challenging classification task (\cite{little2007exploiting}, \cite{sakar2013collection}). Hence we designed a feedforward neural network that consists of an input layer, two dense layers with 22 and 12 nodes respectively, and one output layer. The dense layers deploy the popular rectified linear units (ReLU) \citep{Hinton2010relu}, and the output node uses a softmax activation.
We trained the network for 200 epochs and attained a misclassification error rate of 9.1\% on the test data---in contrast, an SVM with a nonlinear RBF kernel gave 17.9\%. (Meanwhile, classical robust classification methods like B-Y, QLE, and TLE all behaved poorly, with the misclassification error rates $>25\%$.)

On the other hand, the noisy signals and potential outliers occurring suggest the need of robustification. But the overall network criterion barely  resembles  that of regression or logistic regression due to  its substantial nonlinearity and hierarchy,  and how to reweight the samples in this sophisticated  model to limit the influence of outliers sounds  tricky.
Following the additive robustification scheme, we simply added sample-indicator features into the input layer (cf. Figure \ref{figure:nn-plot}), accompanied by a group $\ell_0 + \ell_2$ regularization ($q = 0.05n$, $\nu =$ 1e-4) to correct  for sample outlyingness   in constructing predictive factors.
Our optimization-based PIQ algorithm  can be seamlessly incorporated into the process of back-propagation.
Though a simple modification, the robustified neural network gave an impressive  misclassification error rate 6.4\%, an improvement of about 30\% over that of the vanilla neural network.

\begin{figure}[!ht]
  \centering{}
 \includegraphics[width=0.75\textwidth]{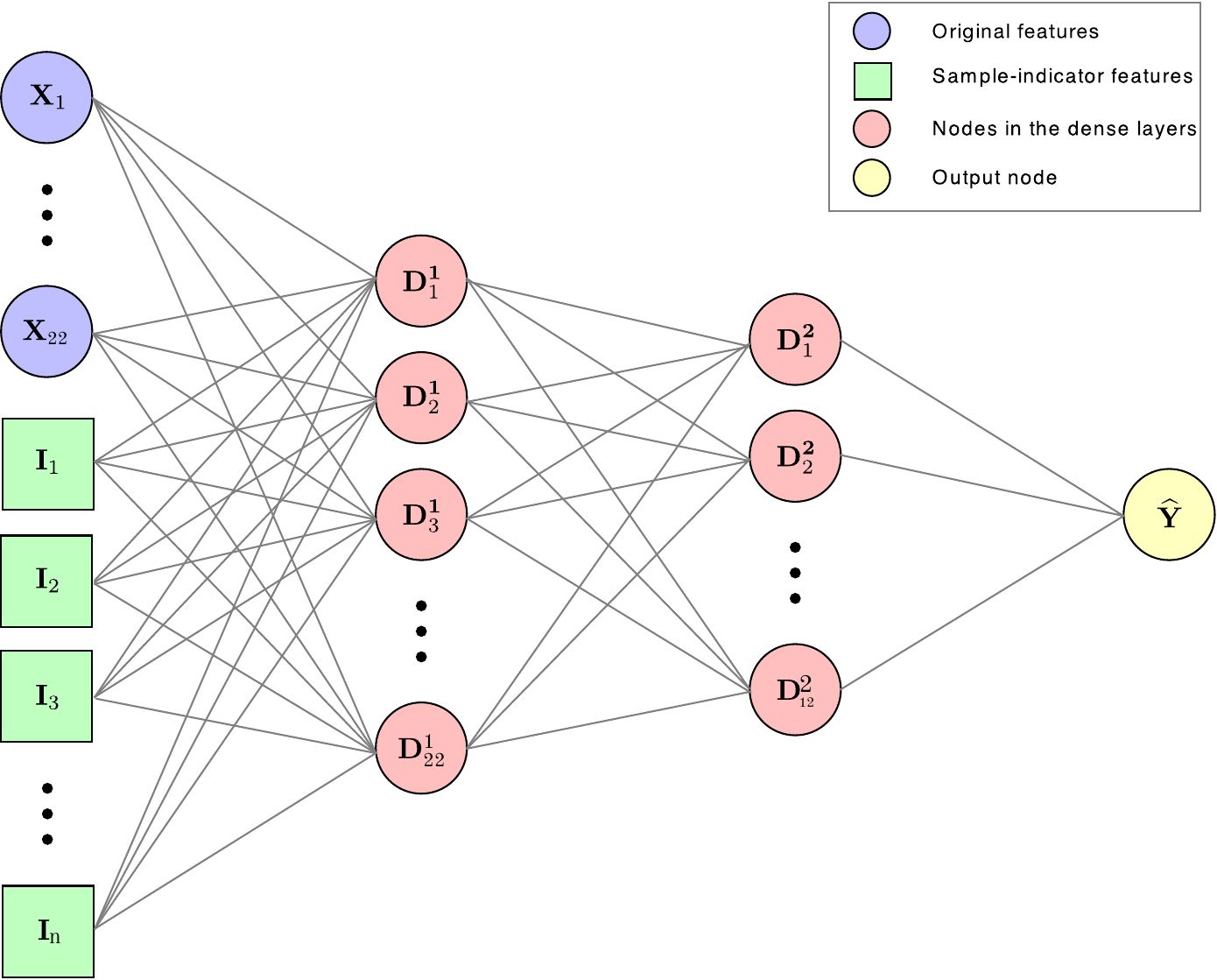}
  \caption{\small Structure of the robustified neural network on Parkinson's disease data, where sample-indicator features (denoted by green square nodes) are added to the input layer.} \label{figure:nn-plot}
\end{figure}

\section{Summary}
The work studied how to gain outlier resistance for a given   loss that is not necessarily quadratic and may go beyond the likelihood setup.
Motivated by the method of trimming, we proposed  an  $\ell_0$+$\ell_2$  regularized learning framework. We  showed that  pursuing its (global) minimizers   can cope with extreme outliers regardless of  leverage and outlyingness,  but  on large-scale data, there is always a   balance between statistical accuracy   and computational efficiency. Hence   a more meaningful and deeper question is how to  design a scalable algorithm to    alleviate    the starting point requirement and the regularity conditions to obtain the desired order of statistical accuracy.

In this work, the interplay between high dimensional statistics  and nonconvex optimization led to fruitful   results  for resistant learning on the   learning rate, cardinality control,  the choice of regularization parameters, and so on. In particular,   a novel progressive backward optimization scheme is brought forward which, unlike subset sampling, is cost effective, and can significantly improve the statistical performance of iterative quantile-thresholding.
We hope that the work could     raise  concern about    outliers among (big) data analysts,  as well as providing a sensible  and universal means to boosting   resistance  in regression, classification and more sophisticated learning tasks.

%


\appendix
\numberwithin{equation}{section}
\numberwithin{thm}{section}
\numberwithin{table}{section}
\section{Proofs} \label{sec:proofs}

\subsection{Proof of Theorem \ref{th:criterion_equivalence}}

First, we introduce a lemma to be used for proving the three statements.
\begin{lemma}  \label{lemma:gamma_optimal}
Let $l_0(\eta;y)$ be a     loss function defined on $\eta\in \mathbb R$ with $y\in \mathcal  Y \subset\mathbb R$  as a parameter and
$ \inf_{\eta}l_0(\eta;y)=L  >-\infty \mbox{ for all } y \in \mathcal Y$. 
 Define $\bsbe = \bsbX \bsbb$  and   $l(\bsbe;\bsby) = \sum_{i=1}^n l_0(\eta_i;y_i)$. Given  any $\bsbb $, let    $l_0^{(1)}(\bsbb) \leq \ldots \leq l_0^{(n)}(\bsbb)$ be the   order statistics of   $l_0(\bsbx_i^T\bsbb; y_i)$ ($1\le i \le n$). Then\begin{align} \label{eq:gamma_optimal}
&\inf_{  \|\bsbg\|_0 \leq q}l(\bsbX \bsbb + \bsbg; \bsby) = \sum_{i=1}^{n-q} l_0^{(i)}(\bsbb).
\end{align}
Moreover, given any $\tau\ge 0$,
\begin{align}\label{eq:gamma_optimal2}
&\inf_{   \bsbg }\, l(\bsbX \bsbb + \bsbg; \bsby) + \tau \| \bsbg\|_0 = \sum_{i=1}^{n} (\tau +L)\wedge l_0(\bsbx_i^T\bsbb; y_i).
\end{align}
\end{lemma}

For a),  let  $(\hat \bsbb, \hat \bsbg)$ be a globally optimal solution to\begin{equation} \label{eq:criterion}
\min_{(\bsbb,   \bsbg)} l(\bsbe;\bsby) \mbox{ s.t. } \bsbe = \bsbX \bsbb + \bsbg, \mbox{ } \|\bsbg\|_0 \leq q.
\end{equation}
Then for any $\bsbb$,
$
l(\bsbX\hat \bsbb + \hat\bsbg; \bsby) \leq  l(\bsbX \bsbb + \hat\bsbg_\bsbb; \bsby) ,
$
where $\hat \bsbg_\bsbb$ is an optimal solution of $\bsbg$   given $\bsbb$.
From Lemma \ref{lemma:gamma_optimal}, $
l(\bsbX\hat \bsbb + \hat\bsbg; \bsby) = \sum_{i=1}^{n-q} l_0^{(i)}(\hat \bsbb), \mbox{ } l(\bsbX \bsbb + \hat\bsbg_\bsbb; \bsby) = \sum_{i=1}^{n-q} l_0^{(i)}(\bsbb),
$
and so $
\sum_{i=1}^{n-q} l_0^{(i)}(\hat \bsbb) \leq \sum_{i=1}^{n-q} l_0^{(i)}(\bsbb).
$

Now suppose that  $\hat \bsbb$ is a globally optimal solution that minimizes the trimmed criterion on the RHS of \eqref{eq:gamma_optimal}. Then for any $(\bsbb, \bsbg)$,
\begin{align*}
l(\bsbX \bsbb + \bsbg; \bsby) & \geq l(\bsbX \bsbb + \hat \bsbg_\bsbb; \bsby)  = \sum_{i=1}^{n-q} l_0^{(i)}(\bsbb) \geq\sum_{i=1}^{n-q} l_0^{(i)}(\hat \bsbb)  = l(\bsbX \hat \bsbb + \hat \bsbg_{\hat\bsbb}; \bsby),
\end{align*}
which means $(\hat \bsbb; \hat \bsbg_{ \hat \bsbb})$ is a global minimizer of \eqref{eq:criterion}.
The second conclusion for the penalized form and winsorized form can be proved similarly (assuming $L=0$).

To prove the last statement,  we set $  \mathcal C(\bsbb) =    \{i:  l_0(\bsbx_i^T \bsbb; y_i)> \tau , 1\le i \le n\} $ and $c(\bsbb)= | \mathcal C(\bsbb)|$. So  $\sum_{i=1}^{n}  \tau  \wedge l_0(\bsbx_i^T\bsbb; y_i)=c(\bsbb) \tau + \sum_{i=1}^{n-c(\bsbb)} l_0^{(i)}( \bsbb)$.

We claim that for any $\hat\bsbb\in \arg\min  \sum  \tau \wedge l_0(\bsbx_i^T\bsbb; y_i)$   with $q =c(\hat \bsbb)$, it is also an optimal solution to $\min  \sum_{i=1}^{n-q} l_0^{(i)}(\bsbb)$. Prove by contradiction: if $  \sum_{i=1}^{n-q} l_0^{(i)}(\hat \bsbb)>  \sum_{i=1}^{n-q} l_0^{(i)}(\bsbb)$ for some $\bsbb$, then since $\tau  \wedge l_0(\bsbx_i^T\hat \bsbb; y_i) = \tau\ge \tau  \wedge l_0(\bsbx_j^T  \bsbb; y_j)$ for all $i \in  \mathcal C(\hat \bsbb), j: 1 \le j\le n $, we would get   $\sum_{i=1}^{n}  \tau  \wedge l_0(\bsbx_i^T\hat \bsbb; y_i)> \sum_{i=1}^{n}  \tau  \wedge l_0(\bsbx_i^T\bsbb; y_i)$.

Finally,  we show that given  any $(\hat \bsbb, \hat \bsbg)\in \arg \min
l(\bsbX \bsbb + \bsbg; \bsby) + \tau \| \bsbg\|_0 $ with $q = \|\hat \bsbg\|_0$ (where $\hat \gamma_i$ are allowed to take $\pm \infty$), it is also an optimal solution to
$\min_{\bsbb , \bsbg } l(\bsbX \bsbb + \bsbg;\bsby)     \mbox{ s.t. } \|\bsbg\|_0 \leq q$.  Suppose that  there exists $(\bsbb, \bsbg)$ with $\|\bsbg\|_0\le q$ so that  $  l(\bsbX \hat \bsbb + \hat \bsbg;\bsby)>  l(\bsbX   \bsbb +   \bsbg;\bsby)$. Then, by considering   $(\bsbb, \bsbg')$ with $\bsbg' = \hat \bsbg_\bsbb$ and $\|\bsbg'\|_0\le q$ as defined previously, $\inf_{  \|\bsbg\|_0 \leq q}l(\bsbX \hat \bsbb +\bsbg; \bsby)=l(\bsbX \hat \bsbb + \hat \bsbg;\bsby) >\inf_{  \|\bsbg\|_0 \leq q}l(\bsbX   \bsbb +\bsbg; \bsby)  $, and thus   $  \sum_{i=1}^{n-q} l_0^{(i)}(\hat \bsbb)>  \sum_{i=1}^{n-q} l_0^{(i)}(\bsbb)$ from Lemma \ref{lemma:gamma_optimal}. Using the fact that $\hat \gamma_i\ne 0$ implies $\tau  + l_0(\bsbx_i^T\hat \bsbb+\hat \gamma_i; y_i) =\tau$, and the construction of $\bsbg' $,  $l(\bsbX \hat \bsbb + \hat\bsbg; \bsby) + \tau \|\hat  \bsbg\|_0 = q\tau + \sum_{i=1}^{n-q} l_0^{(i)}(\hat \bsbb) $ and $l(\bsbX \bsbb +  \bsbg'; \bsby) + \tau \|   \bsbg'\|_0 = q\tau + \sum_{i=1}^{n-q} l_0^{(i)}(  \bsbb) $. Summarizing the above  gives $l(\bsbX \hat \bsbb + \hat \bsbg;\bsby)> l(\bsbX \bsbb + \bsbg';\bsby)$, contradicting the optimality of  $(\hat \bsbb, \hat \bsbg)$. The proof is complete.

\paragraph{Proof of Lemma \ref{lemma:gamma_optimal}}

Recall that   the loss is assumed  bounded from below. Without loss of generality,  let
$\inf_{\eta}l_0(\eta;y) = 0 \mbox{ for all } y \in \mathcal Y$ (otherwise, we can redefine the loss function by $l_0(\eta; y) - \inf_\eta l_0(\eta;y)$).

Given $\bm\beta$,  denote by $\hat \bsbg_\bsbb$ a minimizer of $l(\bm X\bm\beta + \bm\gamma; \bm y)$ subject to $\|\bm\gamma\|_0 \le q$, whose components   can take $\pm \infty$. Then if $\hat \gamma_i = 0$,  $ l_0(\bsbx_i^T  \bsbb + \hat \gamma_i; y_i) = l_0(\bsbx_i^T  \bsbb; y_i) \ge 0$. If $\hat \gamma_i \neq 0$, by varying its value, which will not   violate the constraint, one can argue that $l_0(\bsbx_i^T  \bsbb + \hat \gamma_i; y_i) =\inf_{\eta}l_0(\eta;y_{i}) = 0$. 
Therefore, $l(\bsbX \bsbb + \hat \bsbg_{\bsbb}; \bsby) = \sum_{i=1}^{n-q} l_0^{(i)}(\bsbb)$.
The proof of \eqref{eq:gamma_optimal2} follows similar lines.
\subsection{Proof of Theorem \ref{thm:BCD_linear_rate}}
\label{proof:resistregalg}
First, we have the following result to argue the (alternative) optimality of   $\bsbg^{(t+1)}$ given $\bsbb^{(t)}$, cf.  Lemma C.1 in \cite{She2013Super}. 

\begin{lemma} \label{lemma:Theta_optimality}
Given any $\bsb s \in \mathbb R^n$, $0 \leq q \leq n$, $\nu \geq 0$, $\hat \bsbxi = \Theta^{\#}(\bsb s; q, \nu)$ is a global optimal solution to $\min_{\bsbxi \in \mathbb R^n} \|\bsb s - \bsbxi\|_2^2/2 + \nu \|\bsbxi\|_2^2/2 \mbox{ s.t. }\|\bsbxi\|_0 \leq q.$
\end{lemma}

The  below definition of a sub-Gaussian random vector is standard in the literature, see, e.g., \cite{Vershynin2012}.

\begin{defn}[\textbf{Sub-Gaussian random variable/vector}]\label{def:subgauss}
We call $\xi$ a sub-Gaussian(0, $\sigma^2$) random variable if and only if $\xi$ has mean 0 and there exist constants $C, c>0$ such that $\mathbb P\{|\xi|\geq t\} \leq C e^{-c t^2}, \forall t>0$ with the scale (or $\psi_2$-norm) of $\xi$  defined by $\sigma( \xi) =
\inf \{\sigma>0: \mathbb E\exp(\xi^2/\sigma^2) \leq 2\}$. More generally, $\bm\xi\in \mathbb R^p$ is called a sub-Gaussian random vector with scale  bounded by $\sigma$ if all one-dimensional marginals $\langle \bm\xi, \bm\alpha \rangle$ are sub-Gaussian satisfying $\|\langle \bm\xi, \bm\alpha \rangle\|_{\psi_2}\leq \sigma \|\bm\alpha \|_2$, $\forall \bm\alpha\in \mathbb R^{p}$. 

\end{defn}

a)   Using the  update in \eqref{eq:update_rule_BCD}, we can represent the optimal $\bsbb$   as a function of $\bsbg$, from which    problem    \eqref{reg:criterion} is equivalent to
\begin{equation}
\min_{\bm\gamma} f(\bm\gamma):= \frac{1}{2}\|(\bm I-\bm H)(\bm\gamma-\bm y)\|_2^2 + \frac{\nu}{2}\|\bm\gamma\|_2^2 \text{ s.t. }\|\bm\gamma\|_0\le q.
\end{equation}
Let $l(\bm\gamma) = \|(\bm I-\bm H)(\bm\gamma-\bm y)\|_2^2/2$. From Lemma \ref{lemma:Theta_optimality}, a nice fact of  \eqref{eq:update_BCD_line} is that     $$\bm\gamma^{(t+1)} = \arg\min_{\|\bm\gamma\|_0\le q}g(\bm\gamma;\bm\gamma^{(t)}),$$ where
\begin{align*}
g(\bm\gamma;\bm\gamma^-) & = \|(\bm I-\bm H)(\bm\gamma^--\bm y)\|_2^2/2 + \langle (\bm I-\bm H)(\bm\gamma^--\bm y), \bm\gamma-\bm\gamma^-\rangle \\ & \quad+ \|\bm\gamma-\bm\gamma^-\|_2^2/2 + \nu\|\bm\gamma\|_2^2/2\\ & =    f(\bm\gamma) + (\mathbf D_2-\bm\Delta_l)(\bm\gamma,\bm\gamma^-)\end{align*} with $\bm\Delta_l$ and $\mathbf D_2$ introduced in Section \ref{sec:theory}. By definition,  $g(\bm\gamma^{(t)};\bm\gamma^{(t)})\ge g(\bm\gamma^{(t+1)};\bm\gamma^{(t)})$, and so we have
\begin{equation}\label{BCD_linear_rate-2}
f(\bm\gamma^{(t)}) - f(\bm\gamma^{(t+1)}) \ge (\mathbf D_2-\bm\Delta_l)(\bm\gamma^{(t+1)},\bm\gamma^{(t)}).
\end{equation}
Summing up \eqref{BCD_linear_rate-2} over $t =  1,\ldots,T$ gives the conclusion.
\\

\noindent
b) We begin with a useful lemma. Throughout the proof, given any $\bsbb$, we use $\mathcal J(\bsbb)$ to denote its support, i.e., $\mathcal J(\bsbb)=\{j: \beta_j \ne 0\}$, and $J(\bsbb ) = |\mathcal J(\bsbb)|$. Recall that $s^* = J(\bsbb^*)$ and $o^* = J(\bsbg^*)$.

\begin{lemma} \label{lemma:bound without design}
  Given $q\le p$ and $\nu\ge 0$, a globally optimal solution to the optimization problem $\min_{\bm\beta\in\mathbb R^p}l(\bm\beta) = \|\bm y - \bm\beta\|_2^2/2 + \nu\|\bm\beta\|_2^2/2$ s.t. $\|\bm\beta\|_0\le q$ is given by  $\hat{\bm\beta}=\Theta^{\#}(\bm y; q, \nu)$. Let $\mathcal J = \mathcal J(\bm\beta)$, $\hat{\mathcal J} = \mathcal J(\hat{\bm\beta})$ and assume      $J(\hat{\bm\beta}) = q$. Then, for any $\bm\beta$ with $J(\bm\beta)\le s = q/\vartheta$ and $\vartheta\ge 1$,   $$l(\bm\beta)-l(\hat{\bm\beta}) \ge \{1-\mathcal L(\mathcal J, \hat{\mathcal J})\}(1+\nu)\|\hat{\bm\beta}-\bm\beta\|_2^2/2,$$ where $\mathcal L(\mathcal J,\hat{\mathcal J}) = (|\mathcal J \backslash \hat{\mathcal J}| / |\hat{\mathcal J}\backslash \mathcal J|)^{1/2} \le (s/q)^{1/2} = \vartheta^{-1/2}$.
\end{lemma}

It can be proved by Lemma 9   in \cite{She2017RRRR}. 
By Lemma \ref{lemma:bound without design}, for any $\bsbg$ satisfying $\| \bsbg\|_0\le q $, $$g(\bm\gamma;\bm\gamma^{(t)}) - g(\bm\gamma^{(t+1)};\bm\gamma^{(t)}) \ge \{1-\mathcal L(\mathcal J(\bm\gamma),\mathcal J(\bm\gamma^{(t+1)}))\}(1+\nu)\mathbf D_2(\bm\gamma^{(t+1)},\bm\gamma),$$ and so
\begin{equation*}
\begin{split}
&f(\bm\gamma)+(\mathbf D_2-\bm\Delta_l)(\bm\gamma,\bm\gamma^{(t)}) - f(\bm\gamma^{(t+1)}) - (\mathbf D_2-\bm\Delta_l)(\bm\gamma^{(t+1)},\bm\gamma^{(t)})\\
\ge\,& \{1-\mathcal L(\mathcal J(\bm\gamma),\mathcal J(\bm\gamma^{(t+1)}))\}(1+\nu)\mathbf D_2(\bm\gamma^{(t+1)},\bm\gamma),
\end{split}
\end{equation*}
or
\begin{equation*}
\begin{split}
&f(\bm\gamma)+((1+\nu)\mathbf D_2-\bm\Delta_f)(\bm\gamma,\bm\gamma^{(t)}) - f(\bm\gamma^{(t+1)}) - (\mathbf D_2-\bm\Delta_l)(\bm\gamma^{(t+1)},\bm\gamma^{(t)})\\
\ge\,& \{1-\mathcal L(\mathcal J(\bm\gamma),\mathcal J(\bm\gamma^{(t+1)}))\}(1+\nu)\mathbf D_2(\bm\gamma^{(t+1)},\bm\gamma).
\end{split}
\end{equation*}
Substituting  $\bm\gamma^*$ for $\bsbg $ gives
\begin{equation} \label{BCD_linear_rate-0}
\begin{split}
&\{1-\mathcal L(\mathcal J(\bm\gamma^*),\mathcal J(\bm\gamma^{(t+1)}))\}(1+\nu)\mathbf D_2(\bm\gamma^{(t+1)},\bm\gamma^*) + (\mathbf D_2-\bm\Delta_l)(\bm\gamma^{(t+1)},\bm\gamma^{(t)})\\
\le\,& -\bm\Delta_f(\bm\gamma^{(t+1)},\bm\gamma^*) - \langle\nabla f(\bm\gamma^*), \bm\gamma^{(t+1)}-\bm\gamma^*\rangle +((1+\nu)\mathbf D_2-\bm\Delta_f)(\bm\gamma^*,\bm\gamma^{(t)}).
\end{split}
\end{equation}
Since  $\|\bsbI -\bsbH\|_2 = 1$,  $   (\mathbf D_2-\bm\Delta_l)(\bm\gamma^{(t+1)},\bm\gamma^{(t)}) \ge 0$. Together with  $$\mathcal L(\mathcal J(\bm\gamma^*), \allowbreak\mathcal J(\bm\gamma^{(t+1)})) \le 1/\sqrt{\vartheta},$$ we obtain
\begin{equation} \label{BCD_linear_rate-3}
\begin{split}
&((1-1/\sqrt{\vartheta})(1+\nu)\mathbf D_2 + \bm\Delta_l +\nu \mathbf D_2)(\bm\gamma^{(t+1)},\bm\gamma^*)\\
\le\,& (\mathbf D_2-\bm\Delta_l)(\bm\gamma^*,\bm\gamma^{(t)}) + \langle (\bm I-\bm H)^2(\bm\gamma^*-\bm y) + \nu\bm\gamma^*, - \bm \gamma^{(t+1)}+\bm\gamma^* \rangle\\
=\,& (\mathbf D_2-\bm\Delta_l)(\bm\gamma^*,\bm\gamma^{(t)}) + \langle- \bm U_\perp \bm U_\perp^T\bm\epsilon + \nu\bm\gamma^*, \bm\gamma^*  - \bm\gamma^{(t+1)}\rangle\\
=\,& (\mathbf D_2-\bm\Delta_l)(\bm\gamma^*,\bm\gamma^{(t)}) + \langle \bm\epsilon', \bm U_\perp^T (   \bm\gamma^{(t+1)} - \bm\gamma^*) \rangle + \nu\langle\bm\gamma^*,  \bm\gamma^*  - \bm\gamma^{(t+1)} \rangle,\\
\end{split}
\end{equation}
where $\bm X=[\bm U, \bm U_\perp]\bm D\bm V^T$ is the singular value decomposition of $\bm X$ with $\bm U_\perp\in\mathbb R^{n\times (n-p)}$ and $\bm\epsilon'=\bm U_\perp^T\bm\epsilon$  which is still a sub-Gaussian random vector with mean 0 and scaled bounded by $\sigma$.
From   $\bsbI - \bsbH = \bsbU_{\perp} \bsbU_{\perp}^T$,   $\Delta_l (\bm\gamma,\bm\gamma') = \| \bsbU_{\perp}^T (\bsbg - \bsbg')\|_2^2/2$.

The stochastic term $\langle\bm\epsilon', \bm U_\perp^T(\bm\gamma^{(t+1)}-\bm\gamma^*)\rangle$ can be bounded using the following lemma.
\begin{lemma}  \label{lemma:noise_bound}
Let $\bm\epsilon$ be a sub-Gaussian random vector with mean zero and scale bounded by $\sigma$. Given $\bm X\in\mathbb R^{n\times p}$ and $0\le J \le p$, there exist universal constants $A, C, c > 0$ such that for any $a>0$, the following event
\begin{equation}
\sup_{\bsbb\in\mathbb R^p:\|\bsbb\|_0\le J}\langle\bm\epsilon, \bm X\bsbb\rangle - \frac{1}{2a}\|\bm X\bsbb \|_2^2 - aA^2\sigma^2J\log(ep/J) > a\sigma^2t
\end{equation}
occurs with probability at most $C\exp(-ct) p^{ -cA^2 }$ for any $t\ge 0$.
\end{lemma}
 The result is a variant of Lemma 4 of  \cite{She2016} 
  and can be obtained from its proof.
According to Lemma \ref{lemma:noise_bound}, there exist constants $A,C,c>0$ such that
for any $a>0$,
\begin{equation} \label{BCD_linear_rate-4}
\langle\bm\epsilon', \bm U_\perp^T(\bm\gamma^{(t+1)}-\bm\gamma^*)\rangle \le \frac{1}{a}\bm\Delta_l(\bm\gamma^{(t+1)},\bm\gamma^*) + aA^2\sigma^2(1+\vartheta)o^*\log\frac{en}{o^* }
\end{equation}
occurs with probability at least $1-C (n-p)^{-c}$. Since $\langle\bm\gamma^*, -\bm\gamma^{(t+1)}+\bm\gamma^* \rangle \le (b/2)\|\bm\gamma^*\|_2^2 + (1/2b)\|\bm\gamma^{(t+1)}-\bm\gamma^*\|_2^2$ for any $b>0$,
combining  \eqref{BCD_linear_rate-3} and \eqref{BCD_linear_rate-4}  gives
\begin{equation} \label{BCD_linear_rate-5}
\begin{split}
&\{((1-\frac{1}{\sqrt{\vartheta}}) + (2-\frac{1}{\sqrt{\vartheta}}-\frac{1}{b})\nu)\mathbf D_2+(1- \frac{1}{a})\bm\Delta_l\}(\bm\gamma^{(t+1)},\bm\gamma^*) \\
\le\,& (\mathbf D_2-\bm\Delta_l)(\bm\gamma^{(t)},\bm\gamma^*) + aA^2\sigma^2(1+\vartheta)o^*\log\frac{en}{o^*}  + \frac{b\nu}{2}\|\bm\gamma^*\|_2^2.
\end{split}
\end{equation}
Let
\begin{align}
\kappa = \frac{1}{1-\varepsilon} \big\{(2- \frac{1}{a})\varepsilon - \frac{1}{\sqrt \vartheta}   + (2 - \frac{1}{\sqrt \vartheta} - \frac{1}{b})\nu\big\}. \label{kappadefproof}
\end{align}
By the regularity condition,
\begin{equation*}
\begin{split}
& (1+ \kappa)(\mathbf D_2-\bm\Delta_l) (\bm\gamma^{(t+1)},\bm\gamma^*) \\
\le\, & \{((1-\frac{1}{\sqrt{\vartheta}} ) + (2-\frac{1}{\sqrt{\vartheta}}-\frac{1}{b})\nu)\mathbf D_2+ (1-\frac{1}{a})\bm\Delta_l\}(\bm\gamma^{(t+1)},\bm\gamma^*)\\
\le\,& (\mathbf D_2-\bm\Delta_l)(\bm\gamma^{(t)},\bm\gamma^*) + aA^2\sigma^2(1+\vartheta)o^*\log\frac{en}{o^*}  + \frac{b\nu}{2}\|\bm\gamma^*\|_2^2.
\end{split}
\end{equation*}
Finally, by a recursive argument, we have

\begin{equation*}
\begin{split}
&(\mathbf D_2-\bm\Delta_l)(\bm\gamma^{(t)},\bm\gamma^*)\\
\le\,& \Big(\frac{1}{\kappa+1}\Big)^{t-1} (\mathbf D_2-\bm\Delta_l)(\bm\gamma^{(1)},\bm\gamma^*) + \frac{a}{\kappa}A^2\sigma^2(1+\vartheta)o^*\log\frac{en}{o^*}  + \frac{b\nu}{2\kappa}\|\bm\gamma^*\|_2^2.
\end{split}
\end{equation*}
Choosing $a =1  $, $b =1 $, which ensures  $1+\kappa \ge 0$, we obtain the conclusion on the sequence of $\bsbg$-iterates.

In fact, according to   \eqref{kappadefproof}, we may choose any $a\ge 1$ and $b>0$, and so   the linear convergence rate holds if $\varepsilon > 1/(2\sqrt \vartheta)-(2 -  {1}/({\sqrt \vartheta}) \nu$.

To show the conclusion on the sequence of $\bsbb$-iterates, notice that for any $t\ge 1$,
$$\bsbX \bsbb^{(t)}  - \bsbX \bsbb^* +  \bsbH  \bsbg^{(t)} -\bsbH  \bsbg^*  = \bsbH \bsbeps. $$
and
$$
\EP [ \|\bsbH \bsbeps\|_2^2 \ge c r(\bsbX) \sigma^2  + t\sigma^2] \le C \exp(-c [\{t^2/r(\bsbX)\} \wedge t]), $$
where $r(\bsbX)$ denotes the rank of $\bsbX$.
The last inequality can be shown from the Hanson-Wright inequality \citep{rudelson2013} under the independent sub-Gaussian assumption since $\| \bsbH \|_2 = 1$ and $ \|\bsbH\|_F^2=r(\bsbX)$. Then applying the Cauchy-Schwartz inequality gives the result (details omitted).

\subsection{Proof of Theorem \ref{theorem:IQ_convergence}}
\label{app:proofiq}
First,   
a thresholding function \citep{she2011outlier} is defined as a real-valued function $\Theta(t;\lambda)$   for $-\infty < t <\infty$ and $0\leq \lambda < \infty$ such that
(i) $\Theta(-t;\lambda) = -\Theta(t;\lambda)$;
(ii) $\Theta(t;\lambda) \leq \Theta(t';\lambda)$ for $t\leq t'$;
(iii) $\mathop{\lim}_{t\rightarrow\infty}\Theta(t;\lambda) = \infty$;
(iv) $0\leq \Theta(t;\lambda) \leq t$ for $0\leq t < \infty$.
$\Theta$ is defined component-wise if either $t$ or $\lambda$ is replaced
by a vector. Throughout the paper, assume that  $\lambda$ is the threshold parameter.

 Let $f(\bar{\bm\beta}) = \tilde l(\bar{\bm\beta}) + P(\varrho\bm\beta;\lambda) + \nu\| \bm\gamma\|_2^2/2+\iota(\bsbg; q)$, where $\tilde l(\bar{\bm\beta}) = l(\bar{\bm X}\bar{\bm\beta};\bm y)$ and $\iota(\bsbg; q) = 0$ if $\| \bsbg\|_0\le q$ and $+\infty$ otherwise. Let   $g (\bar{\bm\beta};\bar{\bm\beta}^-) = l(\bar \bsbX \bar \bsbb^-;\bsby) + \langle \bar \bsbX^T \nabla l(\bar \bsbX \bar \bsbb^-;\bsby), \bar \bsbb - \bar \bsbb^- \rangle + \varrho^2 \|\bar \bsbb - \bar \bsbb^-\|_2^2/2 + \nu\|\bsbg\|_2^2/2 + P(\varrho\bm\beta;\lambda)+\iota(\bsbg; q)$. We drop the dependence  on  $\rho$ for notational simplicity.

First, we claim that   for any  $\bar \bsbb, \bar \bsbb^{-}$,
\begin{align}
\breg_{\tilde l}(\bar \bsbb,  \bar \bsbb^-) \le \varrho^2 \mathbf D_2(\bar \bsbb, \bar \bsbb^-)
\end{align}
provided that $\varrho \ge \sqrt L \|\bar \bsbX\|_2$, where the generalized Bregman notations defined in Section \ref{sec:theory} are used. In fact,
\begin{align}
&\ \breg_{\tilde l}(\bar \bsbb,  \bar \bsbb^-) \notag\\
= &\ {\tilde l}(\bar\bsbb)-{\tilde l}(\bar\bsbb^{-})-\langle\nabla {\tilde l}(\bar\bsbb^{-}),\bar\bsbb-\bar\bsbb^{-}\rangle\notag\\
= & \int_{0}^{1}\langle\nabla \tilde l (  \bar\bsbb^{-}+t (\bar\bsbb-\bar\bsbb^{-})),\bar\bsbb-\bar\bsbb^{-}\rangle \rd t-\int_{0}^{1}\langle\nabla \tilde l  ( \bar\bsbb^{-}),\bar\bsbb-\bar\bsbb^{-}\rangle \rd t\notag\\
= & \int_{0}^{1}\langle\bar\bsbX^T    \nabla  l  (\bar\bsbX \bar\bsbb^{-}+t\bar\bsbX(\bar\bsbb-\bar\bsbb^{-})),\bar\bsbb-\bar\bsbb^{-}\rangle \rd t-\int_{0}^{1}\langle\bar\bsbX^T  \nabla  l (\bar\bsbX\bar\bsbb^{-}),\bar\bsbb-\bar\bsbb^{-}\rangle \rd t\notag\\
= & \int_{0}^{1}\langle  \nabla  l (\bar\bsbX\bar\bsbb^{-}+t\bar\bsbX(\bar\bsbb-\bar\bsbb^{-}))- \nabla  l (\bar\bsbX\bar\bsbb^{-}),\bar\bsbX (\bar\bsbb-\bar\bsbb^{-})\rangle \rd t\notag\\
\leq & \int_{0}^{1}\lVert  \nabla  l (\bar\bsbX\bar\bsbb^{-}+t\bar\bsbX(\bar\bsbb-\bar\bsbb^{-}))- \nabla  l (\bar\bsbX\bar\bsbb^{-})\rVert_{2} \|\bar\bsbX (\bar\bsbb-\bar\bsbb^{-})\rVert_{2}  \rd t \notag\\ 
\le & \int_{0}^{1} t L\|\bar\bsbX \bar\bsbb-\bar\bsbX\bar\bsbb^{-} \|_2 \|\bar\bsbX \bar\bsbb-\bar\bsbX\bar\bsbb^{-} \|_2 \rd t \notag\\
= & \frac{L  }{2} \| \bar\bsbX\bar\bsbb- \bar\bsbX\bar\bsbb^{-} \|_2^2   \le  \frac{L M_{ \bar\bsbX}  (p, 2q)}{2} \| \bar\bsbb-\bar\bsbb^{-} \|_2^2 \le  \frac{L \| \bar\bsbX\|_2^2}{2} \| \bar\bsbb-\bar\bsbb^{-} \|_2^2. \notag
\end{align}

Next, we solve $\min_{\bar \bsbb} g (\bar{\bm\beta};\bar{\bm\beta}^-) $. Due to the   separability, the problem is equivalent to
\begin{align*} &\min_{\bsbb}\varrho^2 \|  \bsbb -   \bsbb^-+   \bsbX^T \nabla l(\bar \bsbX \bar \bsbb^-;\bsby)/\varrho^2\|_2^2/2 + P(\varrho\bm\beta;\lambda) , \mbox{ and }\\
&  \min_{\bsbg} \|  \bsbg -   \bsbg^-+   \nabla l(\bar \bsbX \bar \bsbb^-;\bsby)/\varrho^2\|_2^2/2 +    (\nu/ \varrho^2) \|\bsbg\|_2^2/2   +\iota(\bsbg; q).
\end{align*} With a change of variables, it leads to the MM updates given in the theorem, cf.   \cite{She2012}. Furthermore, it is easy to show (details omitted) that the construction of  $g$ is equivalent to    $$g(\bar{\bm\beta};\bar{\bm\beta}^-) = f(\bar{\bm\beta}) + (\varrho^2\mathbf D_2-\bm\Delta_{\tilde l})(\bar{\bm\beta},\bar{\bm\beta}^-).$$Then,   $g (\bar \bsbb; \bar \bsbb^-) - f(\bar \bsbb) = \varrho^2\mathbf D_2(\bar \bsbb, \bar \bsbb^-)-\bm\Delta_{\tilde l}(\bar \bsbb,  \bar \bsbb^-) $.
By   the optimality of $\bar{\bm\beta}^{(t+1)}$, we obtain
\begin{equation}
\varrho^2\mathbf D_2(\bar \bsbb, \bar \bsbb^-)-\bm\Delta_{\tilde l}(\bar \bsbb,  \bar \bsbb^-)+f(\bar \bsbb^{(t+1)}) \leq g(\bar \bsbb^{(t+1)}; \bar \bsbb^{(t)} ) \leq g (\bar \bsbb^{(t)}; \bar \bsbb^{(t)}) = f(\bar \bsbb^{(t)}). \label{MM_decreasing}
\end{equation} It follows that
\begin{equation} \label{joint_sublinear}
(\varrho^2\mathbf D_2-\bm\Delta_{\tilde l})(\bar{\bm\beta}^{(t+1)},\bar{\bm\beta}^{(t)}) \le f(\bar{\bm\beta}^{(t)}) - f(\bar{\bm\beta}^{(t+1)}).
\end{equation}
Summing up \eqref{joint_sublinear} over $t = 0,1,\ldots,T$, and using $\breg_{\tilde l}(\bar \bsbb, \bar \bsbb^-)\le  {L  }\Breg_2 ( \bar\bsbX\bar\bsbb,  \bar\bsbX\bar\bsbb^{-} )$ proved before, we obtain  the convergence rate.

  Finally, we prove the limit-point equation. Since $$\nu\|\bm\gamma^{(t)}\|_2^2/2 +P(\varrho\bm\beta^{(t)};\lambda) \le f(\bar{\bm\beta}^{(t)}) \le f(\bar{\bm\beta}^{(0)}),$$ the sequence of   $\{\bar{\bm\beta}^{(t)}\}_{t=0}^\infty$ is uniformly bounded under $\nu>0$. Moreover,  $\lim_{t\rightarrow \infty} (f(\bar{\bm\beta}^{(t+1)}) - f(\bar{\bm\beta}^{(t)})) = 0$ implies that $\lim_{t\rightarrow \infty}(\varrho^2\mathbf D_2-\bm\Delta_{\tilde l})(\bar{\bm\beta}^{(t+1)}, \bar{\bm\beta}^{(t)}) = 0$ and so $\lim_{t\rightarrow 0}(\bar{\bm\beta}^{(t+1)}-\bar{\bm\beta}^{(t)}) = 0$ since $\varrho^2  >  {L}M_{  \bar \bsbX}(p, 2q )$.
Let $(\hat{\bm\beta},\hat{\bm\gamma})$ be any limit point of $(\bm\beta^{(t)},\bm\gamma^{(t)})$ satisfying $\hat{\bm\beta}=\lim_{k\rightarrow \infty}\bm\beta^{(j_k)}$ and $\hat{\bm\gamma}=\lim_{k\rightarrow \infty}\bm\gamma^{(j_k)}$ for some sequence $\{j_k,k=1,2,\ldots\}$. Then
\begin{equation}
\begin{split}
\bm 0 &= \lim_{k\rightarrow\infty}\{\bm\beta^{(j_k+1)}-\bm\beta^{(j_k)}\}\\
 &= \lim_{k\rightarrow\infty}\Theta(\varrho\bm\beta^{(j_k)}-\bm X^{T}\nabla l(\bm X\bm\beta^{(j_k)}+\bm\gamma^{(j_k)})/\varrho;\lambda)/\varrho-\hat{\bm\beta}\\
&= \Theta(\varrho\hat{\bm\beta}-\bm X^{T}\nabla l(\bm X\hat{\bm\beta}+\hat{\bm\gamma})/\varrho ;\lambda)/\varrho -\hat{\bm\beta}
\end{split}
\end{equation}
and
\begin{equation}
\begin{split}
\bm 0 &= \lim_{k\rightarrow\infty}\{\bm\gamma^{(j_k+1)}-\bm\gamma^{(j_k)}\} \\
&= \lim_{k\rightarrow\infty}\Theta^{\#}(\bm\gamma^{(j_k)}-\nabla l(\bm X\bm\beta^{(j_k)}+\bm\gamma^{(j_k)})/\varrho^2;q,\nu/\varrho^2)-\hat{\bm\gamma} \\
&= \Theta^{\#}(\hat{\bm\gamma}-\nabla l(\bm X\hat{\bm\beta}+\hat{\bm\gamma})/\varrho^2;q,\nu/\varrho^2)-\hat{\bm\gamma}
\end{split}
\end{equation}
due to the continuity assumption of $\Theta$ and the $\Theta^\#$-uniqueness assumption.

\begin{remark} \label{appremlinesearch}
The conclusion of $f(\bar{\bm\beta}^{(t)}) \le  f(\bar{\bm\beta}^{(t+1)})$ holds more generally, as long as $ g(\bar \bsbb^{(t+1)}; \bar \bsbb^{(t)} )\le   g(\bar \bsbb^{(t)}; \bar \bsbb^{(t)})$.  Moreover, the exact value of $L$ or $M_{  \bar \bsbX}(p, 2q )$ need not be known, because we can perform a line search until  $f(\bar \bsbb^{(t+1)}) \leq g(\bar \bsbb^{(t+1)}; \bar \bsbb^{(t)} )$ is satisfied.
\end{remark}
\subsection{Thresholding equations of A-estimators   } \label{proof:Theta_equation}
For the  penalized form problem $\min_{\bsbb, \|\bsbg\|_0\le q} l(\bar{\bm X}\bar{\bm\beta};\bm y) + P(\varrho \bsbb; \lambda) $ with    $P$   associated with a thresholding  $\Theta$:
$
P(\theta; \lambda) - P(0; \lambda) = \int_0^{|\theta|} (\sup \{s:\Theta(s;\lambda) \leq u\} - u) \rd u
$ 
 as defined in Theorem \ref{theorem:IQ_convergence}, let $\mathcal A(\lambda, q)$ be the set of all alternative estimators satisfying $\hat \bsbb \in \arg \min_{ \bm\beta } l(\bm X\bm\beta + \hat \bsbg; \bm y) + P(\varrho \bsbb; \lambda)$  and $\hat \bsbg \in \arg \min_{ \bm\gamma } l(\bm X\hat \bsbb  +   \bsbg; \bm y)$   s.t. $   \|\bm\gamma\|_0 \le q$.

Similarly, for the $\ell_0$-constrained form, let $\mathcal A(q_\beta, q_\gamma)$ be the set of alternative estimators  $(\hat{\bm\beta}, \hat{\bm\gamma})$ satisfying $\hat \bsbb \in \arg \min_{ \bm\beta } l(\bm X\bm\beta + \hat \bsbg; \bm y)$   s.t. $   \|\bm\beta\|_0 \le q_\beta$ and $\hat \bsbg \in \arg \min_{ \bm\gamma } l(\bm X\hat \bsbb  +   \bsbg; \bm y)$   s.t. $   \|\bm\gamma\|_0 \le q_\gamma$.
(Recall      the definition of $M_{  \bsbX}$ in Section \ref{sec:intro} and we know $M_{\bsbX}(q)\le \| \bsbX\|_2^2$ for any $q$.)

\begin{thm} \label{th:Theta_equation}
 Under  the 1-Lipschitz condition of $\nabla l$,  (i) any       $(\hat{\bm\beta},\hat{\bm\gamma})\in \mathcal A (q_\beta, q_\gamma)$ satisfies the following two quantile-thresholding equations:
\begin{equation}\label{eq:fix_point_BCD-proof}
\begin{cases}
\hat \bsbb = \Theta^{\#}(\hat \bsbb - \bsbX^T \nabla l(\bsbX \hat \bsbb + \hat \bsbg;\bm y)/\rho ; q_\beta) \\
\hat \bsbg  = \Theta^{\#}( \hat \bsbg -  \nabla l(\bsbX \hat \bsbb + \hat \bsbg;\bm y) ;q_\gamma)
\end{cases}
\end{equation}
for \emph{any} $\rho> M_{\bsbX}(2 q_\beta)$, and (ii) any        $(\hat{\bm\beta},\hat{\bm\gamma})\in \mathcal A (\lambda, q)$ satisfies the mixed thresholding equations
\begin{equation}\label{eq:fix_point_BCD-proof}
\begin{cases}
\hat \bsbb = \Theta (\varrho \hat \bsbb - \bsbX^T \nabla l(\bsbX \hat \bsbb + \hat \bsbg;\bm y)/\varrho  ; \lambda)/\varrho  \\
\hat \bsbg  = \Theta^{\#}( \hat \bsbg -  \nabla l(\bsbX \hat \bsbb + \hat \bsbg;\bm y) ;q )
\end{cases}
\end{equation}
for \emph{any} $\varrho> 0$, provided that $\Theta$ is continuous at $\varrho \hat \bsbb - \bsbX^T \nabla l(\bsbX \hat \bsbb + \hat \bsbg;\bm y)/\varrho $.
\end{thm}

These equations are also satisfied by the fixed points of our BCD algorithms seen from  the  iterate updates given in Section \ref{sec:algorithm}.
\begin{proof}

We prove the result for the doubly constrained form first. Given $\hat \bsbb$,   construct
\begin{equation}
\begin{split}
g(\bm\gamma;\bm\gamma^-) &= l(\bm X\hat \bsbb +\bm\gamma;\bm y) + (\mathbf D_2-\bm\Delta_{l(\bsbX \hat \bsbb + \cdot)} )(\bm\gamma,\bm\gamma^-).\\
\end{split}
\end{equation}
Under the 1-Lipschitz continuity of $\nabla l$, for any $\bm\gamma, \bm\gamma^-$, $$g(\bm\gamma;\bm\gamma^-) \ge l(\bm X\hat \bsbb+\bm\gamma;\bm y).$$
Recall that   $\hat \bsbg$ is a minimizer of   $ l(\bm X\hat \bsbb +\bm\gamma;\bm y) \text{ subject to }\|\bm\gamma\|_0\le q_\gamma$. Define
\begin{align}\tilde \bsbg  := \Theta^{\#}(\hat \bsbg -\nabla l(\bm X\hat \bsbb+\hat \bsbg); q_\gamma). \label{newgammodef}
\end{align} Then, by use of Lemma \ref{lemma:bound without design} and the construction of $g$, we have
\begin{equation}
l(\bm X\hat \bsbb +\tilde \bsbg;\bm y) \le g(\tilde \bsbg;\hat  \bsbg) \le g(\hat \bsbg;\hat \bsbg)  = l(\bm X\hat \bsbb +\hat\bsbg ;\bm y),
\end{equation}
which means  $\tilde \bsbg$ must   also be a globally optimal solution  to  $\min_{ \bm\gamma:  \|\bm\gamma\|_0 \le q_\gamma } l(\bm X\hat \bsbb  +   \bsbg; \bm y)$.  The optimal support   is uniquely determined due to the $\Theta^\#$-uniqueness assumption.
   Moreover,    $\nabla l(\bm X\hat \bsbb+\hat \bsbg) $ restricted to  $\mathcal J  = \{i:\hat \gamma_i \ne 0\}$ is a zero vector. It follows from \eqref{newgammodef} that $$\tilde  \gamma_{i} = \hat  \gamma_{i}\ne 0,  \forall i \in \mathcal J $$ and so   $\hat \bsbg   = \Theta^{\#}(\hat \bsbg -\nabla l(\bm X\hat \bsbb+\hat \bsbg); q_\gamma)$.

Similar to the proof of Theorem \ref{theorem:IQ_convergence},   given $\hat \bsbg$, any  $\hat \bsbb\in \arg\min_{\|\bsbb\|_0\le q_{\beta}   }    l(\bm X \bsbb +\hat \bsbg ;\bm y)=: \tilde l(\bsbb)$ satisfies
(details omitted) $$
\tilde l (\tilde \bsbb) + (\rho \Breg_2 - \breg_{\tilde l})(\tilde  \bsbb, \hat \bsbb) \le \tilde l (\hat \bsbb),
$$
where
\begin{align}\label{fixedbetaeq}
\tilde \bsbb= \Theta^{\#}(\hat \bsbb - \bm X^T\nabla l(\bm X\hat \bsbb+\hat \bsbg;\bm y)/\rho; q_\beta).
\end{align}
  Therefore, for any
  $\rho > M_{ \bsbX}(2 q_\beta)$, $\tilde l (\tilde \bsbb) =\tilde l (\hat \bsbb) $,  $(\rho \Breg_2 - \breg_{\tilde l})(\tilde \bsbb, \hat \bsbb)= 0$, and so  $\tilde \bsbb = \hat \bsbb$ (cf. Appendix \ref{app:proofiq}).

To prove the result for the penalized form, it suffices to study  the conditions satisfied by a local or coordinate-wise minimizer $\hat\bsbb$ of   $f_{\Theta}:= l( \bsbX \bsbb;  \bsby) + \sum_{j=1}^p P_{\Theta}(|\beta_j|; \lambda)$.  The  proof  of Theorem 1 in \cite{She2016} can be modified for the purpose. For completeness, we give the details below. Denote the criterion by   $f$ for simplicity and define $s(u;\lambda):=\Theta^{-1}(u;\lambda) -u$ for $u\ge 0$. Let $\delta f(\bsbb; \bsbh)$ denote the directional derivative of $f$ at $\bsbb$ with increment  $\bsbh$: $$\delta f(\bsbb; \bsbh) = \lim_{\epsilon \rightarrow 0+} \frac{f(\bsbb + \epsilon \bsbh)-f(\bsbb )}{\epsilon}.$$ By the definition of $P_{\Theta}$,  $\delta f(\bsbb, \bsbh)$ exists for any  $ \bsbh \in \mathbb R^p$. Let $l_0(\bsbb) = l(\bsbX \bsbb)$. Consider the following directional vectors: $\bsbd_j= [d_1, \cdots, d_p]^T$ with $d_j=\pm 1$ and $d_{j'}=0, \forall j'\ne j$. Then for any $j$,
\begin{align}
\delta l_{0}(\bsbb; \bsbd_j) &= d_j \bsbx_j^T \nabla l (\bsbX \bsbb ),\\
\delta  P_{\Theta}(\bsbb; \bsbd_j) & = \begin{cases} s(|\beta_j|) \mbox{sgn}(\beta_j) d_j, &\mbox{ if } \beta_j\ne 0, \\ s(|\beta_j|), & \mbox{ if } \beta_j=0.\end{cases}
\end{align}

Due to the local   coordinatewise optimality of $\hat\bsbb$,  $\delta f(\hat\bsbb; \bsbd_{j})\ge 0$, $\forall j$.
When $\hat \beta_1 \ne 0$, we obtain $\bsbx_1^T \nabla l (\bsbX \hat \bsbb ) + s(|\hat\beta_1|;\lambda) \mbox{sgn}(\hat\beta_1) = 0$. When $\hat \beta_1=0$,  $\bsbx_1^T \nabla l (\bsbX \hat \bsbb ) + s(|\hat\beta_1|;\lambda) \ge  0$ and $-\bsbx_1^T \nabla l (\bsbX \hat \bsbb ) + s(|\hat\beta_1|;\lambda) \ge  0$, i.e., $|\bsbx_1^T \nabla l (\bsbX \hat \bsbb )| \le  s(|\hat\beta_1|;\lambda)=\Theta^{-1}(0;\lambda)$. To summarize,   whenever $f$ achieves   a     {local} coordinate-wise minimum  at $\hat \bsbb$,  we have
\begin{align}
&\hat\beta_j \ne 0 \Rightarrow \Theta^{-1}(|\hat \beta_j|;\lambda)\mbox{sgn} (\hat \beta_j) = \hat \beta_j - \bsbx_j^T \nabla l (\bsbX \hat \bsbb )\label{nzb-1}\\
&\hat\beta_j = 0 \Rightarrow \Theta(\bsbx_j^T \nabla l (\bsbX \hat \bsbb ); \lambda) = 0.
\end{align}
When $\Theta$ is continuous at $\hat \beta_j - \bsbx_j^T \nabla l (\bsbX \hat \bsbb )$, \eqref{nzb-1} implies that $\hat \beta_j = \Theta(\hat \beta_j - \bsbx_j^T \nabla l (\bsbX \hat \bsbb ); \lambda)$. Hence $\hat \bsbb$ must  satisfy  $\hat \bsbb = \Theta( \hat \bsbb - \bsbX^T \nabla l (\bsbX \hat \bsbb ); \lambda)$.
For the original problem, let $\bsbb' = \varrho \bsbb$ and notice that $\nabla_{\bsbb'} l(\bsbX \bsbb + \hat\bsbg) = \bsbX^T \nabla l (\bsbX \bsbb+ \hat\bsbg) /\varrho$. The conclusion thus follows.
   \end{proof}

\subsection{Proof of Theorem \ref{thm:fixed-point}} \label{proof:fixed-point}


Recall that given any $\bsbb$,  $\mathcal J(\bsbb)$   denotes its support, i.e., $\mathcal J(\bsbb)=\{j: \beta_j \ne 0\}$, and $J(\bsbb ) = |\mathcal J(\bsbb)|$ and  $s^* = J(\bsbb^*)$ and $o^* = J(\bsbg^*)$.
Given a matrix $\bsbA$, let  $\mathcal P_{\bm A}$ be the orthogonal projection matrix onto the column space of $\bm A$ and $\mathcal P_{\bm A}^\perp$ be its orthogonal complement. The range of $\bsbA$ is denoted by $\mathcal R(\bsbA)$. For convenience, we also use   $\mathcal P_{\bm A}$ to stand for $\mathcal R(\bsbA)$ or   $\mathcal R(\mathcal P_{\bm A})$.  In the proofs,   $[n]$  denotes the set of $\{1, \ldots, n\}$ and   we use $C$ and $c$ and $L$ to denote universal constants which are not necessarily the same at each occurrence.

\begin{lemma} \label{bcdsurro}
Any     $(\hat{\bm\beta},\hat{\bm\gamma})$ satisfying \eqref{eq:fix_point_BCD} can be  re-characterized  by\begin{equation} \label{fixed-point_surro}
(\hat{\bm\beta},\hat{\bm\gamma}) \in \mathop{\arg\min}_{(\bm\beta,\bm\gamma)} g(\bm\beta,\bm\gamma;\bm\beta^-,\bm\gamma^-)|_{\bm\beta^-=\hat{\bm\beta},\bm\gamma^-=\hat{\bm\gamma}} \text{ s.t. }\|\bm\gamma\|_0\le q_\gamma, \|\bm\beta\|_0\le q_\beta,
\end{equation}
where \begin{align}
g(\bm\beta,\bm\gamma;\bm\beta^-,\bm\gamma^-) = \,& l(\bm X\bm\beta^-+\bm\gamma^-) + \langle\nabla l(\bm X\bm\beta^-+\bm\gamma^-), \bm X\bm\beta-\bm X\bm\beta^- + \bm\gamma-\bm\gamma^-\rangle \notag\\
&  + \rho\|\bm\beta-\bm\beta^-\|_2^2/2 + \|\bm\gamma-\bm\gamma^-\|_2^2/2.
\label{gbcddef}\end{align}
\end{lemma}
The proof of the lemma can be shown from Lemma \ref{lemma:Theta_optimality} and the details are omitted.
 A pleasant   fact  is that \eqref{fixed-point_surro} gives a {joint} optimization problem, rather than an alternative one. Also, notice that   $g$ may \textbf{not} majorize $l$.

From \eqref{gbcddef}, $g(\bsbb, \bsbg; \hat \bsbb, \hat \bsbg)$   is equivalent to   $\rho\|\bm\beta-\hat{\bm\beta}+\bm X^T\nabla l(\bm X\hat{\bm\beta}+\hat{\bm\gamma})/\rho\|_2^2/2 + \|\bm\gamma-\hat{\bm\gamma}+\nabla l(\bm X\hat{\bm\beta}+\hat{\bm\gamma})\|_2^2/2$ up to some additive terms dependent on $\hat \bsbb, \hat \bsbg$ only. By $g(\hat \bsbb, \hat \bsbg; \hat \bsbb, \hat \bsbg)\le g(\bsbb^*, \bsbg^*; \hat \bsbb, \hat \bsbg)$ and Lemma \ref{lemma:bound without design}, \begin{equation}
\begin{split}
& \frac{1}{2}\|\bm\gamma^*-\hat{\bm\gamma}+\nabla l(\bm X\hat{\bm\beta}+\hat{\bm\gamma})\|_2^2 - \frac{1}{2}\|\nabla l(\bm X\hat{\bm\beta}+\hat{\bm\gamma})\|_2^2 \\
& + \frac{\rho}{2}\|\bm\beta^*-\hat{\bm\beta}+\bm X^T\nabla l(\bm X\hat{\bm\beta}+\hat{\bm\gamma})/\rho\|_2^2 - \frac{\rho}{2}\|\bm X^T\nabla l(\bm X\hat{\bm\beta}+\hat{\bm\gamma})/\rho\|_2^2\\
\ge\,& \frac{1}{2}\big[1-\mathcal L(\mathcal J(\bm\gamma^*),\mathcal J(\hat{\bm\gamma}))\big]\|\hat{\bm\gamma}-\bm\gamma^*\|_2^2 + \frac{\rho}{2}\big[1-\mathcal L(\mathcal J(\bm\beta^*),\mathcal J(\hat{\bm\beta}))\big]\|\hat{\bm\beta}-\bm\beta^*\|_2^2.
\end{split}
\end{equation}
From  $\mathcal L(\mathcal J(\bm\gamma^*),\mathcal J(\hat{\bm\gamma})) \le 1/\sqrt{\vartheta}$, $\mathcal L(\mathcal J(\bm\beta^*),\mathcal J(\hat{\bm\beta})) \le 1/\sqrt{\vartheta}$, and   the definition of noise $\bm\epsilon$, we can show   (details omitted)
\begin{equation} \label{fixed-point_bound-1}
\begin{split}
2\bar{\bm\Delta}_l(\bm X\hat{\bm\beta}+\hat{\bm\gamma}, \bm X\bm\beta^*+\bm\gamma^*) \le   \frac{1}{2\sqrt{\vartheta}}(\|\hat{\bm\gamma}-\bm\gamma^*\|_2^2+\rho\|\hat{\bm\beta}-\bm\beta^*\|_2^2)  + \langle\bm\epsilon, \bm X\hat{\bm\beta}-\bm X\bm\beta^*+\hat{\bm\gamma}-\bm\gamma^*\rangle.
\end{split}
\end{equation}
 Define
$$
P (\bm\beta,\bm\gamma) =  J(\bm\beta) + J(\bm\beta)\log(ep/J(\bm\beta)) + J(\bm\gamma) + J(\bm\gamma)\log(en/J(\bm\gamma)).
$$
Because the function depends on cardinality only, we also denote it by
$P (J(\bm\beta),J(\bm\gamma))$.
We will show that with high probability, the last stochastic term or $\langle\bm\epsilon, \bar{\bm X}\hat{\bar{\bm\beta}}-\bar{\bm X}\bar{\bm\beta}^*\rangle$ in \eqref{fixed-point_bound-1} can be bounded by the sum  of $\|\bar{\bm X}\hat{\bar{\bm\beta}}-\bar{\bm X}\bar{\bm\beta}^*\|_2^2$ and  $\sigma^2 P (\bm\beta^*,\bm\gamma^*) +\sigma^2  P (\hat{\bm\beta},\hat{\bm\gamma})$, up to  some multiplicative constants.

 To achieve the purpose, we decompose $\bar\bsbX   ( \hat{\bar{\bm\beta}}-\bar{\bm\beta} ^*)=: \bar{\bm X}\bm\Delta$ as follows
\begin{equation}
\begin{split}
\bar{\bm X}\bm\Delta  &= \mathcal P_{\bar{\bm X}_{\mathcal J(\bar{\bm\beta}^*)}}(\bar{\bm X}\bm\Delta ) + \mathcal P_{\bar{\bm X}_{\mathcal J(\bar{\bm\beta}^*)}}^\perp(\bar{\bm X}\bm\Delta ) \\
&= \mathcal P_{\bar{\bm X}_{\mathcal J(\bar{\bm\beta}^*)}}(\bar{\bm X}\bm\Delta ) + \mathcal P_{\bar{\bm X}_{\mathcal J(\bar{\bm\beta}^*)}}^\perp \mathcal (\bm X\hat{\bm\beta} +\hat{\bm\gamma}) - \mathcal P_{\bar{\bm X}_{\mathcal J(\bar{\bm\beta}^*)}}^\perp \mathcal (\bm X {\bm\beta}^* + {\bm\gamma}^*)\\
&= \mathcal P_{\bar{\bm X}_{\mathcal J(\bar{\bm\beta}^*)}}(\bar{\bm X}\bm\Delta ) + \mathcal P_{\bar{\bm X}_{\mathcal J(\bar{\bm\beta}^*)}}^\perp \mathcal P_{\bar{\bm X}_{\mathcal J(\hat{\bar{\bm\beta}})}} (\bm X\hat{\bm\beta} +\hat{\bm\gamma}) \equiv {I} + {II}.
\end{split}
\end{equation}
It is easy to verify that  $\|I\|_2^2+\|II\|_2^2 = \|\bar{\bm X}\bm\Delta \|_2^2$.   We use Lemma \ref{concenGauss} and Lemma \ref{concenGauss2}  to handle $\langle\bm\epsilon, I\rangle$ and $\langle\bm\epsilon, II\rangle$, respectively. The proof of Lemma \ref{concenGauss} is given at the end; the proof of Lemma  \ref{concenGauss2} is similar and omitted.

\begin{lemma} \label{concenGauss}
Given $[\bm X^{(1)}, \bm X^{(2)}] \in\mathbb R^{n\times (p_1+p_2)}$ where $\bm X^{(1)}\in\mathbb R^{n\times p_1}, \bm X^{(2)}\in\mathbb R^{n\times p_2}, rank(\bm X^{(1)})= r_1, rank(\bm X^{(2)})=r_2$ and $1\le J_1\le p_1$, $1\le J_2\le p_2$, define $\Gamma_{J_1,J_2}=\{\bm\alpha\in\mathbb R^{p_1+p_2}: \|\bm\alpha\|_2 \le 1, \bm\alpha\in \mathcal P_{[\bm X^{(1)}_{\mathcal J_1},\bm X^{(2)}_{\mathcal J_2}]} \text{ for some } \mathcal J_1\subset[p_1], \mathcal J_2\subset[p_2]\text{ with }|\mathcal J_1|\le J_1, |\mathcal J_2| \le  J_2\}$. Let $P' (J_1,J_2) = \sigma^2\{J_1\wedge r_1+J_2\wedge r_2+\log{p_1\choose J_1}+\log{p_2\choose J_2}\}$. Then for any $t\ge 0$,
\begin{equation} \label{fixed-point_bound-2}
\mathbb P\bigg(\sup_{\bm\alpha\in\Gamma_{J_1,J_2}}\langle\bm\epsilon, \bm\alpha\rangle > t\sigma + \sqrt{LP' (J_1,J_2)}\bigg) \le C\exp(-ct^2),
\end{equation}
where $L,C,c$ are universal constant.
\end{lemma}

\begin{lemma} \label{concenGauss2}
Given $\bm X = [\bm X^{(1)}, \bm X^{(2)}] \in\mathbb R^{n\times (p_1+p_2)}$ where $\bm X^{(1)}\in\mathbb R^{n\times p_1}, \bm X^{(2)}\in\mathbb R^{n\times p_2}$, $rank(\bm X^{(1)})=r_1$, $rank(\bm X^{(2)})=r_2$ and $1\le J_1,J_1'\le p_1$, $1\le J_2,J_2'\le p_2$, let  $\mathcal P_{\mathcal J_1,\mathcal J_2} := \mathcal P_{[\bm X^{(1)}_{\mathcal J_1},\bm X^{(2)}_{\mathcal J_2}]}$ and define $\Gamma_{J_1,J_2,J_1',J_2'}=\{\bm\alpha\in\mathbb R^{p_1+p_2}: \|\bm\alpha\|_2 \le 1, \bm\alpha\in \mathcal R(\mathcal P^\perp_{\mathcal J_1,\mathcal J_2}\mathcal P_{\mathcal J_1',\mathcal J_2'})  \text{ for } \mathcal J_1,\mathcal J_1'\subset[p_1], \mathcal J_2,\mathcal J_2\subset[p_2]\text{ and }|\mathcal J_1|\le J_1, |\mathcal J_2| \le  J_2, |\mathcal J_1'|\le J_1', |\mathcal J_2'|\le J_2'\}$. Let $P'' (J_1,J_2, J_1',J_2') = \sigma^2\{J_1'\wedge r_1+ J_2'\wedge r_2+\log{p_1\choose J_1}{p_2\choose J_2} {p_1\choose J'_1} {p_2\choose J'_2}\}$. Then for any $t\ge 0$,
\begin{equation}
\mathbb P\bigg(\sup_{\bm\alpha\in\Gamma_{J_1,J_2,J_1',J_2'}}\langle\bm\epsilon, \bm\alpha\rangle > t\sigma + \sqrt{LP'' (J_1,J_2,J_1',J_2')}\bigg) \le C\exp(-ct^2),
\end{equation}
where $L,C,c$ are universal constant.
\end{lemma}


Define $\Gamma_{s,o} = \{\bm\Delta\in\mathcal P_{[\bm X_{\mathcal J_1}, \  \bm I_{\mathcal J_2}]}, \|\bm\Delta\|_2\le 1, \mathcal J_1\subset [p], \mathcal J_2\subset [n], |\mathcal J_1|\le  s, |\mathcal J_2|\le  o\}$. Given any $a, b, a'>0$, we have
\begin{equation} \nonumber
\begin{aligned}
& \langle \bm\epsilon, \mathcal P_{\bar{\bm X}_{\mathcal J(\bar{\bm\beta}^*)}} \bar{\bm X}\bm\Delta\rangle  - \frac{1}{a} \| \mathcal P_{\bar{\bm X}_{\mathcal J(\bar{\bm\beta}^*)}}\bar{\bm X}\bm\Delta\|_2^2 - bL\sigma^2  P (J(  \bsbb^*), J(\bsbg^*) ) \\
\le\,& \| \mathcal P_{\bar{\bm X}_{\mathcal J(\bar{\bm\beta}^*)}} \bar{\bm X}\bm\Delta\|_2 \langle \bm\epsilon, \frac{ \mathcal P_{\bar{\bm X}_{\mathcal J(\bar{\bm\beta}^*)}} \bar{\bm X}\bm\Delta}{\| \mathcal P_{\bar{\bm X}_{\mathcal J(\bar{\bm\beta}^*)}} \bar{\bm X}\bm\Delta \|_2}\rangle - 2 \sigma \| \mathcal P_{\bar{\bm X}_{\mathcal J(\bar{\bm\beta}^*)}} \bar{\bm X}\bm\Delta\|_2 \sqrt {\frac{b}{a}L P (J(  \bsbb^*), J(\bsbg^*))}  \\
\le\,&  \frac{1}{a'} \| \mathcal P_{\bar{\bm X}_{\mathcal J(\bar{\bm\beta}^*)}} \bar{\bm X}\bm\Delta \|^2_2 + \frac{a'}{4}   \sup_{\bm\Delta\in\Gamma_{s^*,o^*}} [\langle\bm\epsilon, \bm\Delta\rangle - 2\sigma\{{ ({b}/{a})LP (s^*, o^*)}\}^{1/2} ]_+^2
\end{aligned}\label{cauchyin}
\end{equation}
by the Cauchy-Schwarz inequality.

Let $R_{s, o}=  \sup_{\bm\Delta\in\Gamma_{s,o}} [\langle\bm\epsilon, \bm\Delta\rangle - 2\sigma\{{ ({b}/{a})LP (s, o)}\}^{1/2} ]_+ $.
Lemma \ref{concenGauss} indicates  that $R_{s,o}$ is bounded above by a constant times $\sigma^2$ in expectation    when choosing $b/a$ to be a constant greater than $1/4$. Note that when $s=o=0$, $R_{0,0}= 0 $. When $s\ge 1, o\ge 1$, for any $t\ge 0$,
\begin{equation} \label{Rbound}
\begin{split}
&\mathbb P(R_{s, o} > t\sigma)\\
\le\,&   \mathbb P\bigg(\sup_{\bm\Delta\in\Gamma_{s, o}}\langle\bm\epsilon, \bm\Delta\rangle - \sqrt{LP (s,o)} > t\sigma + 2\sigma\sqrt{\frac{b}{a}LP (s,o)} -\sigma \sqrt{LP (s,o)}\bigg)\\
\le\,&   C\exp(-ct^2)\exp[-c(2\sqrt{b/a}-1)^2 P (s,o)]\\
\le\, & C \exp(-ct^2 ) (np)^{-c} \le  C \exp(-ct^2 ),
\end{split}
\end{equation}
where we used $\log  {p \choose J } \le C J \log (ep/J)$ and $C, c>0$ are constants. Similarly, when $s = 0 $, we know $\EP(R_{s,o} > t \sigma)\le C \exp(-ct^2) n^{-c}  $ for any $0\le o \le n$,  and when $o=0$, $\EP(R_{s,o} > t \sigma)\le C \exp(-ct^2) p^{-c}$ for any $0\le s \le p$.   From these tail bounds, we get
$$
\EE \langle \bm\epsilon, \mathcal P_{\bar{\bm X}_{\mathcal J(\bar{\bm\beta}^*)}} \bar{\bm X}\bm\Delta\rangle  \le \EE[  ( \frac{1}{a} + \frac{1}{a'})\| \mathcal P_{\bar{\bm X}_{\mathcal J(\bar{\bm\beta}^*)}}\bar{\bm X}\bm\Delta\|_2^2 + bL \sigma P (J(  \bsbb^*), J(\bsbg^*) )] + a' C \sigma^2.
$$
Repeating the treatment using   Lemma \ref{concenGauss2} gives \begin{align*}
 \EE \langle \bm\epsilon, P_{\bar{\bm X}_{\mathcal J(\bar{\bm\beta}^*)}}^\perp \mathcal P_{\bar{\bm X}_{\mathcal J(\hat{\bar{\bm\beta}})}} (\bm X\hat{\bm\beta} +\hat{\bm\gamma}) \rangle  \le &\EE [ (\frac{1}{a} + \frac{1}{a'}) \|P_{\bar{\bm X}_{\mathcal J(\bar{\bm\beta}^*)}}^\perp \mathcal P_{\bar{\bm X}_{\mathcal J(\hat{\bar{\bm\beta}})}} (\bm X\hat{\bm\beta} +\hat{\bm\gamma})\|_2^2  \\ & +     bL\sigma\{  P (J(\hat \bsbb) , J(\hat \bsbg)) + P (J(\bsbb^*) , J(\bsbg^*))\}] + a' C \sigma^2.
\end{align*}
Combining the two bounds and using the monotone property of $P$, we get for any $4b> a>0, a'>0$,  \begin{equation}\label{fixed-point_bound-2}
\EE \langle\bm\epsilon, \bar{\bm X}{\bm\Delta}\rangle
\le \EE (\frac{1}{a}+\frac{1}{a'})\|\bar{\bm X}{\bm\Delta}\|_2^2+ 3bL\sigma^2 P ( \vartheta s^*,  \vartheta o^*) + a' C\sigma^2
\end{equation}
where $C,c$ are positive constants. It follows from  \eqref{fixed-point_bound-1}, \eqref{fixed-point_bound-2} and the regularity condition \eqref{fixed-point-condition} that
\begin{equation}
\EE [(\frac{\delta}{2} - (\frac{1}{a}+\frac{1}{a'}))\|\bar{\bm X}{\bm\Delta}\|_2^2 ]\le  bL' \sigma^2 \vartheta (o^*\log\frac{en}{ \vartheta o^*} + s^*\log\frac{ep}{ \vartheta s^*})+ C a' \sigma^2,
\end{equation}
where $L'$ is a  constant.
Taking $a=a' = 8/\delta$ and $b = 4/\delta$ gives the conclusion.

In the remaining, we  prove the more  general result under
\begin{equation}
\EE [\Breg_2 (\bar \bsbX   {\bar \bsbb}^{(0)}, \bar \bsbX \bar \bsbb^*  ) ]\le   C M  \allowbreak\{\vartheta o^*\sigma^2 \allowbreak \log\frac{en}{ \vartheta o^*} + \vartheta s^* \sigma^2\log\frac{ep}{ \vartheta s^*} +\sigma^2 \}\label{auxcond1}
\end{equation} where $M$ satisfies $  +\infty\ge M \ge 1$ and   $C\ge 0$ is a constant,  and
\begin{equation}
\Big\{2(1-\frac{1}{M})\bar{\bm\Delta}_l +\frac{C_0}{  M( M\delta_{0} \vee c_{0})} \breg_l -  \delta_{0}\Breg_2 \Big\}(\bar \bsbX  \bar \bsbb , \bar \bsbX  \bar \bsbb ')
 \ge \frac{(1-1/M)}{\sqrt{\vartheta}} (\rho\Breg_2(    \bsbb,    \bsbb ')+ \Breg_2(    \bsbg,    \bsbg '))  \label{auxcond2}
\end{equation}
for all $\bar \bsbb, \bar \bsbb'$:   $\|\bm\beta\|_0\le \vartheta s^*, \|\bm\beta'\|_0 \le  s^*$, $\|\bm\gamma\|_0 \le \vartheta o^*, \| \bm\gamma'\|_0 \le  o^*$ for some $\delta_{0}>0$ and   constants $C_0, c_0>0$.

First, from the above analysis, we have
\begin{equation}
  \EE [(2 \delta\bar{\bm\Delta}_l  - \delta^2 \Breg_2 )(\bar \bsbX \hat {\bar \bsbb} , \bar \bsbX \bar \bsbb^*  ) ]  -  \EE [   \frac{\delta  }{\sqrt \vartheta}\ ( \rho\Breg_2(   \hat  \bsbb,    \bsbb^*)+ \Breg_2(  \hat  \bsbg,    \bsbg^* )  ) ] \le  C E. \label{auxeq1}
\end{equation}
where  $E:=\vartheta o^*\sigma^2 \allowbreak \log\frac{en}{ \vartheta o^*} + \vartheta s^* \sigma^2\log\frac{ep}{ \vartheta s^*} +\sigma^2 $.
Next, from $l( \bar \bsbX \hat {\bar \bsbb} )\le l(  \bar \bsbX \bar \bsbb^{(0)})$,  we have
\begin{equation}
 \breg_l ( \bar \bsbX \hat {\bar \bsbb}, \bar \bsbX \bar \bsbb^*) \le\breg_l(   \bar \bsbX \bar \bsbb^{(0)}, \bar \bsbX \bar \bsbb^*) + \langle \bsbeps, \bsbX \hat {\bar \bsbb} - \bar \bsbX \bar \bsbb^*   \rangle -  \langle \bsbeps, \bsbX   {\bar \bsbb}^{(0)} - \bar \bsbX \bar \bsbb^*   \rangle
\end{equation}
from which it follows that for any $\delta'>0$,
\begin{equation}
\EE [   ( {\breg}_l  - \delta'  \Breg_2) (\bar \bsbX \hat {\bar \bsbb} , \bar \bsbX \bar \bsbb^*  ) ]  \le  (\breg_l + \frac{1}{ M} \Breg_2)(   \bar \bsbX \bar \bsbb^{(0)}, \bar \bsbX \bar \bsbb^*) + C E (\frac{1}{\delta'} +M).
\end{equation}
Since $\nabla l$ is Lip($1$) and  $M\ge 1$, we obtain
\begin{equation}
\EE [   ( {\breg}_l  - \delta'  \Breg_2) (\bar \bsbX \hat {\bar \bsbb} , \bar \bsbX \bar \bsbb^*  ) ]  \le C  E (\frac{1}{\delta'} +  M) \le \frac{C}{c_1 \wedge c_2 }E(\frac{c_{1}}{\delta'} + c_{2} M)   \label{auxeq2raw}
\end{equation}
where    $c_1, c_2>0$ are arbitrary constants and recall that     $C$ may not be the same constant at each occurrence.
Taking $\delta^2 = \delta'^2 /(c_1 + c_2 M \delta')$, we get
\begin{equation}
\EE [   (\frac{\delta^{2}}{\delta'} {\breg}_l  - \delta^2  \Breg_2) (\bar \bsbX \hat {\bar \bsbb} , \bar \bsbX \bar \bsbb^*  ) ]  \le    \frac{C}{c_1 \wedge c_2 }E.    \label{auxeq2}
\end{equation}

Multiplying \eqref{auxeq1} by $(1 - 1/M)$ and   \eqref{auxeq2} by $1/M$ and adding the two inequalities yield
\begin{align*}
 & \EE [  (1-\frac{1}{M})\{2 \bar{\bm\Delta}_l    (\bar \bsbX \hat {\bar \bsbb} , \bar \bsbX \bar \bsbb^*  )  -    \frac{  \rho\Breg_2(   \hat  \bsbb,    \bsbb^*)+ \Breg_2(  \hat  \bsbg,    \bsbg^* )}{\sqrt \vartheta}\  \} \\ &  +( \frac{\delta}{M\delta'} {\breg}_l- \delta   \Breg_2 )(\bar \bsbX \hat {\bar \bsbb} , \bar \bsbX \bar \bsbb^*  ) ]   \le     \frac{E}{\delta} (1+ \frac{C}{c_1 \wedge c_2}).
 \end{align*}
 Simple calculation shows $$ \delta'/   \delta  =   c_2 M \delta + \sqrt{c_2^2 M^2 \delta^2 + 4c_1} \le  (2c_2 M \delta) \vee (4c_1) $$
 Under    \eqref{auxcond2} with $\delta_0 = 2 \delta, C_0 = 1/c_2, c_0 = 4c_1/c_2$, we get $\EE [ \Breg_2  (\bar \bsbX \hat {\bar \bsbb} , \bar \bsbX \bar \bsbb^*  ) ] \lesssim E/ \delta_0^2$. A reparameterization of     \eqref{auxcond2}  gives the assumed regularity condition.

\paragraph{Proof of Lemma \ref{concenGauss}} By definition, $\{\langle\bm\epsilon, \bm\alpha\rangle: \bm\alpha\in \Gamma_{J_1,J_2}\rangle\}$ is a stochastic process with sub-Gaussian increments. The induced metric on $\Gamma_{J_1,J_2}$ is Euclidean $d(\bm\alpha_1,\bm\alpha_2) = \sigma\|\bm\alpha_1-\bm\alpha_2\|_2$.

To bound the metric entropy $\log\mathcal N(\varepsilon,\Gamma_{J_1,J_2},d)$, where $\mathcal N(\varepsilon,\Gamma_{J_1,J_2}, d)$ is the smallest cardinality of an $\varepsilon$-net that covers $\Gamma_{J_1,J_2}$ under $d$, we notice that $\bm\alpha$ is always in a $({r_1\wedge J_1+r_2\wedge J_2})$-dimensional ball and the number of such ball is at most ${p_1\choose J_1}{p_2\choose J_2}$. By a standard volume argument 
\begin{equation*}
\begin{split}
\log\mathcal N(\varepsilon,\Gamma_{J_1,J_2},d) &\le \log  {p_1\choose J_1}{p_2\choose J_2}\Big(\frac{C\sigma}{\varepsilon}\Big)^{r_1\wedge J_1+r_2\wedge J_2} \\
&=\log{p_1\choose J_1}+\log{p_2\choose J_2} + (r_1\wedge J_1+r_2\wedge J_2)\log(C\sigma/\varepsilon),
\end{split}
\end{equation*}
where $C$ is a universal constant.

From Dudley's integral bound, 
\begin{equation*}
\mathbb P\bigg(\sup_{\bm\alpha\in\Gamma_{J_1,J_2}} \langle\bm\epsilon, \bm\alpha\rangle\ge t\sigma + L\int_0^\sigma \sqrt{\log\mathcal N(\varepsilon, \Gamma_{J_1,J_2},d)}\text{d}\varepsilon\bigg) \le C\exp(-ct^2).
\end{equation*}
The conclusion follows from   the Cauchy-Schwarz inequality.
\qed \\

\begin{remark}  Because of the tail bounds,   the proof also gives a high-probability form result:
 $$
 \|\bar \bsbX  \hat{\bar \bsbb} - \bar \bsbX   {\bar \bsbb^*}\|_2^2 \lesssim \frac{\sigma^2 \vartheta}{\delta^2}\Big \{o^*\log\frac{en}{ \vartheta o^*} + s^*\log\frac{ep}{ \vartheta s^*}\Big\}
 $$
 with   probability at least $1 - C n^{-c(1 \wedge o^*)} p^{-c(1 \wedge s^*)}  $
 under the same regularity condition.   

In addition, if the $\bsbb$-optimization problem is convex,      such as in robust regression or logistic regression, the regularity condition can be weakened to $(2\bar{\bm\Delta}_l
-  \delta\Breg_2 )(\bm X\bm\beta +\bm\gamma, \bm X\bm\beta'+\bm\gamma')\ge \|\bm\gamma-\bm\gamma'\|_2^2/(4\sqrt{\vartheta})$ for any $\bm\beta,\bm\gamma,\bm\beta',\bm\gamma'$:  $\|\bm\gamma\|_0 \le \vartheta o^*, \| \bm\gamma'\|_0 \le  o^*$ and some $\delta>0$. It is also worth mentioning that  to get the risk bound,  \eqref{fixed-point-condition} can be stated in expectation.

\end{remark}

\subsection{Proof of Theorem \ref{thm:fixed-point-lasso}}
\label{proof:fixed-point-l1}

We prove a more general theorem  for a strongly smooth   loss   $l$ with $\nabla l$ 1-Lipschitz.

\begin{manualtheorem}{\ref{thm:fixed-point-lasso}'}\label{thm:fixed-point-lasso-generalloss}
Consider $\min_{\bsbb, \bsbg} l(\bsbX \bsbb+ \bsbg) + \lambda \varrho   \| \bsbb\|_1$  s.t. $\| \bsbg\|_0\le q$, where $\varrho >0$,  $\lambda = A\sigma\sqrt{\log(ep)}$ with $A$ a sufficiently large constant and $q = \vartheta  o^{*}$ with $\vartheta\ge 1$. Then the following inequality holds for any coordinatewise minimum point $(\hat{\bm\beta},\hat{\bm\gamma})$
\begin{equation}
\EE [ \|\bar \bsbX \hat{\bar \bsbb}-\bar \bsbX {\bar{\bsbb^*}}\|_2^2 ] \lesssim  {\sigma^2} \vartheta o^*\log({en}/{o^*}) + \sigma^2 K^{2} s^*\log(ep) + \sigma^2
\end{equation}
under the assumption that there exists $K\ge 0$ such that  $(2\bar{\bm\Delta}_l - \delta \Breg_2) (\bar \bsbX  {\bar \bsbb}, \bar \bsbX  {\bar \bsbb}')+ \lambda \varrho\|(\bm\beta-\bm\beta')_{\mathcal J^c}\|_1 + K^{2}\lambda^2J
\ge  \|\bm\gamma-\bm\gamma'\|_2^2/(2\sqrt{\vartheta})  + (1+\varepsilon)\lambda\varrho\|(\bm\beta-\bm\beta')_{\mathcal J}\|_1$ holds  for any $\bm\beta,\bm\beta',\bm\gamma,\bm\gamma'$ satisfying $ \| \bsbg\|_0 \le  \vartheta  o^*  , \| \bm\gamma'\|_0\le o^*, $ where $\mathcal J = \{j:\beta_j^{*}\ne 0\}$,  $J = |\mathcal J|$, and $\varepsilon, \delta$ are positive constants.
\end{manualtheorem}

For    resistant lasso, where $l(\bsbX \bsbb+\bsbg, \bsby) =\| \bsby - \bsbX \bsbb\|_2^2/2$, the regularity condition is implied by $K  \sqrt J \| \bsbX  \bsbb +  \bsbg\|_2    + \varrho\| \bsbb_{\mathcal J^c}\|_1 \ge\|\bsbg\|_2^2/(2\lambda\sqrt{\vartheta})+(1+\varepsilon)\varrho\| \bm\beta_{\mathcal J}\|_1$ for some (redefined) $K\ge 0$, where  $\| \bsbg\|_0 \le  (1+\vartheta)  o^*$ and  $\varepsilon$ is some positive constant, cf. Theorem \ref{thm:fixed-point-lasso}.

\begin{proof}
First, from the proof of   Theorem \ref{th:Theta_equation}, 
any coordinatewise minimum point     $(\hat{\bm\beta},\hat{\bm\gamma})$  must obey
\begin{equation}
\begin{cases}
\bm\beta = \Theta_{\text{soft}}(\bsbb - \bsbX^T \nabla l(\bsbX \bsbb + \bsbg;\bm y) ; \lambda\varrho), \\
\bm\gamma = \Theta^{\#}( \bsbg -  \nabla l(\bsbX \bsbb + \bsbg;\bm y) ;q),
\end{cases}
\end{equation}
where $\Theta_{\text{soft}}(\bm\beta;\lambda) = \text{sign}(\bm\beta)\circ(|\bm\beta|-\lambda)_+$ is the soft-thresholding and $\circ$ is the component-wise multiplication. 

Then, similar to the proof of Lemma  \ref{bcdsurro}, we can prove
\begin{equation}
(\hat{\bm\beta},\hat{\bm\gamma}) \in \mathop{\arg\min}_{\bm\beta,\bm\gamma} g(\bm\beta,\bm\gamma;\bm\beta^-,\bm\gamma^-)|_{\bm\beta^-=\hat{\bm\beta},\bm\gamma^-=\hat{\bm\gamma}} \text{ s.t. }\|\bm\gamma\|_0\le q,
\end{equation}
where $g(\bm\beta,\bm\gamma;\bm\beta^-,\bm\gamma^-) = l(\bm X\bm\beta^-+\bm\gamma^-) + \langle\nabla l(\bm X\bm\beta^-+\bm\gamma^-), \bm X\bm\beta-\bm X\bm\beta^- + \bm\gamma-\bm\gamma^-\rangle + \|\bm\beta-\bm\beta^-\|_2^2/2 + \|\bm\gamma-\bm\gamma^-\|_2^2/2 + \varrho\lambda\|\bm\beta\|_1$.  Then by
 Lemma 2 in \cite{She2016} and
 Lemma \ref{lemma:bound without design},  we have
\begin{equation}
\begin{split}
& \frac{1}{2}\|\bm\gamma^*-\hat{\bm\gamma}+\nabla l(\bm X\hat{\bm\beta}+\hat{\bm\gamma})\|_2^2 - \frac{1}{2}\|\nabla l(\bm X\hat{\bm\beta}+\hat{\bm\gamma})\|_2^2 + \lambda\varrho\|\bm\beta^*\|_1 - \lambda\varrho\|\hat{\bm\beta}\|_1\\
& + \frac{1}{2}\|\bm\beta^*-\hat{\bm\beta}+\bm X^T\nabla l(\bm X\hat{\bm\beta}+\hat{\bm\gamma})\|_2^2 - \frac{1}{2}\|\bm X^T\nabla l(\bm X\hat{\bm\beta}+\hat{\bm\gamma})\|_2^2\\
\ge\,& \frac{1}{2}\big[1-\mathcal L(\mathcal J(\bm\gamma^*),\mathcal J(\hat{\bm\gamma}))\big]\|\hat{\bm\gamma}-\bm\gamma^*\|_2^2 + \frac{1}{2}\|\hat{\bm\beta}-\bm\beta^*\|_2^2,
\end{split}
\end{equation}
or
\begin{equation} \label{lasso-1}
\begin{split}
2\bar{\bm\Delta}_l(\bm X\hat{\bm\beta} +\hat{\bm\gamma}, \bm X\bm\beta^* +\bm\gamma^*) \le   \frac{1}{2\sqrt{\vartheta}}\|\hat{\bm\gamma}-\bm\gamma^*\|_2^2 + \lambda\varrho\|\bm\beta^*\|_1 - \lambda\varrho\|\hat{\bm\beta}\|_1
& + \langle\bm\epsilon, \bar{\bm X}\hat{\bar \bsbb }-\bar{\bm X}\bar\bsbb^*\rangle.
\end{split}
\end{equation}
Next we try to   bound the  stochastic term in a proper way to merge with the $\bsbb$-penalties. Let  $P_H(t; \lambda)= (-t^2/2+\lambda |t|)1_{|t|<\lambda} +(\lambda^2/2) 1_{|t|\geq \lambda}$ and $P_H(\bsbb ; \lambda) = \sum P_H(\beta_j; \lambda)$.
 The following lemma based on combined statistical and computational analysis is useful.

\begin{lemma} \label{lemma:phostochastic}
Define $\Gamma_{q} = \{ (\bm\beta,\bm\gamma): \bsbb \in\mathbb R^p, \bsbg \in \mathbb R^n,  \|\bm\gamma\|_0 \le q\}$ where     $0\le q \le n$. Then for any $\varrho \ge \| \bsbX\|_2$, there exist universal constants $A_0, A_1, C, c, c_1, c_2>0$ such that for any $a\ge 2b>0$, the following event
\begin{align} \nonumber
\sup_{(\bm\beta,\bm\gamma)\in \Gamma_{q}} \{2\langle \bsbeps, \bsbX  \bsbb + \bm\gamma \rangle - \frac{1}{a} \|\bsbX  \bsbb +\bm\gamma\|_2^2 - \frac{1}{b} P_{H}(\varrho\bsbb;\lambda) - aA_0\sigma^2 q\log(en/q)\} \ge  a  \sigma^2 t
\end{align}
occurs with probability at most $ C \exp(-ct)$, where $\lambda = A  \sigma\sqrt{\log(ep)}, A = \sqrt{ab} A_1$ and $t\geq 0$.
\end{lemma}

From Lemma \ref{lemma:phostochastic} and  $\| \hat \bsbg - \bsbg^*\|_0 \le (1+ \vartheta) o^*$, it is easy to see
\begin{equation}\label{lasso-2}
\begin{split}
&\langle\bm\epsilon, \bm X\hat{\bm\beta}-\bm X\bm\beta^*+\hat{\bm\gamma}-\bm\gamma^*\rangle\\
\le\,& \frac{1}{ a}\Breg_2 (\bar \bsbX \bar \bsbb^*, \bar \bsbX \hat {\bar \bsbb}) + \frac{1}{2b}P_H(\varrho(\hat{\bm\beta}- \bm\beta^*),\lambda)   + \frac{1}{2}\sigma^2aA_0o^*(1+\vartheta)\log\frac{en}{o^*} +a R
\end{split}
\end{equation}
with $\EE R \le C \sigma^2 $. Let  $\mathcal J=\{j:\beta_j^*\ne 0\}$.   Due to the sub-additivity of $\|\cdot\|_1$ and the fact that  $P_H(t; \lambda)\le \lambda |t|$,
\begin{equation} \label{lasso-3}
\begin{split}
&\lambda\varrho\|\bm\beta^*\|_1 - \lambda\varrho\|\hat{\bm\beta}\|_1 + \frac{1}{2b}P_H(\varrho(\hat{\bm\beta}-\bm\beta^*);\lambda)\\
\le\,& \lambda\varrho\{\|(\hat{\bm\beta}-\bm\beta^*)_{\mathcal J}\|_1 - \|(\hat{\bm\beta}-\bm\beta^*)_{\mathcal J^c}\|_1\} + \varepsilon \lambda\varrho\{\|(\hat{\bm\beta}-\bm\beta^*)_{\mathcal J^c}\|_1 + \|(\hat{\bm\beta}-\bm\beta^*)_{\mathcal J}\|_1\}\\
\le\,& \lambda\varrho\{(1+\varepsilon)\|(\hat{\bm\beta}-\bm\beta^*)_{\mathcal J}\|_1 - (1-\varepsilon)\|(\hat{\bm\beta}-\bm\beta^*)_{\mathcal J^c}\|_1\}\\
\le\,& \lambda\varrho\{(1+\varepsilon) \|(\hat{\bm\beta}-\bm\beta^*)_{\mathcal J}\|_1 - \|(\hat{\bm\beta}-\bm\beta^*)_{\mathcal J^c}\|_1\},
\end{split}
\end{equation}
where we set $b = 1/(2\varepsilon) > 0$.
Combining \eqref{lasso-1}, \eqref{lasso-2}, and  \eqref{lasso-3} gives
\begin{equation}
\begin{split}
&2\bar{\bm\Delta}_l(\bm X\hat{\bm\beta} +\hat{\bm\gamma}, \bm X\bm\beta^* +\bm\gamma^*)-\frac{1}{2a}\|\bm X\hat{\bm\beta} - \bm X\bm\beta^* +\hat{\bm\gamma}-\bm\gamma^*\|_2^2 \\
\le\,& \frac{1}{2\sqrt{\vartheta}}\|\hat{\bm\gamma}-\bm\gamma^*\|_2^2 + \frac{1}{2}\sigma^2aA_0o^*(1+\vartheta)\log\frac{en}{o^*}\\
& + \lambda\varrho\{(1+\varepsilon)\|(\hat{\bm\beta}-\bm\beta^*)_{\mathcal J}\|_1 - \|(\hat{\bm\beta}-\bm\beta^*)_{\mathcal J^c}\|_1\}+aR.
\end{split}
\end{equation}
Choosing $a = 2/\delta+1/\varepsilon,  A \ge \sqrt{ab} A_1$,  and using the   regularity condition, we obtain the conclusion.
\end{proof}

\paragraph{Proof of Lemma \ref{lemma:phostochastic}}
Define $l_H(\bm\beta,\bm\gamma,q) = 2\langle \bsbeps, \bsbX  \bsbb + \bm\gamma \rangle -   \|\bsbX  \bsbb +\bm\gamma\|_2^2/a -   P_{H}(\varrho\bsbb;\lambda) /b- aA_0\sigma^2 q\log(en/q)$. Let   $P_0(\bm\beta;\lambda) = \lambda^2\|\bm\beta\|_0/2$. Similarly, define $l_0(\bm\beta,\bm\gamma,q) =   2\langle \bsbeps, \bsbX  \bsbb + \bm\gamma \rangle - \|\bsbX  \bsbb +\bm\gamma\|_2^2 /a-  P_0(\bsbb;\lambda) /b- aA_0\sigma^2 q\log(en/q)$.  Let $\mathcal A_H = \{\sup_{(\bm\beta,\bm\gamma)\in\Gamma_{q}}l_H(\bm\beta,\bm\gamma,q)\ge at\sigma^2\}$ and $\mathcal A_0 = \{\sup_{(\bm\beta,\bm\gamma)\in\Gamma_{q}}l_0(\bm\beta,\bm\gamma,q)\ge at\sigma^2\}$.
The occurrence of $\mathcal A_H$ implies
\begin{equation} \label{l_H}
l_H(\bm\beta^o,\bm\gamma^o,q) \ge at\sigma^2
\end{equation}
for any $(\bm\beta^o,\bm\gamma^o)$ that solves
\begin{align}
\min_{\bm\beta,\bm\gamma:\|\bsbg\|_0\le q} \|\bm X\bm\beta+\bm\gamma\|_2^2/2 - a\langle\bm\epsilon, \bm X\bm\beta+\bm\gamma\rangle + (a/(2b))P_H(\varrho\bm\beta;\lambda). \label{lemma8optprob}
\end{align}
\begin{lemma} \label{lemma:P-H-comp}
Given any $\tau\ge \| \bsbX\|_2^2$, there always exists a globally optimal solution $\bm\beta^o$ to $\min_{\bm\beta}\|\bm y-\bsbX\bm\beta\|_2^2/2 + \tau P_H(\bm\beta;\lambda)$ such that for any $j:1\le j\le p$, either $\beta_j^o=0$ or $\beta_j^o\ge \lambda\sqrt{\tau}/\| \bsbX\|_2\ge \lambda$.
\end{lemma}

The lemma can be shown  by linearization and  the property of $P_H$; see  \cite{She2012}. 
From Lemma \ref{lemma:P-H-comp} and $a/ 2b \ge 1 \ge \| \bsbX\|_2^2/ \varrho^2$, \eqref{l_H} implies  the existence  of an optimal solution $(\bm\beta^o,\bm\gamma^o)$ to \eqref{lemma8optprob} such that  $l_0(\bm\beta^o,\bm\gamma^o,q) =l_H(\bm\beta^o,\bm\gamma^o,q)   \ge at\sigma^2$ and so  $\mathcal A_H \subset \mathcal A_0$. It suffices to study $\mathbb P(\mathcal A_0)$.

The remaining  part follows the lines of the proof of Theorem \ref{thm:fixed-point} and is omitted. \qed

\subsection{Proof of Theorem \ref{thm:minimax}}
\label{subsec:proofofminimax}

Assume    the  density  of $\bsby\in \mathcal Y^n$ given $\bsbe^* = \bar \bsbX \bar \bsbb^*$ is given by \begin{align}p_{\bar \bsbb^*}(\bsby)= \exp \{[\langle \bsby, \bar \bsbe^*\rangle - \langle \bsb1, b(\bar\bsbe^*) \rangle] /\sigma^2  \}\end{align} with respect to a base measure $\mu$ defined on $\mathcal Y^n$.
The corresponding loss is
 $l(\bsbe) = \langle 1, l_0(\bsbe)\rangle$ with  $\eta_i =  \bsbx_i^T \bsbb + \gamma_i = \bar \bsbx_i^T \bar\bsbb$ and
\begin{align}\label{l0glmass}
l_0 ( \eta_i; y_i) =  (-y_i \eta_i + b(\eta_i))/\sigma^2 .\end{align}
Assume the natural parameter space $\Omega = \{\bsb{\eta}\in \mathbb R^n: b(\bsb{\eta})<+\infty\}$
is open. Then  the loss corresponds to  a distribution in the regular exponential dispersion family with dispersion $\sigma^2$ and natural parameter $\eta_i$. In the Gaussian case, $l_0 (\eta_i)$ is $ (\eta_i  - y_i)^2/(2\sigma^2)$ up to an additive term independent of $\eta_i$. Consider a signal class
\begin{equation}
\mathcal B(s^*,o^*, M_\beta, M_\gamma) = \{(\bm\beta^*,\bm\gamma^*): \|\bm\beta^*\|_0\le s^*, \|\bm\gamma^*\|_0 \le o^*,  \|\bm\beta^*\|_\infty< M_\beta, \|\bm\gamma^*\|_\infty < M_\gamma\}
\end{equation}
where $0\le s^*\le p$, $0\le o^*\le n$, and  $+\infty\ge M_\beta, M_\gamma\ge 0$.  The following theorem implies
 Theorem \ref{thm:minimax} by setting $M_\beta = M_\beta =+\infty$ and assuming $   \underline \kappa_\beta / \overline \kappa_\beta   $ and  $  \kappa_\gamma     $ are  positive constants.
\begin{manualtheorem}{\ref{thm:minimax}'}\label{thm:minimax-gen}

In the  regular exponential dispersion family with $n\ge 2, p\ge 2, 1\le o^* \le n/2, 1\le s^*\le p/2,$ define
\begin{align}
P_\beta (s^*) = s^*\log(ep/s^*), \quad P_\gamma =  o^*\log(en/o^*).
\end{align}
Let $I(\cdot)$ be any nondecreasing   function with $I(0)=0, I\not\equiv 0$.
 (i) Suppose for some $\kappa_\beta, \kappa_\gamma > 0$    $  \breg_l (\bsb0,  \bsbX
\bsbb     )\sigma^2\allowbreak\le    \kappa_\beta \Breg_2(  \bsb0,    \bsbb    )$, $\forall   {\bm\beta}  :   \|\bm\beta \|_0 \le s^*, \|\bm\beta \|_\infty< M_\beta$, and  $  \breg_l (\bsb0,
\bsbg     )\sigma^2\allowbreak\le    \kappa_\gamma \Breg_2(  \bsb0,    \bsbg    )$, $\forall   \bsbg  :   \|\bsbg \|_0 \le o^*, \|\bsbg\|_\infty< M_\gamma$.
Then there exist positive constants $\tilde c, c$, depending on $I(\cdot)$ only, such that
\begin{equation}
\begin{split}
\inf_{(\hat{\bm\beta},\hat{\bm\gamma})}\,\sup_{(\bm\beta^*,\bm\gamma^*)\in \mathcal B(s^*,o^*,M_\beta, M_\gamma)} \mathbb E\Big [I\Big (\Breg_2  (
\bar \bsbb^*,    \hat{ \bar \bsbb} )/\big \{\tilde c \big[\min\{ \sigma^2  P_\beta (s^*)/\kappa_\beta, M_\beta^{2} s^{*}\} \\ +\min\{ \sigma^2  P_\gamma (o^*)/\kappa_\gamma, M_\gamma^{2} o^{*}\}\big ]\big \}\Big)\Big] \ge c >0,
\end{split}
\label{minimaxrobglm-gen-est}
\end{equation}
where $(\hat{\bm\beta},\hat{\bm\gamma})$ denotes any estimator of $(\bm\beta^*,\bm\gamma^*)$. (ii) Suppose $  \breg_l (\bsb{0},  \bsbX   \bsbb_1  )\sigma^2\allowbreak\le   \overline \kappa_\beta\Breg_2(   \bsb{0},  \bsbb_1      )$ and  $\underline \kappa_\beta\Breg_2(  \bsb0,    \bsbb_2     )\le  \Breg_2( \bsb0,  \bsbX
    \bsbb_2   ) $,    $\forall   {\bm\beta}_i  :    \| \bsbb_i\|_0 \le i s^*,  \|\bm\beta_i \|_\infty< M_\beta, i =1,2$,  and  $  \breg_l (\bsb{0},     \bsbg   )\sigma^2\allowbreak\le  \kappa_\gamma\Breg_2(   \bsb{0},  \bsbg)$,     $\forall   \bsbg   :    \| \bsbg \|_0 \le   o^* , \|\bsbg\|_\infty< M_\gamma$.   where $0\le \underline \kappa_\beta \le \overline \kappa_\beta   , 0\le   \kappa_\gamma$. Then  there exist positive constants $\tilde c, c$ such that
\begin{equation}
\begin{split}
\inf_{(\hat{\bm\beta},\hat{\bm\gamma})}\,\sup_{(\bm\beta^*,\bm\gamma^*)\in \mathcal B(s^*,o^*,M_\beta, M_\gamma)} \mathbb E\Big [I\Big (\Breg_2  (\bar \bsbX   \bar \bsbb^*,   \bar \bsbX \hat{ \bar \bsbb} )/\big \{\tilde c \big[\min\{ ( \underline \kappa_\beta / \overline \kappa_\beta   )\sigma^2  P_\beta (s^*), \underline \kappa_\beta
M_\beta^{2} s^{*}\} \\ +\min\{ \sigma^2  P_\gamma (o^*)/\kappa_\gamma,  M_\gamma^{2} o^{*}\}\big ]\big \}\Big)\Big] \ge c >0.
\end{split}
\label{minimaxrobglm-gen-pred}
\end{equation}
\end{manualtheorem}


\begin{proof} First we introduce a lemma \cite[Lemma 3(iii)]{SheBregman}.
\begin{lemma} \label{lemKLBreg} For any  $p_{\bar\bsbb_1}, p_{\bar\bsbb_2}$ in the regular exponential dispersion family,  the Kullback-Leibler divergence of   $p_{\bar{\bm\beta}_2}$ from  $p_{\bar{\bm\beta}_1}$, defined by $\mathcal K(p_{\bar{\bm\beta}_1},p_{\bar{\bm\beta}_2})=\int  p_{\bar{\bm\beta}_1}\log( p_{\bar{\bm\beta}_1}/p_{\bar{\bm\beta}_2}) \rd \nu$, satisfies\begin{equation} \nonumber
\mathcal K(p_{\bar{\bm\beta}_1},p_{\bar{\bm\beta}_2}) = \breg_l ( \bar \bsbX \bar\bsbb_2, \bar \bsbX \bar\bsbb_1),
\end{equation}
where $\breg_l(\cdot, \cdot)$ is the generalized Bregman divergence notation introduced in Section \ref{sec:theory}.
\end{lemma}
To prove the desired rate for estimation,
we make a discussion in two cases.

\textit{Case (i)} $   \min\{ \sigma^2  P_\gamma (o^*)/\kappa_\gamma, M_\gamma^{2} o^{*}\}\le \min\{ \sigma^2  P_\beta (s^*)/\kappa_\beta, M_\beta^{2} s^{*}\}$. Consider a signal subclass
\begin{equation}\nonumber
\mathcal B^1 = \{\bar{\bm\beta} = [\bm\beta^T, \bm 0^T]^T: \beta_j\in\{0,\tau R\}, \|\bm\beta\|_0\le s^*\},
\end{equation}
where $$R = \sigma (\log(ep/s^*))^{1/2}/{\kappa_\beta}^{1/2}\wedge M_\beta$$ and $1>\tau>0$ is a small constant to be chosen later.
Clearly, $\mathcal B^1\in\mathcal B(s^*,o^*,M_\beta, M_\gamma)$. By Stirling's approximation, $\log |\mathcal B^1| \ge  \log {p \choose s^*} \ge s^*\log(p/s^*) \ge cs^*\log(ep/s^*)$
for some universal constant $c$.

Let $\rho(\bar{\bm\beta}_1,\bar{\bm\beta}_2) = \|\bar{\bm\beta}_1-\bar{\bm\beta}_2\|_0$, the Hamming distance between $\bar{\bm\beta}_1$ and $\bar{\bm\beta}_2$. By Lemma A.3 in   \cite{Rigollet11}, there exists a subset $\mathcal B^{10} \subset \mathcal B^1$ such that $\bar{\bm\beta}_0=\bsb{0}\in  \mathcal B^{10}$ and
\begin{equation}\nonumber
\log |\mathcal B^{10}| \ge c_1s^*\log(ep/s^*), \rho(\bar{\bm\beta}_1,\bar{\bm\beta}_2) \ge c_2s^*,\forall \bar{\bm\beta}_1,\bar{\bm\beta}_2\in \mathcal B^{10}, \bar{\bm\beta}_1\ne \bar{\bm\beta}_2
\end{equation}
for some universal constants $c_1,c_2>0$. Then \begin{align}\label{minimax_ineq-1}\|\bar{\bm\beta}_1-\bar{\bm\beta}_2\|_2^2 = \tau^2R^2\rho(\bar{\bm\beta}_1,\bar{\bm\beta}_2) \ge c_2\tau^2R^2s^*\end{align} 
for any $\bar{\bm\beta}_1,\bar{\bm\beta}_2\in\mathcal B^{10}$, $\bar{\bm\beta}_1\ne \bar{\bm\beta}_2$.

By Lemma \ref{lemKLBreg} and the regularity condition, for any $\bar{\bm\beta}  \in\mathcal B^{10}$, we have
\begin{equation}\nonumber
\mathcal K(p_{\bar{\bm\beta}},p_{\bar\bsbb_0} ) = \breg_l (  \bar \bsbX {\bar \bsbb}_0,  \bar \bsbX \bar\bsbb) \le n {\tau^2}\kappa_\beta R^2s^*/\sigma^2.
\end{equation}
 Therefore,
\begin{equation} \label{minimax_ineq-2}
\frac{1}{|\mathcal B^{10}|-1}\sum_{\bar \bsbb\in\mathcal B^{10}\setminus \{\bar{\bm\beta}_0\}}\mathcal K(  p_{\bar{\bm\beta} },p_{\bar \bsbb_0}) \le {\tau^2} \kappa_\beta s^*\log(ep/s^*).
\end{equation}

Combining \eqref{minimax_ineq-1} and \eqref{minimax_ineq-2} and choosing a sufficiently small value for $\tau$, we can apply Theorem 2.7 of \cite{Tsybakov2008} to get the desired lower bound.

\textit{Case (ii)} $   \min\{ \sigma^2  P_\gamma (o^*)/\kappa_\gamma, M_\gamma^{2} o^{*}\}\ge \min\{ \sigma^2  P_\beta (s^*)/\kappa_\beta, M_\beta^{2} s^{*}\}$. Define a signal subclass
\begin{equation}\nonumber
\mathcal B^2= \{\bar{\bm\beta} = [\bm 0^T,\bsbg^T]^T: \gamma_i\in\{0,\tau R\}, \|\bm\gamma\|_0\le  o^*\},
\end{equation}
where $R =\sigma (\log(en/o^*))^{1/2}/\tau^{1/2}\wedge M_\gamma$ and $1>\tau>0$ is a small constant. The afterward treatment is similar to (i). The details are omitted.

The proof for the lower bound of $\| \bar \bsbX   \bar \bsbb^* -   \bar \bsbX \hat{ \bar \bsbb}\|_2^2$ follows similar lines.
\end{proof}
\subsection{Proof of Theorem \ref{thm:tuning} } \label{proof:tuning}


Let   $J(\bm\gamma) = |\mathcal J(\bm\gamma)| = \| \bsbg\|_0$ and $\mathcal J(\bm\gamma)$ is the support of $\bm\gamma$, i.e., $\mathcal J(\bm\gamma) = \{i:\gamma_i\ne 0\}$. The optimality of $(\hat{\bm\beta},\hat{\bm\gamma})$ implies that
\begin{equation}
\begin{split}
&\bm\Delta_l(\bm X\hat{\bm\beta}+\hat{\bm\gamma},\bm X\bm\beta^*+\bm\gamma^*)
\le   AP_o(\bm\beta^*,\bm\gamma^*)-AP_o(\hat{\bm\beta},\hat{\bm\gamma}) + \langle\bm\epsilon, \bar{\bm X}\hat{\bar{\bm\beta}}-\bar{\bm X}\bar{\bm\beta}^*\rangle.
\end{split}
\end{equation}
The stochastic term $\langle\bm\epsilon, \bar{\bm X}\hat{\bar{\bm\beta}}-\bar{\bm X}\bar{\bm\beta}\rangle$ can be decomposed and bounded in a similar  way as in the proof of Theorem  \ref{thm:fixed-point}. The difference is to  use the union bound to show the $\EE \sup_{s\le p, o\le n} R_{s, o}^2 \le C$.  Take  the first term     $\langle\bm\epsilon, I\rangle$ from the decomposition as an example:
\begin{align*}
&\langle \bm\epsilon, \mathcal P_{\bar{\bm X}_{\mathcal J(\bar{\bm\beta}^*)}} \bar{\bm X}\bm\Delta\rangle  - \frac{1}{a} \| \mathcal P_{\bar{\bm X}_{\mathcal J(\bar{\bm\beta}^*)}}\bar{\bm X}\bm\Delta\|_2^2 - bL P_o (J(  \bsbb^*), J(\bsbg^*) )  \\
\le\,&  \frac{1}{a'} \| \mathcal P_{\bar{\bm X}_{\mathcal J(\bar{\bm\beta}^*)}} \bar{\bm X}\bm\Delta \|^2_2 + \frac{a'}{4}  \sup_{s\le p, o\le n} \sup_{\bm\Delta\in\Gamma_{s ,o }} [\langle\bm\epsilon, \bm\Delta\rangle - 2\{{ ({b}/{a})LP_{o} (s , o )}\}^{1/2} ]_+^2 \\
\equiv\,&  \frac{1}{a'} \| \mathcal P_{\bar{\bm X}_{\mathcal J(\bar{\bm\beta}^*)}} \bar{\bm X}\bm\Delta \|^2_2 + \frac{a'}{4}  \sup_{s\le p, o\le n} R_{s, o}^2.
\end{align*}
When $s=o=0$, $R_{0,0}= 0 $. When $s\ge 1$ and $o\ge 1$, for any $t\ge 0$,
if $4b/a$ is a constant greater than $1$, \begin{equation} \label{Rbound}
\begin{split}
&\mathbb P(\sup_{1\le o \le n, 1\le s \le p}R_{s, o} > t\sigma)\\
\le\,&  \sum_{s=1}^{p} \sum_{o=1}^{n} \mathbb P\bigg(\sup_{\bm\Delta\in\Gamma_{s, o}}\langle\bm\epsilon, \bm\Delta\rangle - \sqrt{LP (s,o)} > t\sigma + 2\sigma\sqrt{\frac{b}{a}LP (s,o)} - \sqrt{LP  (s,o)}\bigg)\\
\le\,&   C\exp(-ct^2)\sum_{s=1}^{p} \sum_{o=1}^{n} \exp[-c(2\sqrt{b/a}-1)^2 P  (s,o)]\\
\le\,&  C\exp(-ct^2)\exp(-c (\log n+ \log p))\sum_{s=1}^p \sum_{o=1}^n \exp(-c(s + o))\\
\le\, & C \exp(-ct^2 ) (np)^{-c} ,
\end{split}
\end{equation}
where the last inequality is due to the sum of geometric series. Similarly, when $s = 0 $, $\EP(\sup_{0\le o\le n}R_{0,o} > t \sigma)\le C \exp(-ct^2) n^{-c}$,  and when $o=0$, $\EP(\sup_{0\le s \le p} R_{s,0} > t \sigma)\le C \exp(-ct^2) p^{-c}$ for any $0\le s \le p$.   Hence
$
\EE \sup_{s, o} R_{s, o}^2 \le C.
$
To summarize, we obtain that for any constants $a,b,a',b' > 0$ satisfying $4b>a$,
\begin{equation}
\mathbb E \langle\bm\epsilon, \bar{\bm X}\hat{\bar{\bm\beta}}-\bar{\bm X}\bar{\bm\beta}\rangle
\le (\frac{1}{a}+\frac{1}{a'})\|\bm X\hat{\bm\beta}-\bm X\bm\beta+\hat{\bm\gamma}-\bm\gamma\|_2^2 + 2bL[P (\hat{\bm\beta},\hat{\bm\gamma})+P (\bm\beta,\bm\gamma)] + a'C.
\end{equation}
Using the regularity condition and choosing constants $a,a',b,b'$ and $A$ sufficiently large such that $1/a+1/a'=\delta/4$, $4b>a$ and $A = A_0 + 2bL$, we   obtain the   error bound.

\subsection{Proof of Theorem \ref{thm:tuning_log_form}}
\label{subsec:sfpic-high}
We prove a   general theorem in possibly high dimensions  where $\bsbb$ is also desired to be sparse.

\begin{manualtheorem}{\ref{thm:tuning_log_form}'}\label{thm:tuning_log_form-largep}
Let $\bm y = \bm X\bm\beta^* + \bm\gamma^*+\bm\epsilon$, where  $\epsilon_i$ are independent sub-Gaussian$( 0,\sigma^2)$ and $\EE \epsilon_i^2 =  \sigma_i^2 = c_i \sigma^2$ with $c_i$  some positive constants and $\sigma^2$ unknown. Let $l_0(\bsbX \bsbb + \bsbg; \bsby)=\|\bm y - \bm X\bm\beta - \bm\gamma\|_2^2/2$. Define $P(\bm\beta,\bm\gamma) =  J(\bm\beta)\log(ep/J(\bm\beta))   + J(\bm\gamma)\log(en/J(\bm\gamma))$.
 Assume the true model is parsimonious in the sense that $P (\bm\beta^*,\bm\gamma^*)\le n/A_0$ for some constant $A_0 > 0$. Let $\delta(\bm\beta,\bm\gamma) = AP(\bm\beta,\bm\gamma)/n$ where $A$ is a positive constant satisfying $A<A_0$, and so $\delta(\bm\beta^*,\bm\gamma^*)<1$. Then for sufficiently large values of $A_0$ and $A$, any $(\hat{\bm\beta},\hat{\bm\gamma})$ that minimizes
$
\log l_0(\bsbX \bsbb + \bsbg; \bsby) + \delta(\bm\beta,\bm\gamma) \text{ s.t. } \delta(\bm\beta,\bm\gamma)<1
$ 
satisfies
\begin{equation}
\Breg_2 (\bar \bsbX \hat {\bar \bsbb}, \bar \bsbX \bar \bsbb^*  )     \lesssim \sigma^2P (\bm\beta^*,\bm\gamma^*)
\end{equation}
with probability at least $1-C' n^{-c'} -C    n^{-c(1 \wedge o^*)} p^{-c(1 \wedge s^*)}$ for some constants $c , c', C, C'>0$.
\end{manualtheorem}

\begin{proof}
Let $h(\bm\beta,\bm\gamma; A) = 1/(n-AP(\bm\beta,\bm\gamma))$. It follows from $1/(1-\delta)\ge \exp(\delta)$ for any $0\le \delta<1$ and $\exp(\delta)\ge 1/(1-\delta/2)$ for any $0\le\delta<2$ that
$n\|\bm y - \bm X\hat{\bm\beta} - \hat{\bm\gamma}\|_2^2 h(\hat{\bm\beta},\hat{\bm\gamma};A/2)
\le \|\bm y-\bm X\hat{\bm\beta}-\hat{\bm\gamma}\|_2^2 \exp[\delta(\hat{\bm\beta},\hat{\bm\gamma})]
\le \|\bm y-\bm X\bm\beta^*-\bm\gamma^*\|_2^2 \exp[\delta(\bm\beta^*,\bm\gamma^*)]
\le  \|\bm y-\bm X\bm\beta^*-\bm\gamma^*\|_2^2 h(\bm\beta^*,\bm\gamma^*;A)n$.

Since $h(\hat{\bm\beta},\hat{\bm\gamma};A/2)>0$, we have $$\|\bm y-\bm X\hat{\bm\beta}-\hat{\bm\gamma}\|_2^2\le\|\bm y-\bm X\bm\beta^*-\bm\gamma^*\|_2^2h(\bm\beta^*,\bm\gamma^*;A)/h(\hat{\bm\beta},\hat{\bm\gamma};A/2).$$ It follows that
\begin{equation*}
\begin{split}
&\|\bm X\hat{\bm\beta}-\bm X\bm\beta^*+\hat{\bm\gamma}-\bm\gamma^*\|_2^2 \\
\le\,& \|\bm\epsilon\|_2^2[h(\bm\beta^*,\bm\gamma^*;A)/h(\hat{\bm\beta},\hat{\bm\gamma};A/2)-1]+ 2\langle\bm\epsilon,\bm X\hat{\bm\beta}-\bm X\bm\beta^*+\hat{\bm\gamma}-\bm\gamma^*\rangle\\
\le\,& \frac{A\|\bm\epsilon\|_2^2}{n\sigma^2-A\sigma^2P (\bm\beta^*,\bm\gamma^*)}\sigma^2P (\bm\beta^*,\bm\gamma^*)-\frac{A\|\bm\epsilon\|_2^2}{2n\sigma^2}\sigma^2P (\hat{\bm\beta},\hat{\bm\gamma}) +2\langle\bm\epsilon,\bm X\hat{\bm\beta}-\bm X\bm\beta^*+\hat{\bm\gamma}-\bm\gamma^*\rangle.
\end{split}
\end{equation*}

The proof of Theorem \ref{thm:tuning} gives a high-probability bound for the stochastic term: for any constant $a,b,a' > 0$ satisfying $4b>a$,
\begin{equation}
\begin{split}
& 2\langle\bm\epsilon,\bm X\hat{\bm\beta}-\bm X\bm\beta^*+\hat{\bm\gamma}-\bm\gamma^*\rangle \\
\le\,& 2(\frac{1}{a}+\frac{1}{a'})\|\bm X\hat{\bm\beta}-\bm X\bm\beta^*+\hat{\bm\gamma}-\bm\gamma^*\|_2^2 + 4bL\sigma^{2}[P(\hat{\bm\beta},\hat{\bm\gamma})+P(\bm\beta^*,\bm\gamma^*)]
\end{split}\label{intermed0}
\end{equation}
with probability at least $1-C    n^{-c(1 \wedge o^*)} p^{-c(1 \wedge s^*)}$ for some $c , C>0$.

Assume $c_0 \sigma^2 \le \EE \epsilon_i^2 \le C_0 \sigma^2$ with $c_0, C_0$ positive constants and let $\varepsilon$ and $\varepsilon'$ be two constants satisfying $0<\varepsilon<1, \varepsilon'>0$. On $\mathcal A=\{(c_{0}-\varepsilon)n\sigma^2 \le \|\bm\epsilon\|_2^2 \le ( C_{0}+\varepsilon')n\sigma^2\}$, we have
\begin{equation*}
\begin{split}
&\frac{A\|\bm\epsilon\|_2^2}{n\sigma^2-A\sigma^2P(\bm\beta^*,\bm\gamma^*)}\sigma^2P(\bm\beta^*,\bm\gamma^*)-\frac{A\|\bm\epsilon\|_2^2}{2n\sigma^2}\sigma^2P(\hat{\bm\beta},\hat{\bm\gamma})\\
\le\,& \frac{(C_{0}+\varepsilon')AA_0}{A_0-A}\sigma^2P (\bm\beta^*,\bm\gamma^*) - \frac{(c_{0}-\varepsilon)A}{2}\sigma^2P (\hat{\bm\beta},\hat{\bm\gamma}).
\end{split}
\end{equation*}
 With $A_0$ large enough, we can choose $a,a',b,A$ such that $1/a+1/a'<1/2$, $4b>a$ and $8bL/ c_{0}(1-\varepsilon)\le A$. From the Hanson-Wright inequality  \citep{rudelson2013} $\mathcal A$ occurs with probability at most $c_2'\exp(-c_2n)$, where $c_2,c_2'$ are dependent on constants $ \varepsilon, \varepsilon'$. The conclusion results.
\end{proof}

\begin{remark}
For robust regression in low dimensions,       applying    Theorem  \ref{thm:tuning_log_form-largep} to   the reduced model     $\bsbU_{\perp}^T \bsby = \bsbU_{\perp}^T \bsbg^* + \bsbeps'$ gives Theorem \ref{thm:tuning_log_form}, where   $\bm U_\perp\in\mathbb R^{n\times (n-r(\bsbX))}$ is the orthogonal complement of $\bsbU$ that is from the SVD: $\bsbX = \bsbU \bsbD \bsbV^T$, and $\bm\epsilon'=\bm U_\perp^T\bm\epsilon$. The proof follows the same lines; in particular, when applying the Hanson-Wright inequality,  note that $\EE [\| \epsilon'\|_2^2 ] = Tr\{ \mbox{diag}\{\sigma_i^2\}\bsbU_\perp \bsbU_\perp^T \}\in [c_0 (n - r(\bsbX))\sigma^2, C_0 (n - r(\bsbX)) \sigma^2] $ and $\| \bsbU_\perp \bsbU_\perp ^T\|_F^2  = n  - r(\bsbX)$ and $\| \bsbU_\perp \bsbU_\perp ^T\|_2 = 1$.

\end{remark}

\subsection{Error bounds of optimal solutions}
\label{proof:global_rate}

This part points out that the error rate remains the same if $(\hat{\bm\beta},\hat{\bm\gamma})$ is {globally} optimal,  but the regularity condition gets  relaxed. For simplicity, we drop the $\ell_2$ penalty terms.

\begin{thm} \label{thm:global_rate}
Let $(\hat{\bm\beta}, \hat{\bm\gamma})$ be any globally optimal solution of $ \min_{  \|\bsbb\|_0 \le q_\beta, \|\bsbg\|_0 \le q_\gamma  }\allowbreak l(\bm X\bm\beta + \hat \bsbg; \bm y)$    with  $\|\hat{\bm\gamma}\|_0 = q_\gamma$, $\|\hat{\bm\beta}\|_0 = q_\beta$,     $  q_\gamma= \vartheta o^*,   q_\beta = \vartheta s^* $ and $\vartheta\ge 1$.
 Assume that there exists some  $\delta>0$ such that \begin{align}\label{regcondglobal}
 (\bm  \Delta_l - \delta\Breg_2 )(\bar \bsbX  {\bar \bsbb}, \bar \bsbX  {\bar \bsbb}') \ge 0
\end{align} holds for any   $\|\bm\beta\|_0\le \vartheta s^*, \|\bm\beta'\|_0 \le  s^*$, $\|\bm\gamma\|_0 \le \vartheta o^*, \| \bm\gamma'\|_0 \le  o^*$. Then the     following error bound  holds
\begin{equation} \label{oracle_ineq}
\EE [\Breg_2 (\bar \bsbX \hat {\bar \bsbb}, \bar \bsbX \bar \bsbb^*  ) ]
 \lesssim \frac{\vartheta}{\delta^2}\sigma^2(s^*\log\frac{ ep}{s^*} + o^*\log\frac{en}{o^*})+ \frac{\sigma^2}{\delta^2}.
\end{equation}
\end{thm}

\eqref{regcondglobal}, apart from replacing
$2\bar{\bm \Delta}_l$   by  ${\bm \Delta}_l$,  removes  the term on the right-hand side of
\eqref{fixed-point-condition}.   It is easy to see that  if  $l$ is $\mu$-strongly convex, 
then \eqref{regcondglobal}   is true with $\delta = \mu$.
In particular, when $l_0 (\eta ; y) = (\eta - y)^2/2$,  the condition holds trivially ($\delta=1$).


Under a slightly stronger regularity condition than that used in Theorem \ref{thm:fixed-point},
we get an estimation error bound.
\begin{thm} \label{th:esterr}
Let $(\hat{\bm\beta}, \hat{\bm\gamma})$ be an  A-estimator   satisfying  $\|\hat{\bm\gamma}\|_0 = q_\gamma$, $\|\hat{\bm\beta}\|_0 = q_\beta$ with $   q_\gamma= \vartheta o^* ,    q_\beta = \vartheta s^* $ and $\vartheta\ge 1$. Assume  $\bsbb^* \ne \bsb0$.
 Then, with  \eqref{fixed-point-condition}   replaced by $
 (1-\epsilon) (2\bar{\bm\Delta}_l -  \delta\Breg_2 )(\bar \bsbX_{\rho} \bar \bsbb , \bar \bsbX_{\rho} \bar \bsbb ')
 \ge \frac{1}{\sqrt{\vartheta}} \Breg_2 (\bar \bsbb,   \bar \bsbb ')
$   $\forall \bar \bsbb,\bar \bsbb'$:  $\|\bm\beta\|_0\le \vartheta s^*, \|\bm\beta'\|_0 \le  s^*$, $\|\bm\gamma\|_0 \le \vartheta o^*, \| \bm\gamma'\|_0 \le  o^*$  for some $\epsilon, \delta >0$,  with probability at least $1 - C p^{-c}$,
$ 
\rho  \|\hat{\bm\beta} - \bm\beta^* \|_2^2+ \|\hat{\bsbg} - \bsbg^* \|_2^2 \lesssim \frac{\sigma^2}{\delta \varepsilon }  \{ \vartheta^{3/2} [s^*\log(ep/s^*) + o^* \log (en/o^*)]\}.
$
In particular, we have the estimation error bound    
\begin{equation}  \label{esterrbnd}
 \|\hat{\bm\beta} - \bm\beta^* \|_2^2 \lesssim \frac{\sigma^2}{\delta \varepsilon } \frac{ \vartheta^{3/2} [s^*\log(ep/s^*) + o^* \log (en/o^*)]} {\rho}
\end{equation}
 with probability  $1 - C p^{-c}$ under \eqref{fixed-point-condition} with   $\bar \bsbX_\rho$ redefined as $ [  {\bsbX}/{\sqrt {\rho(1+\epsilon})} \ \ \bsbI]$.
\end{thm}

 When $\vartheta, \epsilon, \delta$ are constants and $\rho \ge c
n$, the estimation error bound in \eqref{esterrbnd} is of the order $ \sigma^2 \{ s^*\log(ep/s^*) + o^* \log (en/o^*)\}/n$.

\begin{proof}
To prove Theorem \ref{thm:global_rate}, by definition, $\hat{\bar{\bm\beta}} = [\hat{\bm\beta}^T, \hat{\bm\gamma}^T]^T$ satisfies $l(\bar{\bm X}\hat{\bar{\bm\beta}}) \le l(\bar{\bm X} {\bar{\bm\beta}}^*)   $ or equivalently,
\begin{equation}
\bm\Delta_l(\bar{\bm X}\hat{\bar{\bm\beta}},\bar{\bm X}\bar{\bm\beta}^*) \le  \langle\bm\epsilon, \bar{\bm X}\hat{\bar{\bm\beta}}-\bar{\bm X}\bar{\bm\beta}^*\rangle.
\end{equation}
Treating  the stochastic term   in  the same way as in
 the proof of Theorem \ref{thm:fixed-point},  we obtain
\begin{equation}
\mathbb E\langle\bm\epsilon, \bar{\bm X}\hat{\bar{\bm\beta}}-\bar{\bm X}\bar{\bm\beta}^*\rangle \le (\frac{1}{a}+\frac{1}{a'})\|\bar{\bm X}\hat{\bar{\bm\beta}}-\bar{\bm X}\bar{\bm\beta}^*\|_2^2 + 3bL \sigma^2 P(q_\beta, q_\gamma) + a'C\sigma^2,
\end{equation}
for any $a,a',b >0$ satisfying $4b>a$. The regularity condition implies
\begin{equation}
(\delta/2)\|\bar{\bm X}\hat{\bar{\bm\beta}}-\bar{\bm X}\bar{\bm\beta}^*\|_2^2  \le \bm\Delta_l(\bar{\bm X}\hat{\bar{\bm\beta}}, \bar{\bm X}\bar{\bm\beta}^*).
\end{equation}
Combining the above three inequalities and choosing $a,a',b$ such that $a = a' = 8/\delta$ and $b=4/\delta$ gives the desired result.

To prove   Theorem \ref{th:esterr}, recall that  from the proof
of Theorem \ref{thm:fixed-point}, we obtain
\begin{align*}
 2 \bar{\bm\Delta}_l(\bm X\hat{\bm\beta}+\hat{\bm\gamma}, \bm X\bm\beta^*+\bm\gamma^*) \le  \frac{1}{2\sqrt{\vartheta}}(\|\hat{\bm\gamma}-\bm\gamma^*\|_2^2+\rho\|\hat{\bm\beta}-\bm\beta^*\|_2^2)\\+  (\frac{1}{a}+\frac{1}{a'})\|\bar{\bm X}\hat{\bar{\bm\beta}}-\bar{\bm X}\bar{\bm\beta}^*\|_2^2+ 3bL \vartheta \sigma^2 P_o(s^*, o^*)
\end{align*}
with probability at least  $1 - C n^{-c(1 \wedge o^*)} p^{-c}$. Using the regularity condition, and setting $a = a' = 4/\delta$ and $b = 2/ \delta$, we get the estimation error bound \eqref{esterrbnd} with high probability (details omitted).
\end{proof}

\subsection{General noise and stochastic breakdown}
\label{subsec:extra}
Given a random variable $X$,  its Orlicz $\psi$-norm is defined by $$
\| X\|_{\psi} = \inf \big\{ M>0: \EE \psi(\frac{|X|}{M})\le 1\big\},
$$ where  $\psi$ is a strictly increasing convex  function on $[0, +\infty)$ with $\psi(0)=0$; see, e.g., \cite{van1996weak}. Similar to the definition of sub-Gaussian random vectors (cf. Definition \ref{def:subgauss}), we say   that a random vector   $\bsb{\epsilon}$ has its   $\psi$-norm bounded above by $\sigma$ if       $$\|\langle \bsb{\epsilon}, \bsba\rangle\|_{\psi}\le \sigma\|\bsba\|_2 $$  for any  vector $\bsba$.  The components of $\bsb{\epsilon}$ are not required to be independent.

Some commonly used $\psi$-norms in statistics are the $L_p$ norms, $\psi(x) = x^p$ with      $ p\ge 1$, which  means $|X|$ possesses a finite $p$-th moment, and $\psi_p$ norms,  $\psi(x) = \exp(x^p) -1 $, covering  sub-Gaussian    ($p=2$) and sub-exponential ($p=1$) random variables.   By Markov's inequality, that  $X$ has  a finite $\psi$-norm  $\sigma$ implies a tail probability bound
\begin{align*}\EP (|X|> t) \le \frac{ 2}{\psi(t/\sigma)+1}, \quad \forall t>0\end{align*}
which encompasses diverse heavy/light tail decays.

In the following theorem, instead of restricting our attention to a   sub-Gaussian type effective noise $\bsbeps$ (cf.  \eqref{noise-def}),  which \textit{is} sensible in many applications especially when  a Lipschitz continuous loss is in use, we make a more general assumption that $\bsbeps$ has a finite Orlicz $\psi$-norm bounded above by $\sigma$ and  $\psi(x) \ge c x^2$ (i.e., the second moment exists),  together with the regularity condition $\lim\sup_{x,y\rightarrow \infty} \psi(x) \psi(y)/\psi(cxy)<\infty$  \citep{van1996weak}   for some constant $c>0$  that is satisfied by say the $L_p$ norms and $\psi_p$ norms. (Note that the  $\psi$ here is not the same psi function to define an M-estimator.)
\begin{thm}\label{theorem:generalnoise}
Let  $\hat{\bar{\bm\beta}} $ be a solution to $ \min_{ \bar{\bm\beta}:\|\bm\beta\|_0 \le q_\beta,   \|\bm\gamma\|_0 \le q_\gamma } l(\bar{\bm X} \bar{\bm\beta}; \bm y)$
with  $ 0\le  q_\gamma\le n, 0\le   q_\beta \le p $.
  Assume that there exists some $\delta >0$ such that
$
  ({\bm\Delta}_l -  \delta\Breg_2 )(\bar \bsbX \bar \bsbb , \bar \bsbX  \bar \bsbb ')
 \ge 0   $
for sparse $\bar \bsbb, \bar \bsbb'$ satisfying  $\|\bm\beta\|_0\le q_\beta, \|\bm\beta'\|_0 \le   q_\beta$, $\|\bm\gamma\|_0 \le q_\gamma, \| \bm\gamma'\|_0 \le  q_\gamma$. Then the following   \emph{oracle inequality} holds for any $\bar\bsbb: \| \bsbb\|_0\le q_\beta, \| \bsbg\|_0\le q_\gamma$
\begin{equation} \label{fixed-point-rate-gennoise}
\EE [\Breg_2 (\bar \bsbX \hat {\bar \bsbb}, \bar \bsbX \bar \bsbb^*  ) ]
\lesssim \frac{1}{\delta^2}\Big \{\delta\EE\breg_l(\bar \bsbX   {\bar \bsbb}, \bar \bsbX \bar \bsbb^*  )+\sigma^2 \big[\psi^{-1}\big(\exp\{c( q_\gamma\log\frac{en}{ q_\gamma}+ q_\beta  \log\frac{ep}{ q_\beta})\}\big)\big]^2\Big\}
, \end{equation}
 where $c$ is a positive constant. In particular, in the sub-Gaussian case with $\psi (x) = \exp(x^2) - 1$, the RHS becomes  $\frac{1}{\delta^2} \{\delta  \EE\breg_l(\bar \bsbX   {\bar \bsbb}, \bar \bsbX \bar \bsbb^*  )+ {\sigma^2}      q_\gamma\log({en}/{ q_\gamma})+\sigma^2 q_\beta  \log({ep}/{ q_\beta})+\sigma^2)\}$.
 \end{thm}

\begin{proof}
By definition, for any    $\bar\bsbb: \| \bsbb\|_0\le q_\beta, \| \bsbg\|_0\le q_\gamma$,   $l(\bar\bsbX \hat {\bar \bsbb}) \le l(\bar\bsbX   {\bar \bsbb}) $ and so
\begin{align}
\breg_l(\bar \bsbX\hat {\bar \bsbb}, \bar \bsbX  {\bar \bsbb}^*) \le \breg_l(\bar \bsbX  {\bar \bsbb}, \bar \bsbX  {\bar \bsbb}^*) + \langle \bsb{\epsilon}, \bar \bsbX(\hat {\bar \bsbb} -   {\bar \bsbb})/\| \bar \bsbX(\hat {\bar \bsbb} -   {\bar \bsbb})\|_2)\rangle \|\bar \bsbX(\hat {\bar \bsbb} -   {\bar \bsbb})\|_2.  \label{generalnoisebasicineq}
\end{align}

Define $ \Gamma(q_\beta, q_\gamma)= \{   \bsb{\theta}\in \mathbb R^n: \|\bsb{\theta}\|_2 \le 1, \bsb{\theta} =\bsbX \bsbb + \bsbg \mbox{ for some }   \| \bsbb\|_0\le q_\beta, \| \bsbg\|_0\le q_\gamma  \}$ or $\Gamma$ for short when there is no ambiguity. For notational convenience, we  bound $\sup_{\bsb{\theta} \in \Gamma(q_\beta, q_\gamma)} \langle \bsbeps, \bsb{\theta} \rangle$.
The difficulty lies in   the possible \textit{divergence}  of the entropy integral for  a general $\psi$. To conquer this, we apply discretization and make use of the finiteness of the number of range spaces defined by $\Gamma$.

Recall that   $\bsb{\theta}\in \Gamma(q_\beta, q_\gamma)$ means $\bsb{\theta} = \bar\bsbX \bar\bsbb\in \Proj_{\bar\bsbX_{\mathcal J(\bar \bsbb)}}$. Obviously, if  $\bar \bsbb$ has degenerate zeros: $\|\bsbb\|_0< q_\beta$ and/or $\|\bsbg\|_0< q_\gamma$,  the range of $\bar \bsbX_{_{\mathcal J(\bar \bsbb)}}$ is always  included in a column subspace of  $\bar \bsbX$  indexed by exactly $q_\beta$ columns in $\bsbX$ and $q_\gamma$ columns in $\bsbI_{n\times n}$, and so we just need to focus on the $\bar{\bsbb}$'s with  $J(\bsbb) = q_\beta$ and $J(\bsbg)  = q_\gamma$.    To use  an $\varepsilon$-net $\dot \Gamma$ to discretize $\Gamma$ such that for any $\bsb{\theta}\in \Gamma$, there exists $\dot {\bsb{\theta}}\in \dot \Gamma\subset \Gamma$ satisfying   $\|\bsb{\theta} - \dot {\bsb{\theta}}\|_2 \le \varepsilon$, we include all nondegenerate subspaces and apply  a standard volume argument. Then for any $0<\epsilon<1$, the covering number $\mathcal N(\varepsilon , \Gamma, \|\cdot \|_2)$ is bounded by
$$
\mathcal N(\varepsilon, \Gamma, \|\cdot \|_2) \le {p \choose q_\beta} {n \choose q_\gamma} \Big(\frac{3}{\varepsilon}\Big)^{r(\bsbX)\wedge (q_\beta + q_\gamma)}.
$$
Furthermore, the construction shows that for     $\bsb{\theta}\in \Gamma$ and     $\bsb{\theta} = \bar\bsbX \bar\bsbb$ with no degeneracy in $\bar\bsbb$,    $\dot{\bsb{\theta}}$ is also in the subspace determined by $ \mathcal P:=\Proj_{\bar\bsbX_{\mathcal J(\bar \bsbb)}}$. Since  $\mathcal P \bsba \in \bsbGamma$ for any $\|\bsba\|_2\le 1$  we have
\begin{align*}
\langle \bsbeps, \bsb{\theta} \rangle &= \langle \bsbeps, \bsb{\theta}  - \dot {\bsb{\theta}} \rangle +        \max_{\dot{\bsb{\theta}} \in \dot \Gamma} \langle \bsbeps,   \dot {\bsb{\theta}}\rangle    \\
& =  \varepsilon\langle \bsbeps,  {\mathcal P}(\bsb{\theta}  - \dot {\bsb{\theta}})/\varepsilon \rangle+    \max_{\dot{\bsb{\theta}} \in \dot \Gamma} \langle \bsbeps,   \dot {\bsb{\theta}}\rangle\\
&\le  \varepsilon\sup_{\bsb{\theta}\in \Gamma}\langle \bsbeps, \bsb{\theta}\rangle +    \max_{\dot{\bsb{\theta}} \in \dot \Gamma} \langle \bsbeps,   \dot {\bsb{\theta}}\rangle,
\end{align*}
and so
   $ (1-\varepsilon) \sup_{\bsb{\theta}\in \Gamma} \langle \bsbeps, \bsb{\theta}  \rangle \le  \max_{\dot \Gamma}\langle \bsbeps,   \dot {\bsb{\theta}}\rangle$. We can take say $\varepsilon = 0.5$ to turn to a finite-class problem instead of using the entropy integral.

Indeed, by an extension of Massart's finite class lemma (cf. Lemma 2.2.2 in \cite{van1996weak}) and Stirling's formula, we obtain
$$
\| \sup_{\bsb{\theta}\in \Gamma( 2q_\beta  ,  2q_\gamma )} \langle \bsbeps, \bsb{\theta}  \rangle \|_{\psi} \le C   \psi^{-1}\big(\exp\{c (q_\gamma\log\frac{en}{q_\gamma}+q_\beta \log\frac{ep}{ q_\beta})\}\big)  \sigma
$$
where $ C, c$ are constants depending on $\psi$ only.

Based on \eqref{generalnoisebasicineq} and  the regularity condition,
$$
\frac{\delta}{4}\|\bar \bsbX(\hat {\bar \bsbb} -   {\bar \bsbb}^*)\|_2^2 \le \breg_l(\bar \bsbX\hat {\bar \bsbb}, \bar \bsbX  {\bar \bsbb}^*) -\frac{\delta}{4}\|\bar \bsbX(\hat {\bar \bsbb} -   {\bar \bsbb}^*)\|_2^2\le \breg_l(\bar \bsbX  {\bar \bsbb}, \bar \bsbX  {\bar \bsbb}^*) + \frac{1}{\delta} (\sup_{\bsb{\theta}\in \Gamma(2q_\beta, 2q_\gamma)} \langle \bsbeps, \bsb{\theta}  \rangle   )^2.
$$
By assumption, $\{\EE [ (\sup_{\bsb{\theta}\in \Gamma(2q_\beta, 2q_\gamma)} \langle \bsbeps, \bsb{\theta}  \rangle   )^2]\}^{1/2}\lesssim   \| \sup_{\bsb{\theta}\in \Gamma(2q_\beta, 2q_\gamma)} \langle \bsbeps, \bsb{\theta}  \rangle \|_{\psi}  $. Hence
$$
\EE [\Breg_2 (\bar \bsbX \hat {\bar \bsbb}, \bar \bsbX \bar \bsbb^*  ) ]
\lesssim \frac{1}{\delta}\EE\breg_l(\bar \bsbX   {\bar \bsbb}, \bar \bsbX \bar \bsbb^*  )+\frac{1}{\delta^2} \{[\sigma^2\psi^{-1}(\exp\{c q_{\gamma}\log\frac{en}{q_{\gamma}}+c q_{\beta}  \log\frac{ep}{ q_{\beta}}\})]^2\}.
$$
The proof is complete.
\end{proof}

\begin{remark}
The oracle inequality implies that even if $\bar \bsbb^*$ has  $\|\bm\beta^*\|_0 \gg q_\beta,   \|\bm\gamma^*\|_0 \gg q_\gamma$, as along as  $  \bar \bsbb^*$ can be well approximated by some  $ {\bar \bsbb}$: $\|\bm\beta\|_0 \le q_\beta,   \|\bm\gamma\|_0 \le q_\gamma$ in the sense that the bias
$\EE\breg_l(\bar \bsbX   {\bar \bsbb}, \bar \bsbX \bar \bsbb^*  )$ is controlled by $\sigma^2 \big[\psi^{-1}\big(\exp\{c( q_\gamma\log\frac{en}{ q_\gamma}+ q_\beta  \log\frac{ep}{ q_\beta})\}\big)\big]^2/\delta$ up to a multiplicative constant, the prediction risk  bound is of the order  $\sigma^2/\delta^2$ times
\begin{align}  \Big[\psi^{-1}\big(\exp\big\{c( q_\gamma\log\frac{en}{ q_\gamma}+ q_\beta  \log\frac{ep}{ q_\beta})\big\}\big)\Big]^2.\end{align}
  This applicability to approximately sparse signals is quite useful in  reality.
\end{remark}
\begin{remark}\label{rmk:genbp}
Taking a specific $\bar \bsbb = \bar\bsbb^*$ in \eqref{fixed-point-rate-gennoise} and assuming the associated regularity condition holds,  the error bound    depends on $\bsbg^*$ through its support size only, thereby  \emph{finite} regardless of its magnitude or outlyingness. This  gives a risk-based    breakdown point result  that accounts for the randomness of the estimators.

Concretely, to define a general stochastic breakdown point, we  fix    the true systematic component   $ \bsbX \bsbb $, but can freely alter   the response   $\bsby$   in
  \begin{align}\label{eq:bpyvaryingspace}
  \mathcal Y (o) = \{\bsby: \|\bsbg\|_0 \le o, \| \bsbeps\|_{\psi}<+\infty  \}
 \end{align} with $o\in \mathbb N\cup \{0\}$, where   $\bsbeps = -\nabla l (\bsbX \bsbb  +\bsbg; \bsby)$.
Given any estimator   $(\hat { \bsbb} , \hat \bsbg)$  (that implicitly depends on the data $\bsbX, \bsby$),  its  finite-sample breakdown point $\epsilon^*$ can be defined by
 \begin{align} \frac{1}{n} \min\{o: \sup_{\bsby \in\mathcal Y(o)  }   \EE [\Breg_2 (\bar \bsbX \hat {\bar \bsbb}, \bar \bsbX \bar \bsbb   ) ] = +\infty\}.\label{sbp-1}
\end{align}
Then, in the setup of Theorem \ref{theorem:generalnoise}, the estimator given $q_\beta, q_\gamma$ has the stochastic breakdown point
\begin{align}\epsilon^* \ge (q_\gamma+1)/n. \label{bplowerbound}
\end{align}
For example, for   $  \hat {\bar \bsbb} = \arg\min_{ \bar{\bm\beta}:\|\bm\beta\|_0 \le q_\beta,   \|\bm\gamma\|_0 \le q_\gamma } l(\bar{\bm X} \bar{\bm\beta}; \bm y)$
 with a quadratic  $l(\bsbe; \bsby) = \| \bsby - \bsbe\|_2^2/2$, the contamination model associated with $\mathcal Y (o)$ is   $\bsby = \bsbX \bsbb    + \bsbg + \bsbeps$
subject to  $\|\bsbg\|_0 \le o$,  where the nonzero components of $\bsbg$ can be arbitrarily large, but even   in high dimensions the breakdown point of $ \hat {\bar \bsbb}$ is no lower than $(q_\gamma+1)/n$. The direct link between the $\ell_0$-constraint and  breakdown   point facilitates parameter tuning.
\end{remark}

Statistically speaking, obtaining a precise error rate  $\big[\psi^{-1}\big(\exp\{c( q_\gamma\log\frac{en}{ q_\gamma}+ q_\beta  \log\frac{ep}{ q_\beta})\}\big)\big]^2$   in Theorem \ref{theorem:generalnoise} is much more informative than simply knowing that the risk is finite for the purpose of breakdown studies. However, it is an interesting question to determine     what relaxed conditions    $\bsbeps$ should satisfy to guarantee \eqref{bplowerbound}  for a general \textit{extended} real-valued function $l:  \mathbb R^n \rightarrow \mathbb R\cup\{+\infty\}$.   Toward this, we change the $\Breg_2$ in \eqref{sbp-1} to the generalized Bregman $\breg_l$:  $\epsilon^*=   \min\{o: \sup_{\bsby \in\mathcal Y(o)  }   \EE [\breg_l (\bar \bsbX \hat {\bar \bsbb}, \bar \bsbX \bar \bsbb   ) ] = +\infty\}/n$,
 and make a no-model-ambiguity  assumption: $ l$ is differentiable at    $\bar \bsbX \bar \bsbb^* \in D=\{\bsbeta: l(\bsbeta)<+\infty\} $ with the gradient  $   \nabla l(\bar \bsbX \bar \bsbb^*  )=-\bm\epsilon$,  $\bar \bsbX \bar \bsbb^* $ is a finite optimal solution  to the \textit{Fenchel conjugate}   when $\bsbzeta = - \bsbeps$:
\begin{align}
l^*(\bsbzeta) = \sup_{\bsbe\in \mathbb R^n} \langle \bsbzeta ,\bsbe \rangle - l(\bsbe),\label{f-conj}
\end{align}
and the extended real-valued convex function $l^*$ is differentiable at $-\bsbeps$. The assumption simply means that $(\bar \bsbX \bar \bsbb^* , -\bsbeps)$ makes a conjugate pair. Note that   $l$  need not be overall strictly convex, especially when $D$ is compact, according to Danskin's min-max theorem \citep{Bertsekas1999book}. Then we have the following conclusion.
\begin{thm}\label{theorem:bregBP}
Under $\EE [\breg_{l^*}(\bsbeps, -\bsbeps)]<M<+\infty$ for all $\bsbg^*: \|\bsbg^*\|_0\le q$, the   $\breg_l$-risk based   breakdown point   for    any finite solution   $  \hat {\bar \bsbb}$ to $\min_{ \bar{\bm\beta}:   \|\bm\gamma\|_0 \le q } l(\bar{\bm X} \bar{\bm\beta}; \bm y)$    satisfies $\epsilon^* \ge (q+1)/n$.
\end{thm}
In the special case that  $l$ is strongly convex on $\mathbb R^n$ with respect to a certain norm, as long as the effective noise has the second moment, the theorem recovers    \eqref{bplowerbound}  since   $l^*$ can be shown to be strongly smooth with respect to the Euclidean norm \citep{Rockafellar1970} (albeit not delivering  a concrete  error rate as before). Two other notable features apart from the randomness:  $l$ is   a general extended real-valued function (not restricted to  be a  function of $\bsby - \bsbe$  as in M-estimation), and its conjugate plays an important role in bounding the risk; also, we do not need to assume a positive norm penalty which is a key element to the proof of Proposition 1 in \cite{alfons2013sparse}.

\begin{proof}
First, we show that
\begin{align}
\breg_l(\bar\bsbX  \hat {\bar\bsbb}, \bar\bsbX    {\bar\bsbb^*}) \le \breg_{l^*}(\bsbeps, -\bsbeps). \label{genbregbyFY}
\end{align}
By definition, $l(\bar\bsbX  \hat {\bar\bsbb}) \le l( \bar\bsbX    {\bar\bsbb^*})$, from which it follows that $\breg_l(\bar\bsbX  \hat {\bar\bsbb}, \bar\bsbX    {\bar\bsbb^*}) \le \langle \bsbeps, \bar\bsbX  \hat {\bar\bsbb}- \bar\bsbX    {\bar\bsbb^*} \rangle$. Define  $\bsb{\eta}^* =  \bar\bsbX    {\bar\bsbb^*}$ and
$$
h(\bsbdelta) = \breg_l(\bsbdelta+\bsb{\eta}^* , \bsb{\eta}^*).
$$
By assumption, $l$ is a proper function and applying Fenchel-Young's inequality gives
\begin{align*}
 \breg_l(\bar\bsbX  \hat {\bar\bsbb}, \bar\bsbX    {\bar\bsbb^*}) \le \langle \bsbeps, \bsbdelta \rangle|_{\bsbdelta = \bar\bsbX  \hat {\bar\bsbb}- \bar\bsbX    {\bar\bsbb^*}}\le \frac{1}{c}\breg_l(\bsbdelta+\bar\bsbX    {\bar\bsbb^*} , \bar\bsbX    {\bar\bsbb^*})|_{\bsbdelta = \bar\bsbX  \hat {\bar\bsbb}- \bar\bsbX    {\bar\bsbb^*}}  + \frac{1}{c}h^*(c \bsbeps)
\end{align*}
or $(1-1/c)\breg_l(\bar\bsbX  \hat {\bar\bsbb}, \bar\bsbX    {\bar\bsbb^*})\le h^*(c \bsbeps)/c $
for any $c>0$.

On the other hand,
\begin{align*}
h^*(\bsbzeta) &= \sup_\bsbdelta \langle \bsbzeta, \bsbdelta \rangle  - l(\bsb{\eta}^*+\bsbdelta)
 + l(\bsb{\eta}^*) + \langle \nabla l(\bsbeta^*), \bsbdelta\rangle \\
 &= \sup_\bsbdelta \langle \bsbzeta+\nabla l(\bsbeta^*), \bsbdelta \rangle  - l(\bsb{\eta}^*+\bsbdelta)
 + l(\bsb{\eta}^*) \\
 &= \sup_\bsbdelta \langle \bsbzeta+\nabla l(\bsbeta^*),\bsb{\eta}^*+ \bsbdelta \rangle  - l(\bsb{\eta}^*+\bsbdelta)
 + l(\bsb{\eta}^*)-  \langle \bsbzeta+\nabla l(\bsbeta^*),\bsb{\eta}^* \rangle \\
 &= l^{*} (\bsbzeta+\nabla l(\bsbeta^*))
 + l(\bsb{\eta}^*)- \langle \bsbzeta+\nabla l(\bsbeta^*),\bsb{\eta}^* \rangle,
 \end{align*}
 where $ l(\bsb{\eta}^*),  \nabla l(\bsbeta^*)$ are known to be finite.
 Using the optimality of $\bsbeta^*$, we have further
 \begin{align*}
h^*(\bsbzeta) &= l^{*} (\bsbzeta+\nabla l(\bsbeta^*))
 - l^{*}(-\bsbeps)- \langle \bsbzeta,\bsb{\eta}^*\rangle.
\end{align*}
Moreover, from the assumption and definition \eqref{f-conj}, it is easy to show that $ \bsb{\eta}^*\in \partial l^{*}(-\bsbeps)$, and so   $\nabla l^{*}(-\bsbeps)=\bsb{\eta}^*$, from which it follows that
\begin{align}
h^*(\bsbzeta) = \breg_{l^{*}} (\bsbzeta-\bsbeps, -\bsbeps).
\end{align}
Taking $c=2$  gives \eqref{genbregbyFY}. The rest of the lines  follow Remark \ref{rmk:genbp}.
\end{proof}

\subsection{More simulations} \label{sec:add_exp}
In this part, we present more experiment results by varying the correlation strength, covariance structure and sparsity level.

Concretely, in Example 1 with $n=1000, p=10, o^*=100$, we changed the correlation strength $\rho$ from 0.2 to 0.8. In Example 2 with $n=1000, p=10, o^*=120$, three possible covariance structures for $\bm\Sigma$ were included, Toeplitz structure $[\rho^{|i-j|}]$, equally correlated structure $[\rho \mbox{1}_{i\neq j}]$, and blocked structure $\mbox{diag}\{ \bm\Sigma_0, \bm\Sigma_1\}$, with two blocks of equal size and the predictors within each block equally correlated with  $\rho=0.5$. In Example 3 and Example 4, where $n = 200, p = 1000, o^*=30$,  we considered different sparsity levels. The later steps of introducing high-leveraged outliers remain the same. The results are summarized in Table \ref{table:new1}, Table \ref{table:new2}, and Table \ref{table:new3}.
 The overall comparison evidently supports PIQ as an extremely resistant and efficient method.

\begin{table}[!ht]
  \setlength{\tabcolsep}{3pt}\centering\footnotesize

{\footnotesize 
  \caption{\small Example 1 with different correlation strengths. \label{table:new1}
}

\vspace{.1in}

  \begin{tabular}{l cccc c cccc c cccc c cccc}
  \hline
  & \multicolumn{4}{c}{$\rho=0.2$} && \multicolumn{4}{c}{$\rho=0.4$} && \multicolumn{4}{c}{$\rho=0.6$} && \multicolumn{4}{c}{$\rho=0.8$}\\
  \cmidrule(lr){2-5} \cmidrule(lr){7-10} \cmidrule(lr){12-15} \cmidrule(lr){17-20}
              & Err & M & JD & \textbf{T} && Err & M & JD & \textbf{T} && Err & M & JD &  \textbf{T} && Err & M & JD & \textbf{T}\\
  \hline
  S            & 0.17 & 49  & 38 & 0.1 &  & 0.06 & 7.3 & 76 & 0.1 &  &  0.08 & 0.3 & 80 & 0.1 &  & 0.12 & 0.3 & 80&0.1\\
  LTS          & 0.11 & 36  & 44 & 0.2 &  & 0.04 & 9.1 & 64 & 0.2 &  &  0.03 & 0.5 & 78 & 0.2 &  & 0.06 & 0.2 & 82&0.2\\
RLARS  & 0.23 & 84  &  0 & 0.9 &  & 0.22 & 81  &  0 & 0.9 &  &  0.23 & 74  & 0  & 0.7 &  & 0.29 & 59  & 0 &0.7 \\
  PENSE        & 0.25 & 85  &  0 & 13  &  & 0.24 & 82  &  0 & 13  &  &  0.21 & 62  & 16 & 13  &  & 0.11 & 0.2 & 88&13 \\
  \textbf{PIQ} & 0.02 & 0.20& 90 & 0.1 &  & 0.02 & 0.2 & 90 & 0.1 &  &  0.02 & 0.2 & 90 & 0.1 &  & 0.06 & 0.2 & 88&0.1\\
  \hline
  \end{tabular}
}
\end{table}

\begin{table}[!ht]
  \setlength{\tabcolsep}{3pt}\centering\footnotesize

{\footnotesize 
  \caption{\small Example 2 with different covariance structures. \label{table:new2}
}

\vspace{.1in}

  \begin{tabular}{l cccc c cccc c cccc}
  \hline
  & \multicolumn{4}{c}{Toeplitz structure} && \multicolumn{4}{c}{Blocked structure} &&    \multicolumn{4}{c}{Equally correlated}\\
  \cmidrule(lr){2-5} \cmidrule(lr){7-10} \cmidrule(lr){12-15}
             & Err &M & JD & \textbf{T} && Err & M & JD & \textbf{T}  && Err & M & JD & \textbf{T}\\
  \hline
  B-Y          & 0.29 & 100 & 0   &0.6  && 0.34 &  97 & 0   &0.8  && 0.34 & 97 & 0  & 0.7 \\
  QLE          & 0.29 & 100 & 0   &0.1  && 0.33 &  97 & 0   &0.1  && 0.33 & 97 & 0  & 0.1 \\
  TLE          & 0.11 &  16 & 84  &6.0  && 0.07 &  0  & 100 &7.0  && 0.07 & 0  & 100& 6.0 \\
  \textbf{PIQ} & 0.07 &  0  & 100 &0.1  && 0.06 &  0  & 100 &0.1  && 0.06 & 0  & 100& 0.1 \\

  \hline
  \end{tabular}
}
\end{table}

\begin{table}[!ht]
  \setlength{\tabcolsep}{3pt}\centering \footnotesize

{\footnotesize 
\caption{\small Example 3 and Example 4 with different sparsity levels.} \label{table:new3}
\vspace{.1in}

  \begin{tabular}{l ccccccc c ccccccc}
  \hline
   & \multicolumn{15}{c}{\textbf{Regression}}\\

  & \multicolumn{7}{c}{$s^*=2$}  & \multicolumn{7}{c}{$s^*=4$}   \\ \cmidrule(l){2-8} \cmidrule(l){10-16}
    & Err  & M$^{\gamma}$ & JD$^{\gamma}$ & M$^{\beta}$ & FA$^{\beta}$ & JD$^{\beta}$ & \textbf{T} & & Err  & M$^{\gamma}$ & JD$^{\gamma}$ & M$^{\beta}$ & FA$^{\beta}$ & JD$^{\beta}$ & \textbf{T}\\
  \hline
  QL          & 0.70 & 79 & 0  & 4  & 0.2  & 92 & 44   && 0.89 &  81 & 0  &  33 & 0.2  &  8  & 45  \\
  S-MTE       & 0.58 & 77 & 0  & 0  & 0.2  & 100& 35   && 0.82 &  79 & 0  &  28 & 0.2  & 26  &42 \\
 RLARS    & 1.24 & 88 & 0  & 50 & 0.1  & 0  & 40   && 1.52 &  91 & 0  &  58 & 0.3  &  0  & 41 \\

  S-LTS       & 0.93 & 62 & 18 & 32 & 1.6  & 36 & 134  && 1.24 &  79 & 10 &  29 & 1.4  &  16 &  132\\
  \textbf{PIQ}& 0.15 & 2  & 76 & 3  & 0.1  & 94 & 4    && 0.34 & 3 &  78 & 16 & 0.3 &  46 & 4  \\
  \midrule[1pt]
  & \multicolumn{7}{c}{$s^*=6$}  & \multicolumn{7}{c}{$s^*=8$}   \\ \cmidrule(l){2-8} \cmidrule(l){10-16}
      & Err  & M$^{\gamma}$ & JD$^{\gamma}$ & M$^{\beta}$ & FA$^{\beta}$ & JD$^{\beta}$ & \textbf{T} & & Err  & M$^{\gamma}$ & JD$^{\gamma}$ & M$^{\beta}$ & FA$^{\beta}$ & JD$^{\beta}$ & \textbf{T}\\
  \hline
  QL          & 1.14 &  81 & 0 & 36 & 0.2 & 0  & 51    && 1.59 & 86 & 0 & 42 & 0.1 &  0 & 52   \\
  S-MTE       & 0.98 &  79 & 0 & 29 & 0.2 & 4  & 49    && 1.34 & 83 & 0 & 31 & 0.2 &  0 & 60  \\
RLARS      & 1.21 &  85 & 0 & 40 & 0.4 & 0  & 41    && 1.27 & 83 & 0 & 32 & 0.6 &  2 & 41  \\
  S-LTS       & 1.29 &  77 & 22& 25 & 1.0 & 20 & 134   && 1.55 & 90 & 4 & 23 & 1.0 &  8 & 144  \\
  \textbf{PIQ}& 0.64 &  21 & 50 & 20 & 0.4 &18 & 4     && 1.25 & 33 & 34& 31 & 0.7 & 0 & 4 \\
  \hline\hline

  & \multicolumn{15}{c}{\textbf{Classification}}\\

  & \multicolumn{7}{c}{$s^*=2$}  & \multicolumn{7}{c}{$s^*=4$}   \\ \cmidrule(l){2-8} \cmidrule(l){10-16}
    & Err  & M$^{\gamma}$ & JD$^{\gamma}$ & M$^{\beta}$ & FA$^{\beta}$ & JD$^{\beta}$ & \textbf{T} & & Err  & M$^{\gamma}$ & JD$^{\gamma}$ & M$^{\beta}$ & FA$^{\beta}$ & JD$^{\beta}$ & \textbf{T}\\
  \hline
  enetLTS     & 41 &100 &0  & 36  & 1.0 &28 & 65 & &43 & 96& 4 & 62 &1.0 &4 & 72  \\
  \textbf{PIQ}& 15 &4   &96 & 2   & 0.1 &96 & 14 & &18 & 5 & 92& 15 &0.3 &48& 15  \\
  \midrule[1pt]
  & \multicolumn{7}{c}{$s^*=6$}  & \multicolumn{7}{c}{$s^*=8$}   \\ \cmidrule(l){2-8} \cmidrule(l){10-16}

      & Err  & M$^{\gamma}$ & JD$^{\gamma}$ & M$^{\beta}$ & FA$^{\beta}$ & JD$^{\beta}$ & \textbf{T} & & Err  & M$^{\gamma}$ & JD$^{\gamma}$ & M$^{\beta}$ & FA$^{\beta}$ & JD$^{\beta}$ & \textbf{T}\\
  \hline
  enetLTS      & 46 &96 & 4 & 82 & 0.8 & 14& 74 & & 47 & 96 & 4 & 88 &0.6 &0 &76  \\
  \textbf{PIQ} & 25 &65 & 20& 33 & 0.5 & 6 & 15 & & 26 & 93 & 4 & 39 &0.7 &0 &15\\
  \hline

  \end{tabular}
}
\end{table}

\bibliographystyle{apalike}
\bibliography{PIQrefs}

\begin{thebibliography}{}

\bibitem[Alfons et~al., 2013]{alfons2013sparse}
Alfons, A., Croux, C., and Gelper, S. (2013).
\newblock Sparse least trimmed squares regression for analyzing
  high-dimensional large data sets.
\newblock {\em The Annals of Applied Statistics}, 7(1):226--248.

\bibitem[Avella-Medina, 2017]{avella2017influence}
Avella-Medina, M. (2017).
\newblock Influence functions for penalized m-estimators.
\newblock {\em Bernoulli}, 23(4B):3178--3196.

\bibitem[Avella-Medina and Ronchetti, 2018]{AMR18}
Avella-Medina, M. and Ronchetti, E. (2018).
\newblock {Robust and consistent variable selection in high-dimensional
  generalized linear models}.
\newblock {\em Biometrika}, 105(1):31--44.

\bibitem[Belloni and Chernozhukov, 2011]{Belloni2011}
Belloni, A. and Chernozhukov, V. (2011).
\newblock $\ell_1$-penalized quantile regression in high-dimensional sparse
  models.
\newblock {\em The Annals of Statistics}, 39(1):82--130.

\bibitem[Bertsekas, 1999]{Bertsekas1999book}
Bertsekas, D.~P. (1999).
\newblock {\em {Nonlinear Programming}}.
\newblock Athena Scientific, 2nd edition.

\bibitem[Bianco and Yohai, 1996]{bianco1996robust}
Bianco, A.~M. and Yohai, V.~J. (1996).
\newblock Robust estimation in the logistic regression model.
\newblock In {\em Robust statistics, data analysis, and computer intensive
  methods}, pages 17--34. Springer, New York, NY.

\bibitem[Boyd and Vandenberghe, 2004]{Boyd2004}
Boyd, S. and Vandenberghe, L. (2004).
\newblock {\em Convex Optimization}.
\newblock Cambridge University Press, Cambridge.

\bibitem[Bregman, 1967]{Bregman1967}
Bregman, L.~M. (1967).
\newblock {The relaxation method of finding the common point of convex sets and
  its application to the solution of problems in convex programming}.
\newblock {\em USSR Computational Mathematics and Mathematical Physics},
  7:200--217.

\bibitem[B{\"u}hlmann and Yu, 2003]{buhlmann2003boosting}
B{\"u}hlmann, P. and Yu, B. (2003).
\newblock Boosting with the {$L_2$} loss: regression and classification.
\newblock {\em Journal of the American Statistical Association},
  98(462):324--339.

\bibitem[Candes and Tao, 2005]{Candes2005}
Candes, E.~J. and Tao, T. (2005).
\newblock Decoding by linear programming.
\newblock {\em IEEE Trans. Inf. Theor.}, 51(12):4203--4215.

\bibitem[Cantoni and Ronchetti, 2001]{robustGLM2001}
Cantoni, E. and Ronchetti, E. (2001).
\newblock Robust inference for generalized linear models.
\newblock {\em Journal of the American Statistical Association}, 96:1022--1030.

\bibitem[Chapelle, 2007]{Chapelle2007}
Chapelle, O. (2007).
\newblock Training a support vector machine in the primal.
\newblock {\em Neural Computation}, 19(5):1155--1178.

\bibitem[Chinchor, 1992]{FScore}
Chinchor, N. (1992).
\newblock {MUC-4 evaluation metrics}.
\newblock In {\em Proceedings of the 4th Conference on Message Understanding},
  pages 22--29.

\bibitem[Croux and Haesbroeck, 2003]{croux2003implementing}
Croux, C. and Haesbroeck, G. (2003).
\newblock Implementing the {B}ianco and {Y}ohai estimator for logistic
  regression.
\newblock {\em Computational statistics \& data analysis}, 44(1):273--295.

\bibitem[Freue et~al., 2019]{freue2019robust}
Freue, G. V.~C., Kepplinger, D., Salibi{\'a}n-Barrera, M., and Smucler, E.
  (2019).
\newblock Robust elastic net estimators for variable selection and
  identification of proteomic biomarkers.
\newblock {\em The Annals of Applied Statistics}, 13(4):2065--2090.

\bibitem[Hadi and Luce{\~n}o, 1997]{hadi1997maximum}
Hadi, A.~S. and Luce{\~n}o, A. (1997).
\newblock Maximum trimmed likelihood estimators: a unified approach, examples,
  and algorithms.
\newblock {\em Computational Statistics \& Data Analysis}, 25(3):251--272.

\bibitem[Hampel et~al., 2011]{hampel2011robust}
Hampel, F.~R., Ronchetti, E.~M., Rousseeuw, P.~J., and Stahel, W.~A. (2011).
\newblock {\em Robust Statistics: The Approach based on Influence Functions},
  volume~1.
\newblock John Wiley \& Sons, Hoboken, NJ.

\bibitem[Hastie et~al., 2015]{hastie2015statistical}
Hastie, T., Tibshirani, R., and Wainwright, M. (2015).
\newblock {\em Statistical {L}earning with Sparsity: the Lasso and
  Generalizations}.
\newblock Chapman and Hall/CRC.

\bibitem[Huber and Ronchetti, 2009]{huber2009robust}
Huber, P. and Ronchetti, E. (2009).
\newblock {\em Robust Statistics}.
\newblock John Wiley \& Sons, Hoboken, NJ.

\bibitem[Hunter and Lange, 2004]{Hunter2004tutorial}
Hunter, D.~R. and Lange, K. (2004).
\newblock {A tutorial on MM algorithms}.
\newblock {\em The American Statistician}, 58(1):30--37.

\bibitem[Khan et~al., 2007]{khan2007robust}
Khan, J.~A., Van~Aelst, S., and Zamar, R.~H. (2007).
\newblock Robust linear model selection based on least angle regression.
\newblock {\em Journal of the American Statistical Association},
  102(480):1289--1299.

\bibitem[Kurnaz et~al., 2018]{Kurnaz2018}
Kurnaz, F.~S., Hoffmann, I., and Filzmoser, P. (2018).
\newblock Robust and sparse estimation methods for high-dimensional linear and
  logistic regression.
\newblock {\em Chemometrics and Intelligent Laboratory Systems}, 172:211--222.

\bibitem[Little et~al., 2007]{little2007exploiting}
Little, M.~A., McSharry, P.~E., Roberts, S.~J., Costello, D.~A., and Moroz,
  I.~M. (2007).
\newblock Exploiting nonlinear recurrence and fractal scaling properties for
  voice disorder detection.
\newblock {\em Biomedical engineering online}, 6(1):23.

\bibitem[Loh, 2017]{loh2017statistical}
Loh, P.-L. (2017).
\newblock Statistical consistency and asymptotic normality for high-dimensional
  robust $ m $-estimators.
\newblock {\em The Annals of Statistics}, 45(2):866--896.

\bibitem[Maronna et~al., 2006]{maronna2006robust}
Maronna, R.~A., Martin, D.~R., and Yohai, V.~J. (2006).
\newblock {\em Robust Statistics: Theory and Methods}, volume~1.
\newblock John Wiley \& Sons.

\bibitem[M{\"u}ller and Welsh, 2005]{muller2005outlier}
M{\"u}ller, S. and Welsh, A. (2005).
\newblock Outlier robust model selection in linear regression.
\newblock {\em Journal of the American Statistical Association},
  100(472):1297--1310.

\bibitem[Nair and Hinton, 2010]{Hinton2010relu}
Nair, V. and Hinton, G.~E. (2010).
\newblock Rectified linear units improve restricted {B}oltzmann machines.
\newblock In {\em Proceedings of the 27th international conference on machine
  learning (ICML-10)}, pages 807--814.

\bibitem[Needell and Tropp, 2009]{needell2009cosamp}
Needell, D. and Tropp, J.~A. (2009).
\newblock {CoSaMP}: Iterative signal recovery from incomplete and inaccurate
  samples.
\newblock {\em Applied and computational harmonic analysis}, 26(3):301--321.

\bibitem[{\"O}llerer et~al., 2015]{ollerer2015influence}
{\"O}llerer, V., Croux, C., and Alfons, A. (2015).
\newblock The influence function of penalized regression estimators.
\newblock {\em Statistics}, 49(4):741--765.

\bibitem[Pregibon, 1982]{pregibon1982resistant}
Pregibon, D. (1982).
\newblock Resistant fits for some commonly used logistic models with medical
  applications.
\newblock {\em Biometrics}, pages 485--498.

\bibitem[Qin et~al., 2017]{Qin2017}
Qin, Y., Li, S., Li, Y., and Yu, Y. (2017).
\newblock Penalized maximum tangent likelihood estimation and robust variable
  selection.
\newblock {\em arXiv:1708.05439}.

\bibitem[Rigollet and Tsybakov, 2011]{Rigollet11}
Rigollet, P. and Tsybakov, A. (2011).
\newblock Exponential screening and optimal rates of sparse estimation.
\newblock {\em Annals of Statistics}, 39(2):731--771.

\bibitem[Rockafellar, 1970]{Rockafellar1970}
Rockafellar, R.~T. (1970).
\newblock {\em {Convex Analysis}}.
\newblock Princeton University Press, Princeton, NJ.

\bibitem[Ronchetti et~al., 1997]{ronchetti1997robust}
Ronchetti, E., Field, C., and Blanchard, W. (1997).
\newblock Robust linear model selection by cross-validation.
\newblock {\em Journal of the American Statistical Association},
  92(439):1017--1023.

\bibitem[Rousseeuw and Yohai, 1984]{rousseeuw1984robust}
Rousseeuw, P. and Yohai, V. (1984).
\newblock {Robust regression by means of S-estimators}.
\newblock In {\em Robust and nonlinear time series analysis}, pages 256--272.
  Springer, New York, NY.

\bibitem[Rousseeuw, 1985]{rousseeuw1985multivariate}
Rousseeuw, P.~J. (1985).
\newblock Multivariate estimation with high breakdown point.
\newblock {\em Mathematical statistics and applications}, 8:283--297.

\bibitem[Rousseeuw and Driessen, 1999]{rousseeuw1999computing}
Rousseeuw, P.~J. and Driessen, K.~V. (1999).
\newblock {Computing LTS Regression for Large Data Sets}.
\newblock Technical report, Institute of Mathematical Statistics Bulletin.

\bibitem[Rousseeuw and Hubert, 1997]{rousseeuw1997recent}
Rousseeuw, P.~J. and Hubert, M. (1997).
\newblock Recent developments in {PROGRESS}.
\newblock {\em Lecture Notes-Monograph Series}, pages 201--214.

\bibitem[Rudelson and Vershynin, 2013]{rudelson2013}
Rudelson, M. and Vershynin, R. (2013).
\newblock Hanson-{W}right inequality and sub-gaussian concentration.
\newblock {\em Electron. Commun. Probab.}, 18:1--9.

\bibitem[Sakar et~al., 2013]{sakar2013collection}
Sakar, B.~E., Isenkul, M.~E., Sakar, C.~O., Sertbas, A., Gurgen, F., Delil, S.,
  Apaydin, H., and Kursun, O. (2013).
\newblock Collection and analysis of a parkinson speech dataset with multiple
  types of sound recordings.
\newblock {\em IEEE Journal of Biomedical and Health Informatics},
  17(4):828--834.

\bibitem[Salibian-Barrera and Van~Aelst, 2008]{salibian2008robust}
Salibian-Barrera, M. and Van~Aelst, S. (2008).
\newblock Robust model selection using fast and robust bootstrap.
\newblock {\em Computational Statistics \& Data Analysis}, 52(12):5121--5135.

\bibitem[She, 2012]{She2012}
She, Y. (2012).
\newblock An iterative algorithm for fitting nonconvex penalized generalized
  linear models with grouped predictors.
\newblock {\em Computational Statistics \& Data Analysis}, 56(10):2976--2990.

\bibitem[She, 2016]{She2016}
She, Y. (2016).
\newblock {On the finite-sample analysis of $\Theta$-estimators}.
\newblock {\em Electronic Journal of Statistics}, 10(2):1874--1895.

\bibitem[She and Chen, 2017]{She2017RRRR}
She, Y. and Chen, K. (2017).
\newblock Robust reduced-rank regression.
\newblock {\em Biometrika}, 104(3):633--647.

\bibitem[She and Owen, 2011]{she2011outlier}
She, Y. and Owen, A.~B. (2011).
\newblock Outlier detection using nonconvex penalized regression.
\newblock {\em Journal of the American Statistical Association},
  106(494):626--639.

\bibitem[She and Tran, 2019]{SCV}
She, Y. and Tran, H. (2019).
\newblock On cross-validation for sparse reduced rank regression.
\newblock {\em Journal of the Royal Statistical Society: Series B},
  81:145--161.

\bibitem[She et~al., 2013]{She2013Super}
She, Y., Wang, J., Li, H., and Wu, D. (2013).
\newblock Group iterative spectrum thresholding for super-resolution sparse
  spectral selection.
\newblock {\em IEEE Transactions on Signal Processing}, 61(24):6371--6386.

\bibitem[She et~al., 2021]{SheBregman}
She, Y., Wang, Z., and Jin, J. (2021).
\newblock Analysis of generalized {B}regman surrogate algorithms for nonsmooth
  nonconvex statistical learning.
\newblock {\em The Annals of Statistics}, 49(6):3434--3459.

\bibitem[Tibshirani, 1996]{tibshirani1996regression}
Tibshirani, R. (1996).
\newblock Regression shrinkage and selection via the lasso.
\newblock {\em Journal of the Royal Statistical Society. Series B
  (Methodological)}, pages 267--288.

\bibitem[Tsybakov, 2008]{Tsybakov2008}
Tsybakov, A.~B. (2008).
\newblock {\em Introduction to Nonparametric Estimation}.
\newblock Springer, New York, NY.

\bibitem[van~de Geer and B{\"u}hlmann, 2009]{van2009conditions}
van~de Geer, S.~A. and B{\"u}hlmann, P. (2009).
\newblock {On the conditions used to prove oracle results for the Lasso}.
\newblock {\em Electronic Journal of Statistics}, 3:1360--1392.

\bibitem[van~der Vaart and Wellner, 1996]{van1996weak}
van~der Vaart, A. and Wellner, J. (1996).
\newblock {\em Weak Convergence and Empirical Processes: With Applications to
  Statistics}.
\newblock Springer Series in Statistics. Springer.

\bibitem[van~der Vaart, 1998]{vdV98}
van~der Vaart, A.~W. (1998).
\newblock {\em Asymptotic {S}tatistics}.
\newblock Cambridge University Press, Cambridge.

\bibitem[Vandev and Neykov, 1998]{vandev1998regression}
Vandev, D. and Neykov, N. (1998).
\newblock About regression estimators with high breakdown point.
\newblock {\em Statistics: A Journal of Theoretical and Applied Statistics},
  32(2):111--129.

\bibitem[Vershynin, 2012]{Vershynin2012}
Vershynin, R. (2012).
\newblock {Introduction to the non-asymptotic analysis of random matrices}.
\newblock {\em Compressed sensing}, pages 210--268.

\bibitem[Zhang, 2013]{zhang2013multi}
Zhang, T. (2013).
\newblock Multi-stage convex relaxation for feature selection.
\newblock {\em Bernoulli}, 19(5B):2277--2293.

\end{thebibliography}
\end{document}